\documentclass[10pt]{amsart}

\usepackage{amsmath}
\usepackage{amsthm}
\usepackage{amssymb}
\usepackage{amsfonts}
\usepackage{color}
\usepackage[colorlinks=true, citecolor=black, filecolor=black, linkcolor=black, urlcolor=black]{hyperref}
\usepackage{verbatim} 

\def\leg{\Lambda} 
\def\bp{{\bar\partial}}

\def\bfs{\boldsymbol}
\def\ms{\medskip}
\def\no{\noindent}
\def\pa{\partial}
\def\sm{\setminus}
\def\ss{\smallskip}
\def\ti{\tilde}
\def\wh{\widehat}
\def\wt{\widetilde}
\def\ve{\varepsilon}

\def\eff{\mathrm{eff}}
\def\hol{\mathrm{hol}}
\def\Aut{\mathrm{Aut}}
\def\OPE{\mathrm{OPE}}
\def\PPS{\mathrm{PPS}}
\def\PS{\mathrm{PS}}
\def\NC{\mathrm{NC}}
\def\id{\mathrm{id}}
\def\SLE{\mathrm{SLE}}
\def\Re{\mathrm{Re}}
\def\Im{\mathrm{Im}}
\def\supp{\mathrm{supp}}

\DeclareMathOperator{\Reg}{Reg}
\DeclareMathOperator{\Res}{Res}
\DeclareMathOperator{\Sing}{Sing}

\def\dd{\mathrm{d}}
\def\ee{\mathrm{e}}
\def\EE{\mathcal{E}}
\def\FF{\mathcal{F}}
\def\LL{\mathcal{L}}
\def\OO{\mathcal{O}}
\def\PP{\mathcal{P}}
\def\VV{\mathcal{V}}
\def\XX{\mathcal{X}}
\def\YY{\mathcal{Y}}

\def\C{\mathbb{C}}

\def\D{\mathbb{D}}
\def\E{\mathbf{E}}
\def\H{\mathbb{H}}
\def\P{\mathbf{P}}
\def\R{\mathbb{R}}

\newcommand\subsubsec[1]{\medskip\noindent\textbullet~\emph{#1.}~}

\newcommand\arXiv[1]{\href{http://arxiv.org/abs/#1}{arXiv:#1}}

\newcommand{\overbar}[1]{\mkern 1.5mu\overline{\mkern-1.5mu#1\mkern-1.5mu}\mkern 1.5mu}

\makeatletter
\def\@MRnumber{}
\def\@scanforMR#1 #2\endscan{%
 \def\@MRnumber{#1}%
 }
\def\reviewed#1{\ifx#1\stop \let\next=\relax \else \ifx#1(\advance\count255 by1 \else\let\next=\relax \fi \let\next=\reviewed \fi \next}
\def\MR#1{\relax
 \ifhmode\unskip\spacefactor3000 \space\fi
 \count255=0 \reviewed#1\stop
 \ifnum \count255>0
 {\@scanforMR#1\endscan\href{http://www.ams.org/mathscinet-getitem?mr=\@MRnumber}{MR#1}}
 \else
 {\href{http://www.ams.org/mathscinet-getitem?mr=#1}{MR#1}}
 \fi}
\makeatother

\theoremstyle{plain}
\numberwithin{equation}{section}
\newtheorem{thm}{Theorem}[section]
\newtheorem{lem}[thm]{Lemma}
\newtheorem{cor}[thm]{Corollary}
\newtheorem{prop}[thm]{Proposition}
\newtheorem*{claim*}{Claim}

\theoremstyle{definition}
\newtheorem*{eg*}{Example}\newtheorem*{egs*}{Examples}
\newtheorem*{q*}{Question}
\newtheorem*{def*}{Definition}

\theoremstyle{remark}
\newtheorem*{rmk*}{Remark}
\newtheorem*{rmks*}{Remarks}

\usepackage{etoolbox}
\makeatletter
\pretocmd{\section}{\addtocontents{toc}{\protect\addvspace{5\p@}}}{}{}
\makeatother

\begin{document}

\title{Conformal field theory on the Riemann sphere and its boundary version for SLE}

\author{Nam-Gyu Kang} 
\address{School of Mathematics, Korea Institute for Advanced Study, \newline Seoul, 02455, Republic of Korea} 
\email{namgyu@kias.re.kr} 
\thanks{Nam-Gyu Kang was partially supported by Samsung Science and Technology Foundation (SSTF-BA1401-51), a KIAS Individual Grant (MG058103) at Korea Institute for Advanced Study, and the National Research Foundation of Korea (NRF-2019R1A5A1028324).} 

\author{Nikolai~G.~Makarov}
\address{Department of Mathematics, California Institute of Technology, \newline Pasadena, CA 91125, USA} 
\email{makarov@caltech.edu}
\thanks{Nikolai~G.~Makarov was supported by NSF grant no. 1500821.}

\keywords{conformal field theory, Gaussian free field, martingale-observables, SLE}
\subjclass[2010]{Primary 60J67, 81T40; Secondary 30C35}

\date{}

\begin{abstract}
From conformal field theory on the Riemann sphere, we implement its boundary version in a simply-connected domain using the Schottky double construction. 
We consider the statistical fields generated by background charge modification of the Gaussian free field with Dirichlet boundary condition under the OPE multiplications. We prove that the correlation functions of such fields with symmetric background charges form a collection of martingale-observables for (forward) chordal/radial SLE with force points and spins.
We also present the connection between conformal field theory with Neumann boundary condition and the theory of backward SLE.
\end{abstract}

\maketitle

\tableofcontents

\section{Introduction and results} \label{sec: intro}
Since Belavin, Polyakov, and Zamolodchikov \cite{BPZ84} introduced an operator algebra formalism to relate some critical models to the representation theory of Virasoro algebra, conformal field theory has been applied in string theory and condensed matter physics. 
In mathematics, conformal field theory inspired the development of algebraic theories, e.g., the theory of vertex algebras. 
Also, it has been used to derive several exact results for the conformally invariant critical clusters in two-dimensional lattice models. 

During the last two decades, some outstanding predictions from statistical physics have been proved using Schramm-Loewner evolution (SLE). 
For instance, the remarkable achievements in this area include the work of Lawler-Schramm-Werner on Mandelbrot conjecture (see \cite{LSW01a} and references therein), Smirnov's work on percolation model and Ising model at criticality (see \cite{Smirnov01, Smirnov10}), and the work of Schramm-Sheffield on contour lines of the two-dimensional discrete Gaussian free field (see \cite{SS09}). 
Under quite general conditions, it requires just one martingale-observable, which determines the law uniquely to prove the scaling limit convergence of interface curves in lattice models.

In this paper we use the method of conformal field theory to study a certain family of chordal/radial SLE martingale-observables.
To be precise, we develop a version of conformal field theory on the Riemann sphere and apply the Schottky double construction to this theory to implement its boundary version in a simply-connected domain $D$ with marked boundary/interior points. 
This theory is based on non-random (background charge) modification of the Gaussian free field $\Phi$ in $D$ with Dirichlet boundary condition. 
The theory with Neumann boundary condition is presented in Section~\ref{sec: Neumann CFT}. 
We state Ward's equations in terms of Lie derivatives for the fields in the OPE family $\FF_{\bfs\beta}$ generated by the modified Gaussian free field $\Phi_{\bfs\beta}$ with a (symmetric) background charge $\bfs\beta $ placed at the marked points. 
Using the BPZ-Cardy equations for the fields in $\FF_{\bfs\beta}$ for a symmetric background charge $\bfs\beta$ (so that the associated partition function $Z_{\bfs\beta}$ is real-valued) 
with a specific charge at the growth point $p$ of SLE, we show that the correlation functions of the fields in $\FF_{\bfs\beta}$ form a collection of martingale-observables for chordal/radial $\SLE[\bfs\beta]$ associated with the partition function $Z_{\bfs\beta}.$ 

The partition function $Z_{\bfs\beta}$ we consider is defined in terms of the Coulomb gas correlation function $C[\bfs\beta]$ or the correlation function of the effective one-leg operator.  
In the physics literature, the partition function in conformal field theory is expressed as the correlation function of the boundary condition changing operator, see (\cite[\S 11.3.1]{DFMS97}). 
The other approaches to define the partition function in the context of SLE (up to a multiplicative constant depending possibly on $D$) can be found in \cite{Dubedat09, Lawler09}.
After introducing the partition functions of free field and the partition functions of SLE, Dub\'edat showed that they coincide and use these identities to prove local/global couplings of SLEs and free fields in \cite{Dubedat09}.
The SLE partition function can be viewed as the total mass of natural measure on SLE curves and used to describe the restriction property, e.g., see \cite{Lawler09}.

It is a well-known statement in the physics literature (e.g., see \cite{Cardy04,CD07}) that under the insertion of the normalized one-leg operator $\leg_p/(\E\,\leg_p)$ rooted at $q\in D\cup \pa D$ all correlation functions of fields in the OPE family $\FF_{\check{\bfs\beta\,}}$ are martingale-observables for $\SLE(\kappa)$ in $D$ from $p\in\pa D$ to $q.$
(Here $\E\,\leg_p$ is the correlation function of $\leg_p$ and a specific background charge $\check{\bfs\beta\,}$ is placed at $q.$)
We construct the one-leg ``operator" $\leg_p$ rooted at $q$ as one of the OPE exponentials of the linear combination of bosonic fields with proper charges at $p,q$ so that it is a primary field with the desired conformal dimensions at $p$ and $q$ producing a level two degenerate vector at $p.$ 
Considering general non-random pre-pre-Schwarzian modifications of the Gaussian free field, we extend this result to the chordal/radial SLE with forces and spins. 

In constructing such OPE exponentials, it is essential to require the so-called neutrality conditions on the charges. 
We clarify and explain two neutrality conditions: one condition for the linear combination of bosonic fields to be a well-defined (Fock space) field; the other condition for the (Coulomb gas) correlation function to be conformally invariant. 
To reconcile these two neutrality conditions, we place the ``background charge(s)" $\bfs\beta$ at the marked point(s). 
We also derive the neutrality condition on the background charges using an extension of the Gauss-Bonnet theorem to the flat metric with conical singularities at the marked points. 
It appears that our approach can be extended to much more general settings like general Riemann surfaces with marked points (e.g., see \cite{BKT} for this approach in a doubly-connected domain) and various patterns of insertion (e.g., $N$-leg operators with the method of screening, see \cite{AKM}).
However, we do not claim that what we develop here is the only relevant conformal field theory to SLE theory. 
For example, as another relevant theory, a twisted conformal field theory can be related to radial SLE. 
Its bosonic observable with central charge $c=1$ corresponds to a discrete observable of the harmonic explorer with $\mathbb{Z}_2$-symmetry.

\subsection{A spherical conformal field theory and its boundary version} \label{ss: Radial CFT} 
The fields we treat in this paper are statistical fields constructed from the Gaussian free field $\Phi$ with Dirichlet boundary condition via differentiation and Wick's formula. 
In the last section we develop a version of conformal field theory constructed from the Gaussian free field $N$ with Neumann boundary condition.
The \emph{Gaussian free field} $\Phi$ in a simply-connected domain $D$ can be viewed as a Fock space field and its 2-point correlation is given by
$$\E\,\Phi(z_1)\Phi(z_2) = 2 G(z_1,z_2),$$
where $G$ is the Dirichlet Green's function for $D.$ 
A \emph{Fock space field} is a linear combination of \emph{basic fields}, which are, by definition, formal expressions written as Wick's product ($\odot$-product) of derivatives of $\Phi$ and Wick's exponentials $\ee^{\odot\alpha\Phi}\,(\alpha\in\C)$ of $\Phi.$ 
If $X_1,\cdots,X_n$ are Fock space fields and $z_1,\cdots,z_n$ are distinct points in $D,$ then
a \emph{correlation function}
$\E[X_1(z_1)\cdots X_n(z_n)]$
can be defined by Wick's calculus.
See Section~\ref{sec: CFT of GFF} or \cite[Lecture~1]{KM13} for more details.
(A more traditional notation for the correlation function is $\langle X_1(z_1)\cdots X_n(z_n)\rangle.$) 
The points $z_1,\cdots,z_n$ are called the nodes of the (tensor) product $\XX = X_1(z_1)\cdots X_n(z_n)$ of $X_j(z_j)$'s and we write $S_\XX$ for $\{z_1,\cdots,z_n\}.$
We extend the collection of Fock space fields by adding finite linear combinations of formal bosonic fields 
$$\Phi[\bfs\tau^+,\bfs\tau^-]:=\sum \tau_j^+\Phi^+(z_j) - \tau_j^-\Phi^-(z_j)$$
with the neutrality condition 
\begin{equation*} 
\sum_j (\tau_j^+ + \tau_j^-) = 0 \tag{$\NC_0$}
\end{equation*}
and its Wick's exponential $\ee^{\odot i\Phi[\bfs\tau^+,\bfs\tau^-]}.$ 
Here $\bfs\tau^+$ and $\bfs\tau^-$ are divisors in $D$ which takes the value $0$ at all points except for finitely many points:
$$\bfs\tau^\pm = \sum \tau_j^\pm\cdot z_j.$$
The centered Gaussian formal fields $\Phi^+,\Phi^-$ have formal correlations 
$$\E[\Phi^+(z_1)\Phi^+(z_2)] = \log\frac1{z_1-z_2},\qquad \E[\Phi^+(z_1)\Phi^-(z_2)] = \log(1-z_1\bar z_2)$$
in $\D$ and the relations: $\Phi^- = \overline{\Phi^+},\Phi = \Phi^++\Phi^-.$

\subsubsec{A spherical conformal field theory} We consider a compact Riemann surface $S$ of genus-zero with (finitely many) marked points $q_k.$ 
We introduce the formal fields $\Psi^\pm$ ($\Psi^- = \overline{\Psi^+}$) on $S$ with formal correlations
$$\E[\Psi^+(z_1)\Psi^+(z_2)] = \log\frac1{z_1-z_2},\qquad \E[\Psi^+(z_1)\Psi^-(z_2)] = 0$$
on $\wh\C.$ 
(A bi-variant field $\Psi^\pm(z,z_0) = \Psi^\pm(z) - \Psi^\pm(z_0)$ can be defined as a multivalued correlation functional valued field.)
We then modify $\Psi^+$ (and $\Psi^-$ in a similar way) by adding a non-random pre-pre-Schwarzian form (PPS form) $\psi^+$ of order $(ib,0)$ with logarithmic singularities at $q_k$'s.
Here $b$ is the modification parameter. 
Such $\psi^+$ resolves a background charge $\bfs\beta = i\pa\bp \psi^+/\pi = \sum \beta_k \delta_k.$ 
We write $\psi^+_{\bfs\beta}=\psi^+,$ $\Psi^+_{\bfs\beta}= \Psi^+ + \psi^+_{\bfs\beta}.$ 
A version of the Gauss-Bonnet theorem leads to the neutrality condition $(\NC_b):$
$$\sum_k \beta_k = 2b.$$ 
For $\bfs\tau = \sum \tau_j\cdot z_j$ satisfying $(\NC_0)$ we define the OPE exponential $\OO_{\bfs\beta}[\bfs\tau]$ (whose meaning is explained in Subsection~\ref{ss: V}) of $\Psi^+[\bfs\tau]$ by 
$$\OO_{\bfs\beta}[\bfs\tau] := \frac{C_{(b)}[\bfs\tau + \bfs\beta]}{C_{(b)}[\bfs\beta]}\,\ee^{\odot i\Psi^+[\bfs\tau]}, \quad \Psi^+[\bfs\tau] := \sum \tau_j^+\Psi^+(z_j),$$
where the (formal) Coulomb gas correlation function $C_{(b)}[\bfs\sigma]$ is a $\bfs\lambda$-differential $(\bfs\lambda = \sum\lambda_j\cdot z_j, \lambda_j = \sigma_j^2/2-\sigma_jb)$ such that 
its evaluation on $\wh\C= \C \cup\{\infty\}$ is given by
$$(C_{(b)}[\bfs\sigma] \,\|\, \id_{\wh\C}) = \prod_{j < k}(z_j-z_k)^{\sigma_j\sigma_k},$$
where the product is taken over finite $z_j$'s and $z_k$'s.
One can view $C_{(b)}[\bfs\sigma]$ as the formal correlations
$$C_{(b)}[\bfs\sigma] = \E\, \ee^{\odot i\sigma_1\Psi^+(z_1)}\cdots \ee^{\odot i\sigma_n\Psi^+(z_n)}$$
by Wick's rule of centered jointly (generalized) Gaussians:
$$\ee^{\odot \alpha_1\xi_1} \cdots \ee^{\odot \alpha_n\xi_n} = \ee^{\sum_{j<k}\alpha_j\alpha_k\E[\xi_j\xi_k]} \ee^{\odot \sum_j \alpha_j\xi_j}.$$

\subsubsec{Schottky double construction}
We consider a simply-connected domain $D$ with marked boundary/interior points $q_k \in \overbar D$ and the Schottky double $S$ of $D$ with the canonical involution $z\mapsto z^*.$
For a divisor $\bfs\sigma = \sum \sigma_j\cdot z_j,$ let us denote 
$\bfs\sigma_* = \sum \sigma_j\cdot z_j^*, \overbar{\bfs\sigma} = \sum \overbar\sigma_j\cdot z_j.$
Suppose that a (double) background charge $(\bfs\beta^+,\bfs\beta^-)$ $(\bfs\beta^\pm = \sum\beta_k^\pm \cdot q_k)$ satisfies the neutrality condition $(\NC_b):$
$$ \sum (\beta_k^+ + \beta_k^-)=2b.$$
For this pair $(\bfs\beta^+,\bfs\beta^-)$ and a double divisor $(\bfs\tau^+,\bfs\tau^-)$ satisfying the neutrality condition $(\NC_0),$ we define the modification $\Phi_{\bfs\beta}$ of $\Phi$ $(\bfs\beta = \bfs\beta^++\bfs\beta^-_*, \bfs\tau = \bfs\tau^++\bfs\tau^-_*)$ by 
$$\Phi_{\bfs\beta}[\bfs\tau] = \frac1{\sqrt2} \Psi_{\sqrt2\,\bfs\beta^+,\sqrt2\,\bfs\beta_*^-}[\bfs\tau, \bfs\tau_*],$$
where 
$\Psi_{\bfs\beta^+,\bfs\beta^-}[\bfs\tau^+,\bfs\tau^-]:= \Psi_{\bfs\beta^+}^+[\bfs\tau^+]+\Psi_{\bfs\beta^-}^-[\bfs\tau^-]$ and $\Psi_{\bfs\beta^-}^-[\bfs\tau^-] = \overline{\Psi^+_{\overbar{\bfs\beta^-}}[\overline{\bfs\tau^-}]}.$
We also define the multi-vertex field $V_{\bfs\beta}[\bfs\beta+\bfs\tau]$ and the OPE exponential $\OO_{\bfs\beta}[\bfs\tau]$ by 
$$V_{\bfs\beta}[\bfs\beta+\bfs\tau] = C_{(b)}[\bfs\beta+\bfs\tau]\,\ee^{\odot i\Phi[\bfs\tau]}, \qquad \OO_{\bfs\beta}[\bfs\tau] = \frac{C_{(b)}[\bfs\tau + \bfs\beta]}{C_{(b)}[\bfs\beta]}\,\ee^{\odot i\Phi[\bfs\tau]}.$$

In \cite{KM13} we study the standard chordal theory with $\bfs\beta^+ = 2b\cdot q,$ $\bfs\beta^- = \bfs0$ ($q\in \pa D$).
In this case, the modification $\Phi_{\bfs\beta}$ reads as 
$$\Phi_{\bfs\beta} = \Phi + \varphi_{\bfs\beta},\qquad \varphi_{\bfs\beta} = -2b\,\arg w',$$
where $w:(D,q)\to(\H,\infty)$ is a conformal map from $D$ onto the upper half-plane $\H.$

For the standard radial theory ($\bfs\beta^+ = \bfs\beta^- = b\cdot q,$ $q\in D$) without further marked points or spins, the modification $\Phi_{\bfs\beta}$ reads as 
$$\Phi_{\bfs\beta} = \Phi + \varphi_{\bfs\beta},\qquad \varphi_{\bfs\beta} = -2b\,\arg \frac{w'}w,$$
where $w:(D,q)\to(\D,0)$ is a conformal map from $D$ onto the unit disc $\D.$
The corresponding conformal field theory is $\Aut(D,q)$-invariant. 
The non-random harmonic function $\varphi_{\bfs\beta}$ on a punctured domain $D^*:=D\setminus\{q\}$ is multivalued.

The \emph{bosonic} fields $\Phi_{\bfs\beta}[\bfs\tau]$ with the neutrality condition generate the \emph{OPE family} $\FF_{\bfs\beta},$ the algebra (over $\C$) spanned by $1$ and the derivatives of $\Phi_{\bfs\beta},$ $\OO_{\bfs\beta}[\bfs\tau]$ ($\supp\,\bfs\beta \cap \supp\,\bfs\tau = \emptyset$) under the OPE multiplication $*.$
(If a single-valued field $X$ is holomorphic, then the OPE product $X*Y$ of two fields $X$ and $Y$ is the zeroth coefficient of the regular part in the operator product expansion $X(\zeta)Y(z)$ as $\zeta\to z.$)
For example, $\FF_{\bfs\beta}$ contains 
$1, J_{\bfs\beta}:= \pa\Phi_{\bfs\beta}, \pa J_{\bfs\beta}*(\Phi_{\bfs\beta}*\Phi_{\bfs\beta}),$ etc.,
and the \emph{Virasoro field} $T_{\bfs\beta},$
$$T_{\bfs\beta}=-\frac12 J_{\bfs\beta}* J_{\bfs\beta}+ib\pa J_{\bfs\beta}.$$
In Section~\ref{sec: F} we extend the OPE family by adding the OPE functionals at the punctures where the background charges are placed. 
For example, for $\check{\bfs\beta\,} = b\cdot q + b\cdot q^*,$ $(q\in D)$ 
$$(J_{\check{\bfs\beta\,}})_q = J(q) + \frac{ib}2N_w(q), \qquad N_w = \frac{w''}{w'}$$
and $\OO_{\bfs\beta}[\bfs\tau]$ ($\supp\,\bfs\beta \cap \supp\,\bfs\tau \ne \emptyset$) belong to the extended OPE family. 

\subsubsec{Insertion formulas} 
We explain how the insertion of Wick's exponential $V^{\odot}[\bfs\tau]:=\ee^{\odot i\Phi[\bfs\tau]}$ of the Gaussian free field gives rise to the change of background charge modifications. 
A particular insertion procedure plays an essential role in establishing the connection between chordal/radial SLE theory with force points (and spins) and conformal field theory in a simply-connected domain with certain symmetric background charges. 
\begin{thm} \label{main: change of beta}
Given two background charges $\check{\bfs\beta\,}, \bfs\beta $ with the neutrality condition $(\NC_b),$
the image of $\FF_{\check{\bfs\beta\,}}$ under the insertion of $V^{\odot}[\bfs\beta-\check{\bfs\beta\,}]$ is $\FF_{\bfs\beta}.$ 
\end{thm}

\subsubsec{Ward's equations and BPZ-Cardy equations}
We mostly concern ourselves with a symmetric background charge $\bfs\beta = \overbar{\bfs\beta_*}.$ 
One of our main goals is to present the precise relation between conformal field theory $\FF_{\bfs\beta}$ and $\SLE[\bfs\beta]$ for such a symmetric background charge $\bfs\beta,$ see Theorem~\ref{main} below.

Let $A_{\bfs\beta}:=T_{\bfs\beta}-\E\,T_{\bfs\beta}.$
For a meromorphic vector field in $\bar D,$ we define 
$$W_{\bfs\beta}(v)=2\,\Re\,W_{\bfs\beta}^+(v), \qquad W_{\bfs\beta}^+(v)=\frac1{2\pi i}\int_{\pa D} vA_{\bfs\beta}-\frac1{\pi }\int_{D} (\bp v)A_{\bfs\beta},$$
where $\bp v$ in the last integral should be understood in the sense of distributions and the first integral should be taken in the sense of the Cauchy principal value if necessary.
The following theorem enables us to represent the action of the Lie derivatives operators $\LL_v$ by inserting the Ward functionals into correlation functions of fields in the extended OPE family.

\begin{thm}[Ward's identities] \label{main: Ward identities in D}
Let $X_{\bfs\beta}$ be a tensor product of fields in the extended OPE family $\FF_{\bfs\beta}.$ 
If the nodes of $X_{\bfs\beta}$ do not intersect with the poles of $v,$ then we have 
$$\E\,\LL_v X_{\bfs\beta} = \E\,W_{\bfs\beta}(v) X_{\bfs\beta}.$$
\end{thm}
The proof of Theorem~\ref{main: Ward identities in D} is based on the following residue form of Ward identity
\begin{equation} \label{main: eq: local Ward}
\LL_v^+(z) X_{\bfs\beta} = \frac1{2\pi i}\oint_{(z)} vA_{\bfs\beta}X_{\bfs\beta}
\end{equation}
for a field $X_{\bfs\beta}$ in the OPE family with a proper meaning of equality. 
Here $\LL_v^+$ is the $\C$-linear part of the Lie derivative $\LL_v.$ 
We emphasize that \eqref{main: eq: local Ward} has a different nature at $z$ in $\supp\,\bfs\beta^+ \cup \supp\,\bfs\beta^-$ and at each node $z$ of $X_{\bfs\beta}.$
As Fock space functionals, \eqref{main: eq: local Ward} holds at each node $z$ of $X_{\bfs\beta},$ see Lemma~\ref{*F}. 
On the other hand, \eqref{main: eq: local Ward} holds \emph{within correlations} at $z$ in $\supp\,\bfs\beta^+ \cup \supp\,\bfs\beta^-,$ see Lemma~\ref{Ward@q}.

Ward's equations (see Theorem~\ref{main: Ward's equations} below) describe the insertion of the Virasoro field within correlations of fields in $\FF_{\bfs\beta}$ in terms of the Lie derivative operators and the puncture operators $\PP_{\bfs\beta} := C_{(b)}[\bfs\beta].$
Let $k_z$ (resp. $v_z$) be the chordal (resp. radial) Loewner vector field with pole at $z:$
$$k_z(\zeta) =\frac1{z-\zeta}, \quad \Big(\textrm{resp. } v_z(\zeta) = \zeta\frac{z+\zeta}{z-\zeta}\Big)$$
in the identity chart of $\H$ (resp. in the identity chart of $\D$). 

\begin{thm}[Ward's equations] \label{main: Ward's equations}
Given a symmetric background charge $\bfs\beta $ with the neutrality condition $(\NC_b),$
let $X_1,\cdots, X_n\in \FF_{\bfs\beta}$ and let $X$ be the tensor product of $X_j$'s.
Then
$$\E\, T_{\bfs\beta}(z)\,X=\E\,\PP_{\bfs\beta}^{-1} \LL^+_{k_z} \PP_{\bfs\beta} X+\E\,\PP_{\bfs\beta}^{-1} \LL^-_{k_{\bar z}}\PP_{\bfs\beta} X,$$
where all fields are evaluated in the identity chart of $\H$ 
and 
$$2z^2\E\, T_{\bfs\beta}(z)\,X=\E\,\PP_{\bfs\beta}^{-1} \LL^+_{v_z} \PP_{\bfs\beta} X+\E\,\PP_{\bfs\beta}^{-1} \LL^-_{v_{z^*}}\PP_{\bfs\beta} X,$$
where all fields are evaluated in the identity chart of $\D.$ 
\end{thm}

Let $\bfs\beta$ be a symmetric background charge with the neutrality condition $(\NC_b)$ and a specific charge $a$ at a marked boundary point $p\in\pa D.$ 
We consider a symmetric divisor $\bfs\tau = a\cdot p + \cdots$ satisfying the neutrality condition $(\NC_0)$ 
and let $\check{\bfs\beta\,} = \bfs\beta - \bfs\tau.$
By Theorem~\ref{main: change of beta}, $\FF_{\bfs\beta}$ is the image of $\FF_{\check{\bfs\beta\,}}$ under the insertion of 
$$\frac{\leg_p}{\E\,\leg_p},$$
where $\leg_p := \OO_{\check{\bfs\beta\,}}[\bfs\tau]$ is the one-leg operator and $\E\,\leg_p = C_{(b)}[\bfs\beta]/ C_{(b)}[\check{\bfs\beta}]$ is its correlation.
Let $\XX$ be a string in $\FF_{\check{\bfs\beta\,}}.$
We assume that $\zeta$ is not a node of $\XX.$ 
Denote 
$$\wh\E_\zeta\,\XX : = \frac{\E\,\leg_\zeta \XX}{\E\,\leg_\zeta}.$$
We define the partition function associated with a symmetric background charge $\bfs\beta$ by 
$$Z_{\bfs\beta} := \big|C_{(b)}[\bfs\beta]\big|$$
and let 
$Z_\zeta = Z_{\bfs\beta_\zeta}, \bfs\beta_\zeta= \bfs\beta + a \cdot \zeta- a\cdot p.$ 
The following version of BPZ-Cardy equations play a crucial role in understanding precise relation between conformal field theory and SLE.

\begin{thm}\label{main: BPZ-Cardy}
Given a symmetric background charge $\bfs\beta $ with $\bfs\beta(p)=a$ and the neutrality condition $(\NC_b),$ let $X_1,\cdots, X_n\in \FF_{\bfs\beta}$ and let $X$ be the tensor product of $X_j$'s.
Suppose that the parameters $a$ and $b$ satisfy 
$$2a(a+b) = 1.$$
If $\xi\in\R,$ then in the identity chart of $\H,$ we have
$$\frac1{2a^2}\Big(\pa_\xi^2 + 2\big(\pa_\xi\log Z_\xi \big) \pa_\xi \Big) \wh\E_{\xi}\,X= \check \LL_{k_\xi}\wh\E_{\xi}\,X,$$
where $\pa_\xi = \pa + \bp$ is the operator of differentiation with respect to the real variable $\xi$ and $\check{\LL}_{k_\xi}$ is taken over the finite notes of $X$ and $\supp\,\bfs\beta_\xi\setminus\{\xi\}.$
In the $(\D,0)$-uniformization we have
$$-\frac2{a^2}\Big(\frac12\pa_\theta^2 + \big(\pa_\theta\log Z_\zeta \big) \pa_\theta\Big)\wh\E_\zeta X = \check{\LL}_{v_\zeta}\wh\E_\zeta X, \qquad(\zeta = \ee^{i\theta}, \theta\in\R),$$
where the Lie derivative $\check{\LL}_{v_\zeta}$ is taken over the finite notes of $X$ and $\supp\,\bfs\beta_\zeta\setminus\{\zeta\}.$
\end{thm}

\subsection{SLE and martingale-observables} \label{ss: intro MO} 
Since Schramm introduced SLE in \cite{Schramm00} as the only possible candidates for the scaling limits of interface curves in critical 2-D lattice models, SLE has been used with remarkable success to prove some important conjectures in statistical physics. 
For example, see the work of Lawler-Schramm-Werner (\cite{LSW01a,LSW04}), Schramm-Sheffield (\cite{SS09,SS13}), and Smirnov (\cite{Smirnov01,Smirnov10}).

\subsubsec{Standard chordal/radial SLEs} 
For a simply-connected domain $(D,p,q)$ with a marked boundary point $p$ and a marked interior point $q,$ radial Schramm-Loewner evolution (SLE) in $(D,p,q)$ with a positive parameter $\kappa$ is the conformally invariant law on random curves from $p$ to $q$ satisfying the so-called ``domain Markov property" (see \eqref{eq: D Markov} below). 
Let us review the basic definitions of radial $\SLE(\kappa)$ in technical terms. For each $z \in D,$ let $g_t(z)$ be the solution (which exists up to a time $\tau_z\in(0,\infty]$) of the equation
\begin{equation} \label{eq: g radial}
\partial_t g_t(z) = g_t(z)\frac{\zeta_t+g_t(z)}{\zeta_t-g_t(z)}, 
\end{equation}
where $g_0:(D, p,q)\to (\D,1,0)$ is a given conformal map, $\zeta_t = \ee^{i\theta_t},\theta_t=\sqrt\kappa B_t,$ and $B_t$ is a standard Brownian motion with $B_0=0.$
Then for all $t,$ 
$$w_t:(D_t,\gamma_t,q)\to(\D,1,0), \qquad w_t(z):=g_t(z)/\zeta_t= g_t(z)\ee^{-i\sqrt{\kappa}B_t}$$
is a well-defined conformal map from 
$D_t := \{z \in D: \tau_z>t\}$
onto the unit disc $\D.$
It is known that the SLE stopping time $\tau_z$ (defined to be the first time when the solution of Loewner equation \eqref{eq: g radial} does not exist) satisfies
$\lim_{t\uparrow \tau_z} w_t(z) = 1.$
The \emph{radial SLE curve} $\gamma$ is defined by the equation
$$\gamma_t \equiv \gamma(t) := \lim_{z\to1} w_t^{-1}(z)$$
and satisfies the ``domain Markov property," 
\begin{equation} \label{eq: D Markov}
\mathrm{law}\,\left( \gamma[t,\infty)\,|\,\gamma[0, t]\right)\,=\,\mathrm{law}\,\gamma_{D_t,\gamma_t,q}[0,\infty).
\end{equation}
The sets $K_t:= \{z \in \overline{\D}: \tau_z\le t\}$ are called the \emph{hulls} of the SLE. 

For a simply-connected domain $(D,p,q)$ with two marked boundary points $p,q,$ the chordal SLE map $g_t$ with a positive parameter $\kappa$ satisfies the equation 
\begin{equation} \label{eq: g chordal}
\partial_t g_t(z) = \frac{2}{g_t(z) -\xi_t}, 
\end{equation}
where $g_0:(D, p,q)\to (\H, 0,\infty)$ is a given conformal map and $\xi_t = \sqrt{\kappa}B_t.$
Then for all $t,$
$$w_t(z):=g_t(z) -\xi_t$$
is a well-defined conformal map from the domain $D_t := \{z \in D: \tau_z>t\}$
onto $\mathbb{H},$
where the SLE stopping time $\tau_z$ satisfies 
$\lim_{t\uparrow \tau_z} w_t(z) = 0.$
The \emph{chordal SLE curve} $\gamma$ is defined by the equation
$$\gamma_t \equiv \gamma(t) := \lim_{z\to0} w_t^{-1}(z).$$

\subsubsec{Martingale-observables} 
Many results in the SLE theory and its applications depend on the explicit form of certain martingale-observables. 
A non-random \emph{conformal} field $f$ is an assignment of a (smooth) function
$(f\,\|\,\phi): ~\phi U\to\C$
to each local chart $\phi:U\to\phi U.$
For a non-random conformal field $M$ of $n$ variables either in $\H$ (in the chordal case) or in $\D$ (in the radial case), let us define $M$ in any simply-connected domain $(D,p,q)$ with marked points $p\in \pa D, q\in \pa D \cup D$ by 
$$(M_{D,p,q}\,\|\,\id) = (M\,\|\,w^{-1}).$$
Here $w:(D,p,q)\to(\H,0,\infty)$ is a conformal map in the chordal case ($q \in \pa D$) and $w:(D,p,q)\to(\D,1,0)$ is a conformal map in the radial case ($q \in D$). 
We say that $M$ is a \emph{martingale-observable} for chordal/radial $\SLE(\kappa)$ if for any $z_1,\cdots ,z_n\in D,$ the process
$$M_t(z_1,\cdots, z_n)=M_{D_t,\gamma_t,q}(z_1,\cdots, z_n)$$
(stopped when any $z_j$ exits $D_t$) is a local martingale on chordal/radial SLE probability space. 
For instance, we can use the identity chart of $D,$ and then for $[h^+,h^-]$-differentials $M$ with conformal dimensions $[h_q^+, h_q^-]$ at $q,$ we have 
$$M_t(z) = (w_t'(z))^{h^+} (\overline{w_t'(z)})^{h^-} (w_t'(q))^{h_q^+} (\overline{w_t'(q)})^{h_q^-}M(w_t(z)).$$ 
The notion of martingale-observables can be extended naturally to chordal/radial $\SLE[\bfs\beta]$ with force points and spins for a general symmetric background charge $\bfs\beta = (\bfs\beta^+,\bfs\beta^-)$ with the neutrality condition $(\NC_b)$ and a specific charge $a$ at $p,$ see Theorem~\ref{main} below. 
The difference $\lambda_b(\beta_k^+)-\lambda_b(\overline{\beta_k^+})$ is called the (conformal) \emph{spin} of $\bfs\beta$ at $q_k.$ 
Here $\lambda_b(\sigma) := {\sigma^2}/2 - \sigma b$ is the conformal dimension of charge $\sigma.$
 
\subsubsec{SLEs with force points and spins} 
We now generalize the standard chordal/radial SLEs to $\SLE[\bfs\beta]$ with finitely many force points and spins. 
Let 
\begin{equation} \label{eq: ab} 
a = \pm\sqrt{2/\kappa}, \quad b = a(\kappa/4-1)
\end{equation}
so that $a$ and $b$ satisfy $2a(a+b)=1.$ 
By definition, for a given symmetric background charge $\bfs\beta$ on $S = D^{\mathrm{double}}$ satisfying the neutrality condition $(\NC_b)$ with a specific charge $a$ at $p,$
the chordal (resp. radial) $\SLE[\bfs\beta]$ map $g_t$ satisfies the chordal Loewner equation~\eqref{eq: g chordal} (resp. the radial Loewner equation~\eqref{eq: g radial}) driven by the real process $\xi_t:$ 
\begin{equation} \label{eq: driving for chordal SLE[beta]}
\dd\xi_t = \sqrt\kappa\, \dd B_t + \lambda(t)\,\dd t, \quad \lambda(t) = (\lambda\,\|\,g_t^{-1}), \quad \lambda = \kappa \,\pa_\xi \log Z_{\bfs\beta_\xi},
\end{equation}
(resp. by the real process $\theta_t:$
\begin{equation} \label{eq: driving for radial SLE[beta]}
\dd\theta_t = \sqrt\kappa\, \dd B_t + \lambda(t)\,\dd t, \quad \lambda(t) = (\lambda\,\|\,g_t^{-1}), \quad \lambda = \kappa\, \pa_\theta \log Z_{\bfs\beta_\zeta}, \quad \zeta = \ee^{i\theta}, \quad \zeta_t = \ee^{i\theta_t}.)
\end{equation}

\ss The chordal $\SLE(\kappa,\bfs\rho)$ driven by
$$\dd\xi_t = \sqrt{\kappa}\, \dd B_t + \sum_k \frac{\rho_k\,\dd t}{\xi_t - q_k(t)},\qquad 
\dd q_k(t) = \frac2{q_k(t)-\xi_t}\,\dd t$$
has the background charge 
$$\bfs\beta = a\cdot p + \sum_k \beta_k\cdot q_k +(2b-a-\beta)\cdot q,$$
where $\beta_k=\rho_k/\sqrt{2\kappa}$ and $\beta = \sum_k\beta_k.$ 
It is known that the chordal $\SLE(\kappa,\bfs\rho)$ is almost surely a continuous path, see \cite[Theorem~1.3]{MS16}.
The radial $\SLE_\eta(\kappa,\bfs\rho)$ 
driven by
$\zeta_t = \ee^{i\theta_t}:$ 
\begin{equation} \label{eq: d theta}
\dd\theta_t = \sqrt\kappa \,\dd B_t + \eta\, \dd t + i\sum_k \frac{\rho_k}{2} \cot \frac{\theta_t - \vartheta_k(t)}2 \,\dd t, \quad\dd q_k(t) = q_k(t)\frac{\zeta_t + q_k(t)}{\zeta_t - q_k(t)}\,\dd t, \quad q_k = \ee^{i\vartheta_k}
\end{equation}
has the background charge 
$$\bfs\beta = a\cdot p + \sum_k \beta_k\cdot q_k + \Big(b-\frac {a+\beta+i\delta}2\Big)\cdot q+ \Big(b-\frac {a+\beta-i\delta}2\Big)\cdot q^*, \qquad \delta = \eta a.$$
In \cite{MS17}, Miller and Sheffield identified the radial $\SLE_\eta(\kappa,\bfs\rho)$ with the flow lines of a (formal) vector field $\OO_{\bfs\beta}[\sigma\cdot z, - \sigma\cdot z],$
where $\sigma$ is chosen such that the conformal spin $s:=\lambda^+-\lambda^-=-2\sigma b= 1.$

\ss We now briefly explain how the chordal $\SLE[\bfs\beta]$ process can be produced from the standard chordal $\SLE(\kappa)$ or $\SLE[\bfs\beta^0]$ ($\bfs\beta^0 := a\cdot p + (2b-a)\cdot q$) by the density $\E\,\OO_{\bfs\beta^0}[\bfs\beta-\bfs\beta^0].$ 
By Proposition~14.3 in [KM13] or the special case of Theorem~\ref{main} below, the ratio of partition functions
$$M:=\frac{Z_{\bfs\beta}}{Z_{\bfs\beta^0}} = \E\,\OO_{\bfs\beta^0}[\bfs\beta-\bfs\beta^0]$$
is a martingale-observable for $\SLE[\bfs\beta^0]$ with
$$\frac{\dd M_t}{M_t} = \sqrt\kappa\, \pa_\xi |_{\xi =\xi_t} (\log \frac{Z_{\bfs\beta_\xi}}{Z_{\bfs\beta_\xi^0}}\,\|\, g_t^{-1})\,\dd B_t.$$
From the fact that $Z_{\bfs\beta_\xi^0}=1$ in the $(\H,\infty)$-uniformization, it follows that 
the drift term 
\begin{equation}\label{eq: Girsanov}
\frac{\dd\langle \xi, M\rangle_t}{M_t} = \kappa\, \pa_\xi\big|_{\xi=\xi_t} (\log Z_{\bfs\beta_\xi}\,\|\, g_t^{-1}) \,\dd t
\end{equation}
from Girsanov's theorem corresponds to the drift in \eqref{eq: driving for chordal SLE[beta]}, see \cite[Theorem~6]{SW05}.

\ms In Subsection~\ref{ss: proof of main theorem} we present a certain collection of chordal/radial SLE martingale-observables not by It\^o's calculus but conformal field theory. 

\begin{thm} \label{main}
Let $\bfs\beta$ be a symmetric background charge with the neutrality condition $(\NC_b)$ and a specific charge $a$ at $p$ satisfying \eqref{eq: ab}. 
Suppose $\XX_{\bfs\beta}$ is a string in the OPE family $\FF_{\bfs\beta}$ of $\Phi_{\bfs\beta},$ then the non-random field
$$M =\E\,\XX_{\bfs\beta}$$
is a martingale-observable for chordal/radial $\SLE[\bfs\beta].$
\end{thm}

In the physics literature, particular forms of this version appeared in \cite{BB03, BB04, Cardy04b, Kytola06, RBGW07}.
In \cite{Dubedat09} Dub\'edat coupled $\SLE[\bfs\beta]$ with $\Phi_{\bfs\beta}.$

\subsection{Examples of SLE martingale-observables} 

In Subsection~\ref{ss: EgsMO} we present several examples of radial SLE martingale-observables.

\subsubsec{Radial $\SLE(0)$ observables} The case $\kappa = 0$ reveals some aspects of ``field Markov property." 
Indeed, $\SLE(0)$ curves are hyperbolic geodesics, and martingale-observables (of one variable) are non-random fields $F\equiv F_{D,p,q}$ with the property 
$$F\big|_{D_t} = F_{D_t,\gamma_t,q}.$$
In a sense one can think of them as integrals of the motion $t \mapsto \{D_t,\gamma_t,q\}$ in the corresponding Teichm\"uller space.
The reader is invited to check that 
$$\arg\big((1-w)w^{-3/2}w'\big),\qquad S_w + \frac38 \Big(\frac{w'}w\Big)^2\Big(1-\frac{4w}{(1-w)^2}\Big)$$
are $\SLE(0)$ observables. 
Here 
$S_w = N_w' -\frac12{N_w^2}, \, N_w = (\log w')'$
are Schwarzian and pre-Schwarzian derivatives of $w,$ respectively. 
(For the first $\SLE(0)$ observable, consider the radial version of Schramm-Sheffield martingale-observables, 
$$a\Big(\arg\frac{(1-w)^2}{w} - 2\big(\frac\kappa4-1\big)\arg\frac{w'}w\Big), \qquad (a = \sqrt{2/\kappa}),$$
(see the first example in Subsection~\ref{ss: EgsMO}) and normalize them so that the limit exists as $\kappa\to0.$
For the second $\SLE(0)$ observable, consider the 1-point functions of the Virasoro fields,
$$\frac{c}{12} S_w + h_{1,2} \frac{w'^2}{w(1-w)^2} + h_{0,1/2}\frac{w'^2}{w^2},$$
where the central charge $c$ and the conformal dimensions $h_{1,2}, h_{0,1/2}$ are given by 
$$c = \frac{(3\kappa-8)(6-\kappa)}{2\kappa},\quad h_{1,2} = \frac{6-\kappa}{2\kappa},\quad h_{0,1/2} = \frac{(6-\kappa)(\kappa-2)}{16\kappa}.$$
See \eqref{eq: T hat}.) 
Recently, in \cite{AKM20} multiple $\SLE(0)$ observables are used to show Peltola and Wang's theorem (\cite{PW20}): 
multiple $\SLE(0)$ curves are contained in the real locus of real rational functions with prescribed real critical points.

\subsubsec{Radial $\SLE(\kappa)$ observables} 
For $\kappa >0, $ the usual way to find martingale-observables of a given conformal type is using It\^o's calculus. 
A couple of well-known important examples are referred to below. 

\begin{eg*}[$\kappa = 2$] The scalar (i.e., $[0,0]$-differential)
$$M =\frac{1-|w|^2}{|1-w|^2} = \frac{P_\D(1,w)}{P_\D(1,0)} = \frac{P_D(p,z)}{P_D(p,q)}$$
played an important role in the theory of loop erased random walk (LERW), see \cite{LSW04}. 
Here $P_D$ is the Poisson kernel of a domain $D.$ 
\end{eg*}

\begin{eg*}[$\kappa = 6$] The martingale-observable
$$M_t(z) = \ee^{t/4}\frac{\sqrt[3]{1-w_t(z)}}{\sqrt[6]{w_t(z)}}$$
is a scalar with respect to $z$ and a $[\frac18,\frac18]$-differential with respect to $q.$
Lawler, Schramm, and Werner applied the optional stopping theorem to the martingale $M_t(\ee^{i\theta})$ and estimated the probability that the point $\ee^{i\theta}$ is not swallowed by the $\SLE_6$ hull $K_t$ at time $t$ to be 
$$\P[\, \ee^{i\theta} \notin K_t\,] \asymp \ee^{-2t h_q} \sqrt[3]{\sin \frac\theta2},\qquad h_q = \frac18,\qquad (0\le\theta<2\pi).$$
The exponent $2h_q = 1/4$ is one of many exponents in \cite{LSW01c}.
See the example (derivative exponents on the boundary) in Subsection~\ref{ss: 1pt O} with $\kappa=6$ and $h=0.$ 
\end{eg*}

\subsubsec{Lawler-Schramm-Werner's derivative exponents} 
Examples of 1-point rooted vertex observables include Lawler-Schramm-Werner's derivative exponents (\cite{LSW01c}) of radial SLEs on the boundary: given $h$
$$\E[|w'_t(\ee^{i\theta})|^h \mathbf{1}_{\{\tau_{\ee^{i\theta}} > t\}}] \asymp \ee^{-2h_q t} \Big(\sin^2\frac{\theta}2\Big)^{\frac12a\sigma},$$
where $\sigma$ and $h_q$ are given by 
$$\sigma = \frac a 4\big(\kappa-4 + \sqrt{(\kappa-4)^2+16\kappa h } \big), \qquad h_q =\frac{\sigma^2}{8} + \frac{a\sigma}4.$$ 

\subsubsec{Friedrich-Werner's formula}
In the chordal case with $\kappa = 8/3$ $(c=0),$ it is well known (\cite{FW03}) that the $n$-point function $\wh\E\,[T(x_1)\cdots T(x_n)\,\|\,\id_\H\,]$ of Virasoro field $T \equiv T_{\bfs\beta}, \bfs\beta = 2b\cdot q$ coincides with Friedrich-Werner's function $B(x_1,\cdots,x_n):$
$$B(x_1,\cdots,x_n) = \lim_{\ve\to0}\frac{\P(\SLE(8/3)\textrm{ hits all }[x_j,x_j+i\ve\sqrt2])} {\ve^{2n}}.$$
We derive the radial version of this formula. 
See Theorem~\ref{radial FW}.

\subsubsec{Restriction property} 
We also use the one-leg operator $\leg$ to present a field theoretic proof of the restriction property of radial $\SLE(8/3)$: for all hull $K,$
\begin{equation}\label{eq: avoid_radial}
\P(\SLE(8/3) \textrm{ avoids }K) = |\Psi_K'(1)|^\lambda (\Psi_K'(0))^\mu,\qquad \Big(\lambda=\frac58,\,\mu = \frac5{48}\Big),
\end{equation}
where $\Psi_K$ is the conformal map $(\D\sm K,0)\to (\D,0)$ satisfying $\Psi_K'(0)>0.$
In particular, we explain the restriction exponents $\lambda$ and $\mu$ in terms of conformal dimensions of the one-leg operators and their effective versions: 
$$\lambda = h(\leg) := \frac{a^2}2-ab = \frac{6-\kappa}{2\kappa},\qquad \mu=H_q(\leg^\eff):=\frac{a^2}4-b^2 = \frac{(\kappa-2)(6-\kappa)}{8\kappa},$$
where $H_q = h_q^++h_q^-$ and 
$$\leg^\eff = \PP_q \,\leg=V\big[a\cdot z +\big(b- \frac{a}2\big)\cdot q+\big(b- \frac{a}2\big)\cdot q^*\big], \qquad (\PP_q= C_{(b)}[b\cdot q+b\cdot q^*]).$$
See Subsection~\ref{ss: Example: one-leg operators}.

\subsection{Neumann boundary condition and backward SLEs}
In Section~\ref{sec: Neumann CFT} we sketch an implementation of a version of conformal field theory with Neumann boundary condition. 
The Gaussian free field $N(z,z_0)$ in $D$ with Neumann boundary condition can be constructed from the Gaussian free field $\Psi(z)$ on $S = D^{\mathrm{double}}:$
$$N(z,z_0) = \frac1{\sqrt2} \big(\Psi(z,z_0)+\Psi(z^*,z_0^*)\big).$$ 
For a background charge $\bfs\beta $ on $S$ with the neutrality condition $(\NC_b),$ we introduce the background modification $N_{\bfs\beta}$ of $N.$ 
We present the connection between the OPE family $\FF_{\bfs\beta}^N$ of $N_{\bfs\beta}$ and the backward chordal/radial $\SLE[\bfs\beta].$

Suppose that the parameters $a$ and $b$ are related to the SLE parameter $\kappa$ as 
$$a =\pm\sqrt{2/\kappa}, \qquad b = -a (\kappa/4+1).$$
Let $\bfs\beta$ be a symmetric background charge with the neutrality condition $(\NC_b)$ and a specific charge $a$ at a marked boundary point $p\in\pa D.$
The backward chordal $\SLE[\bfs\beta]$ map $f_t$ from $\H$ satisfies the equation 
\begin{equation} \label{eq: f chordal}
\partial_t f_t(z) = -\frac{2}{f_t(z) -\xi_t} 
\end{equation}
driven by the real process $\xi_t:$ 
$$\dd\xi_t = \sqrt\kappa\, \dd B_t + \lambda(t)\,\dd t, \quad \lambda(t) = (\lambda\,\|\,f_t^{-1}), \quad \lambda = \kappa \,\pa_\xi \log Z_{\bfs\beta_\xi},$$
where the partition function $Z_{\bfs\beta_\xi}$ is given by $Z_{\bfs\beta_\xi} = C_{(-ib)}[-i\bfs\beta_\xi]$ and $\bfs\beta_\xi = \bfs\beta + a\cdot\xi-a\cdot p.$
Its radial counterpart can be defined in a similar way, see Subsection~\ref{ss: Neumann and backward}

We say a non-random conformal field $M$ is a martingale-observable for backward chordal/radial $\SLE[\bfs\beta]$
if for any $z_1,\cdots,z_n\in\H,$ the process
$$M_t(z_1,\cdots,z_n) = (M\,\|\,f_t^{-1})(z_1,\cdots,z_n)$$
is a local martingale on backward chordal/radial SLE probability space. 

\begin{thm} \label{main: Neumann SLE}
The correlations of fields in $\FF_{\bfs\beta}^N$ form a collection of martingale-observables for backward chordal/radial $\SLE[\bfs\beta].$ 
\end{thm}

In this theory, the central charge $c$ is given by 
$c = 1 + 12b^2 = 13 + 12(\kappa/8+2/\kappa) \ge 25.$

\section{Coulomb gas correlations}
In this section we introduce the Coulomb gas correlations as the (holomorphic) differentials 
with conformal dimensions
$\lambda_j = {\sigma_j^2}/2 - \sigma_j b$
at $z_j$'s (including infinity) and with values 
$$\prod_{\substack{j<k\\z_j,z_k\ne\infty}}(z_j-z_k)^{\sigma_j\sigma_k},\qquad (z_j\in\wh\C)$$
in the identity chart of $\C$ and the chart $z\mapsto -1/z$ at infinity. 
After explaining this definition, we prove that under the neutrality condition $\sum \sigma_j = 2b$ the Coulomb gas correlation functions are conformally invariant with respect to the M\"obius group $\Aut(\wh\C).$

\subsection{Coulomb gas correlations on the Riemann sphere}
Let
$$\bfs\sigma= \sum \sigma_j\cdot z_j,$$ 
where $\{z_j\}_{j=1}^N$ is a finite set of (distinct) points on $\wh\C$ and $\sigma_j$'s are real
numbers, (``charges" at $z_j$'s),
$\sigma_j =\sigma_{z_j} =\bfs\sigma(z_j).$
We can think of $\bfs\sigma$ as a divisor (a function $\bfs\sigma:\wh\C \to \R$ which takes the value $0$ at all points except for finitely many points) or as an atomic measure:
$\bfs\sigma= \sum \sigma_j\cdot \delta_{z_j}.$ 
So
$$\int \bfs\sigma= \sum \sigma_j.$$
(Some of $\sigma_j$'s can be zero, and in any case $\sigma_z = 0$ if $z$ is not one of $z_j$'s.
Sometimes we allow $\sigma_j$'s to be complex. 
We need this case e.g., for the one-leg operators with spin $s\in\C,$ see Subsection~\ref{ss: 1-leg radial}.)
For a divisor $\bfs\sigma$ in $\C,$ i.e., $\sigma_\infty = 0,$ we define the correlation function $C_\C[\bfs\sigma]$ by 
\begin{equation} \label{eq: CG in C}
C_\C[\bfs\sigma]\equiv C_\C^{\bfs\sigma}(\bfs z) =\prod_{j<k}(z_j - z_k)^{\sigma_j\sigma_k}.
\end{equation}
This is a holomorphic function of $\bfs z$ in $\C^N_{\textrm{distinct}} = \{\bfs z = (z_1,\cdots,z_N)\in\C^N\,|\, z_j \ne z_k \textrm{ if } j\ne k\}.$

Typically it is multivalued except in special cases that all $\sigma$'s are integers. 
(If they are all even, then the order in the product \eqref{eq: CG in C} does not matter.)
If all $\sigma_j = 1,$ then the correlation function is the Vandermonde determinant.

In general, we need to interpret formulas in terms of single-valued branches.
(Sometimes, single-valuedness is referred to as ``physically.")
If all $z_j$'s are in the real line $\R,$ then it is physical. 
If $\bfs\sigma$ is symmetric or anti-symmetric with respect to $\R$ in $\C,$ then the correlation function is single-valued. 
See Examples (b) and (c) in Subsection~\ref{ss: C}.

We can extend this definition to divisors on $\wh\C$ by simply ignoring the
charge at infinity: if $\bfs\sigma:\wh\C\to\R,$ then
$C_{\wh\C}[\bfs\sigma]:=C_\C[\bfs\sigma|_\C].$

\subsection{Conformal weights and neutrality condition} 
Fix a (real) parameter $b.$ 
We say that a divisor $\bfs\sigma:\wh\C\to\R$ satisfies the neutrality condition ($\mathrm{NC}_b$) if
\begin{equation} 
\int \bfs\sigma= 2b.
\end{equation}
Note that there is a 1-to-1 correspondence between divisors on $\wh\C$ satisfying ($\mathrm{NC}_b$) and arbitrary divisors in $\C:$
$\bfs\sigma \mapsto \bfs\sigma|_\C.$
Let 
$$\lambda_b(\sigma) = \frac{\sigma^2}2 - \sigma b \quad (\sigma\in\R).$$ 
Using this function $\lambda:\R\to\R,$ we define the ``weights" or ``dimensions" $\lambda_j$ at $z_j$ by 
$$\lambda_j = \lambda_b(\sigma_j) \equiv \frac{\sigma_j^2}2 - \sigma_j b.$$
It is obvious that $\lambda_b(\sigma) = \lambda_b(2b-\sigma).$

\subsection{Coulomb gas correlation functions as differentials} \label{ss: C as differentials}
Let us recall the definitions of conformal fields and certain transformation laws such as differentials. 
A local coordinate chart on a Riemann surface $M$ is a conformal map $\phi:U\to\phi(U) \subset \C$ on an open subset $U$ of $M.$
By definition, a non-random \emph{conformal} field $f$ is an assignment of a (smooth) function
$(f\,\|\,\phi): ~\phi(U)\to\C$
to each local chart $\phi:U\to\phi(U).$
A non-random conformal field $f$ is a \emph{differential} of weights or (conformal) dimensions $[\lambda,\lambda_*]$ if for any two overlapping charts $\phi$ and $\wt\phi,$ we have
$$f = (h')^\lambda(\overline{h'})^{\lambda_*}\wt{f}\circ h,$$
where $h=\wt\phi\circ\phi^{-1}:~ \phi(U\cap\wt U)\to \wt\phi(U\cap\wt U)$ is the transition map, and $f$ (resp. $\wt f$) is the notation for $(f\,\|\,\phi),$ (resp. $(f\,\|\,\wt\phi)$). 
\emph{Pre-pre-Schwarzian forms} of order $(\mu,\nu)$ (or $\PPS(\mu,\nu)$ forms), \emph{pre-Schwarzian forms} of order $\mu$ (or $\PS(\mu)$ forms), and \emph{Schwarzian forms} of order $\mu$ are fields with transformation laws 
$$f=\wt{f}\circ h +\mu \log h' + \nu\,\log\overline{h'} ,\qquad f=h'\wt{f}\circ h +\mu N_h, \qquad f=(h')^2\wt{f}\circ h + \mu S_h,$$
respectively, where 
$$N_h =(\log h')',\qquad S_h = N_h' -\frac12{N_h^2}$$
are pre-Schwarzian and Schwarzian derivatives of $h.$ 
The transformation laws can be extended to the random field: e.g., a field $X$ is called a $[\lambda,\lambda_*]$-differential if the non-random field $z\mapsto \E[X(z)\YY]$ is a $[\lambda,\lambda_*]$-differential in $z$ for each $\YY.$

Let $S$ be a compact Riemann surface of genus zero, and let
$\bfs\sigma : S\to\R$
be a divisor satisfying ($\mathrm{NC}_b$). 
We define the Coulomb gas correlation function
$$C[\bfs\sigma] (\equiv C_S[\bfs\sigma]) \equiv C^{\bfs\sigma}(\bfs z)$$
as a (multivalued) holomorphic differential in $z_j$'s such that

\renewcommand{\theenumi}{\alph{enumi}}
\begin{enumerate}
\ss\item the conformal dimensions at $z_j$'s are the numbers $\lambda_j$'s;
\ss\item \label{item: C} if $S = \wh\C,$ then $C_\C^{\bfs\sigma}(\bfs z)$ is the value of $C^{\bfs\sigma}(z)$ in the identity charts at finite $z_j$'s and the chart $z\mapsto -1/z$ at infinity (in the case $\sigma_\infty\ne0)$.
\end{enumerate}

Alternatively, we can restate \eqref{item: C} as follows: if $\phi: S \to \wh\C$ is a uniformizing conformal map, then
$$\{C[\bfs\sigma]\,\|\,\phi\} = C_{\C}[\bfs\sigma\circ\phi^{-1}].$$
(By convention, $\|\phi$ refers to the chart $-1/\phi$ at the pre-image $\phi^{-1}(\infty)$ of infinity.)

\begin{egs*} (a) If $\bfs \sigma = 2b\cdot z,$ then
$C[\bfs\sigma] $ is a (scalar) function of $z$ namely, $C[\bfs\sigma] \equiv 1.$
 
\ss \no (b) If $\bfs\sigma = \sigma\cdot z_1 + (2b-\sigma) z_2,$ then
$$C[\bfs\sigma] = (z_1-z_2)^{-2\lambda}, \qquad \lambda = \lambda_b(\sigma) = \frac12\sigma(\sigma-2b)$$
for finite $z_1,z_2$ and $C[\bfs\sigma] \equiv 1$ for $z_2 = \infty.$
We note that $C[\bfs\sigma]$ is a $\lambda$-differential in $z_1,z_2,$ and a more traditional notation would be 
$$\bigg(\frac{dz_1\,dz_2}{(z_1-z_2)^2}\!\bigg)^\lambda.$$
M\"obius invariance of the bi-differential 
$$\frac{dz_1\,dz_2}{(z_1-z_2)^2}$$
is of course well known, and could be used as an alternate way to derive our theorem.
\end{egs*}

\subsection{M\"obius invariance} 
To justify the above definition of Coulomb gas correlation function, we need to verify that the differentials $C[\bfs\sigma]$ are M\"obius invariant on $\wh\C.$ 
Since translation invariance of $C[\bfs\sigma]$ is obvious, we need to verify the invariance of $C[\bfs\sigma]$ under the dilation-rotations $\tau(z) = az$ and the inversion $\tau(z)=-1/z.$ 
First we assume $\sigma_\infty = 0$ and denote
$$C(\bfs z) = \prod_{j<k}(z_j - z_k)^{\sigma_j\sigma_k}, \quad (\bfs z = (z_1,\cdots,z_N))$$
(so that $C = C_\C[\bfs\sigma]$).
For a M\"obius map $\tau,$ we also denote $\tau \bfs z = (\tau z_1,\cdots,\tau z_N).$ 

\begin{lem}
If $\bfs\sigma:\C\to\R$ satisfies the neutrality condition $(\mathrm{NC}_b)$, and if $\tau$ is a M\"obius map such that $\tau(z_j)\ne\infty,$ then
\begin{equation} \label{eq: Moebius}
C(\bfs z) = C(\tau \bfs z) \prod_j \big(\tau'(z_j)\big)^{\lambda_j}.
\end{equation}
\end{lem}

\begin{proof}
For a dilation-rotation $\tau(z) = az\,(a\in\C\sm\{0\}),$ we have
$$C(\tau \bfs z) = C(\bfs z) a^{\sum_{j<k} \sigma_j\sigma_k} \textrm{ and } 
\prod_j \big(\tau'(z_j)\big)^{\lambda_j} = a^{\sum_j \lambda_j},$$
so conformal invariance of $C[\bfs\sigma]$ under the dilation-rotations means that 
$$\sum_{j<k}\sigma_j\sigma_k +\frac12 \sum_j \sigma_j^2 = \sum_j \sigma_j b.$$
This identity holds by the neutrality condition $(\mathrm{NC}_b)$ -- both sides are equal to $\frac12(\sum_j\sigma_j)^2.$

For the inversion $\tau(z) = -1/z,$ we have
$$C(\tau \bfs z) = C(\bfs z)\prod_{j<k} (z_jz_k)^{-\sigma_j\sigma_k} \textrm{ and } 
\prod_j \big(\tau'(z_j)\big)^{\lambda_j} = \prod_j z_j^{2\sigma_jb - \sigma_j^2}.$$
Thus \eqref{eq: Moebius} reduces to
$$\prod_{j<k}(z_jz_k)^{\sigma_j\sigma_k}\prod_j z_j^{\sigma_j^2}=\prod_j z_j^{2\sigma_jb}.$$
This holds by $(\mathrm{NC}_b)$ because the exponents of $z_j$ on both sides coincide:
$$\sigma_j^2 + \sum_{k\ne j} \sigma_j\sigma_k = \sigma_j^2 + \sigma_j(2b-\sigma_j) = 2\sigma_jb.$$
\end{proof}

\begin{thm} \label{thm: conformal invariance}
Under the neutrality condition $(\mathrm{NC}_b)$, the differentials $C[\bfs\sigma]$ are M\"obius invariant on $\wh\C.$ 
\end{thm}
\begin{proof}
Due to the previous lemma, it remains to check the M\"obius invariance of $C[\bfs\sigma]$ on the Riemann sphere in the case that a charge $\sigma_\infty$ at infinity is non-zero. 
Let $$\bfs \sigma = \sigma_\infty\cdot\infty + \sum \sigma_j\cdot z_j.$$
We may assume that $\bfs\sigma(0) = 0$ by translation. 
For the inversion $\tau(z) = -1/z,$ denote $\wt{\bfs\sigma}:=\bfs\sigma\circ\tau^{-1}$ and $\wt C(\bfs z):=C_\C[\wt{\bfs\sigma}].$
Then 
$\wt{\bfs\sigma} = \sigma_\infty\cdot 0 +\sum \sigma_j\cdot(-1/z_j)$
and
\begin{equation} \label{eq: Moebius1}
\wt C(\bfs z) = C(\bfs z) \prod_{j<k} (z_jz_k)^{-\sigma_j\sigma_k} \prod_j z_j^{-\sigma_\infty\sigma_j}.
\end{equation}
It follows from 
$C(\bfs z) = \{C[\bfs\sigma]\,\|\,\id\}$ and $\wt C(\bfs z) = \{C[\bfs\sigma]\,\|\,\tau\}$
that the conformal invariance of $C[\bfs\sigma]$ under $\tau$ means 
\begin{equation} \label{eq: Moebius2}
C(\bfs z) =\wt C(\bfs z) \big(h'(0)\big)^{\lambda_\infty} \prod_j \big(\tau'(z_j)\big)^{\lambda_j} = \wt C(\bfs z) \prod_j z_j^{-2\lambda_j},
\end{equation}
where $h = \id$ is the transition map between the two charts at infinity.
The equations \eqref{eq: Moebius1}~--~\eqref{eq: Moebius2} reduce to the identity
$$\prod_{j<k}(z_j z_k)^{\sigma_j\sigma_k}\prod_j z_j^{\sigma_\infty\sigma_j+2\lambda_j}=1.$$
This holds by the neutrality condition $(\mathrm{NC}_b)$ because the exponent of $z_j$ on the left-hand side is
$$\sigma_j^2-2\sigma_j b + \sigma_j \Big(\sigma_\infty + \sum_{k\ne j}\sigma_k\Big)= 0.$$
Finally, for a dilation-rotation $\tau(z) = az\,(a\ne 0),$ we have 
$\wt{\bfs\sigma}:=\bfs\sigma\circ\tau^{-1}= \sigma_\infty\cdot\infty + \sum \sigma_j\cdot (az_j)$
and 
$$\wt C(\bfs z) = C(\bfs z) a^{\sum_{j<k}\sigma_j\sigma_k}.$$
We need to check that 
$$C(\bfs z) =\wt C(\bfs z) \big(h'(0)\big)^{\lambda_\infty} \prod_j \big(\tau'(z_j)\big)^{\lambda_j} = \wt C(z) a^{-\lambda_\infty+\sum_j\lambda_j},$$
where $h(w) = w/a$ is the transition map between the two charts at infinity. 
(Indeed, the first chart is $z\mapsto-1/z$ and the second one is $z\mapsto-1/\tau(z).$)
Thus the condition for conformal invariance is
$$\sum_j\lambda_j + \sum_{j<k}\sigma_j\sigma_k=\lambda_\infty.$$
The left-hand side simplifies to $\lambda_b(\sum_j\sigma_j)$. 
It follows from the neutrality condition ($\mathrm{NC}_b$) that
$$\lambda_b(\sum_j\sigma_j)=\lambda_b(2b-\sigma_\infty) = \lambda_b(\sigma_\infty) = \lambda_\infty.$$
\end{proof}

\subsection{Schottky double construction} \label{ss: C}
In this subsection we introduce Coulomb gas correlations for simply-connected domain $\subsetneq \C.$
They are constructed from those of the Schottky double.
We compute them in the $\H$- and $\D$-uniformizations. 

Suppose $D$ is a simply-connected domain ($D\subsetneq \C$). Let $\pa D$ be its Carath\'eodory ``boundary" (prime ends).
Consider the Schottky double
$S = D^{\mathrm{double}},$
which equips with the canonical involution 
$\iota\equiv\iota_D:S\to S,\,z\mapsto z^*.$
For example, we identify $\wh\C$ with the Schottky double of $\H$ or that of $\D.$
Then the corresponding involution $\iota$ is
$\iota: z\mapsto z^* = \bar z$ for $D=\H$ and $\iota: z\mapsto z^* = 1/\bar z$ for $D=\D.$

For two divisors $\bfs\sigma^+ = \sum \sigma_j^+\cdot z_j,$ and $\bfs\sigma^- = \sum \sigma_j^-\cdot z_j$ $ (\sigma_j^+, \sigma_j^-\in\C)$ in $D \cup \pa D,$ we define 
$$\bfs\sigma = \bfs\sigma^+ + \bfs\sigma_*^-,\qquad \bfs\sigma_*^- := \sum \sigma_j^-\cdot z_j^*.$$
Then $\bfs\sigma$ is a divisor in $S.$ 
(We may assume that $\bfs\sigma^-$ is a divisor in $D$, i.e., $\sigma_j^- = 0$ if $z_j\in\pa D.$) 
We call $(\bfs\sigma^+,\bfs\sigma^-)$ a \emph{double divisor} in $D \cup \pa D.$
By definition, 
$$C_D^{(\bfs\sigma^+,\bfs\sigma^-)}(\bfs z) (\equiv C_D[\bfs\sigma^+,\bfs\sigma^-]) := C_S[\bfs\sigma].$$
(We often omit the subscripts $D,S$ when there is no danger of confusion.) 
More precisely, the above definition means 
\begin{equation} \label{eq: def C on D}
\{C[\bfs\sigma^+,\bfs\sigma^-]\,\|\,\phi\} = \{C[\bfs\sigma]\,\|\,(\phi,i\circ\phi\circ\iota)\},
\end{equation}
where $i$ is the complex conjugation.

Let $\lambda_j^\pm = \lambda_b(\sigma_j^\pm).$ 
The following theorem is immediate from Theorem~\ref{thm: conformal invariance}.
\begin{thm} \label{C under NCb}
If $\bfs\sigma^+ + \bfs\sigma^-_*$ satisfies the neutrality condition $(\mathrm{NC}_b),$ then $C[\bfs\sigma^+,\bfs\sigma^-_*]$ is a well-defined differential with conformal dimensions $[\lambda_j^+,\lambda_j^-]$ at $z_j.$ 
\end{thm}

\subsubsection*{$\H$-uniformization}
Let us consider $D=\H.$ Then its prime end is $\pa D = \wh\R.$
We may assume that $\sigma_j^- = 0$ if $z_j\in\wh\R.$
We now define 
$$C_\H[\bfs\sigma^+,\bfs\sigma^-]=\prod_{j<k}(z_j-z_k)^{\sigma_j^+\sigma_k^+}(\bar z_j-\bar z_k)^{\sigma_j^-\sigma_k^-}\prod_{j,k} (z_j-\bar z_k)^{\sigma_j^+\sigma_k^-},$$
where the product is taken over finite $z_j$'s and $z_k$'s, and as always we use the convention $0^0 := 1.$

\begin{thm} \label{thm: C in H} 
Under the neutrality condition $(\mathrm{NC}_b),$
$C_\H[\bfs\sigma^+,\bfs\sigma^-]$ is the value of the differential $C[\bfs\sigma^+,\bfs\sigma^-]$ in the identity chart of $\H$ (and the chart $z\mapsto -1/z$ at infinity). 
\end{thm}

\begin{rmk*} The expression $C_\H[\bfs\sigma^+,\bfs\sigma^-]$ of course makes sense without any neutrality condition but we should always reconstruct neutrality by adjusting the charge at infinity.
(Recall that there is a 1-to-1 correspondence between divisors on $\wh\C$ satisfying ($\mathrm{NC}_b$) and arbitrary divisors in $\C:$
$\bfs\sigma \mapsto \bfs\sigma|_\C.$)
\end{rmk*}

We leave the proof of the above theorem as a trivial exercise. In the next subsection we state the version of this theorem in the $\D$-uniformization and provide its proof. 
For a divisor $\bfs\sigma^+ = \sum \sigma_j^+\cdot z_j$ in $D\cup\pa D,$ we define $\overline{\bfs\sigma^+}= \sum \overline{\sigma_j^+}\cdot z_j.$ 

\begin{egs*} We have
 
\ms\no (a) if $\bfs\sigma^- = \bfs 0,$ then 
$$C_\H[\bfs\sigma^+,\bfs0] = \prod_{j<k}(z_j-z_k)^{\sigma_j^+\sigma_k^+};$$

\no (b) if $\bfs\sigma^- = \overline{\bfs\sigma^+},$ then (up to a phase)
$$C_\H[\bfs\sigma^+,\overline{\bfs\sigma^+}] = \prod_{j<k}\Big|(z_j-z_k)^{\sigma_j^+\sigma_k^+}(z_j-\bar z_k)^{\sigma_j^+\overline{\sigma_k^+}}\Big|^ 2 \prod_{\Im\,z_j>0} (2\,\Im\,z_j)^{|\sigma_j^+|^2};$$
\no (c) if $\bfs\sigma^- = -\,\overline{\bfs\sigma^+},$ then (up to a phase)
$$C_\H[\bfs\sigma^+,-\,\overline{\bfs\sigma^+}] = \prod_{j<k}\Big|(z_j-z_k)^{\sigma_j^+\sigma_k^+}(z_j-\bar z_k)^{-\sigma_j^+\overline{\sigma_k^+}}\Big|^ 2 \prod_{\Im\,z_j>0} (2\,\Im\,z_j)^{-|\sigma_j^+|^2}.$$
(The products are taken over finite $z_j$'s and $z_k$'s.)
\end{egs*}

\subsubsection*{$\D$-uniformization}
In the unit disc $\D,$ we define 
$$C_\D[\bfs\sigma^+,\bfs\sigma^-]=\prod_{j<k}(z_j-z_k)^{\sigma_j^+\sigma_k^+}(\bar z_j-\bar z_k)^{\sigma_j^-\sigma_k^-}\prod_{j,k} (1-z_j\bar z_k)^{\sigma_j^+\sigma_k^-}.$$

\begin{thm} \label{thm: C in D}
Under the neutrality condition $(\mathrm{NC}_b),$
$C_\D[\bfs\sigma^+,\bfs\sigma^-]$ is the value of the differential $C[\bfs\sigma^+,\bfs\sigma^-]$ in the identity chart of $\D.$ 
\end{thm}

\begin{proof}
We identify $\wh\C$ with the Schottky double of $\D.$
Then the corresponding involution $\iota$ is
$$\iota:z\mapsto z^* = \frac1{\bar z}.$$
A double divisor $(\bfs\sigma^+,\bfs\sigma^-)$ corresponds to a divisor $\bfs\sigma$ on $\wh\C;$
$\bfs\sigma =\bfs\sigma^+ + \bfs\sigma^-_*= \sum (\sigma_j^+\cdot z_j+\sigma_j^-\cdot z_j^*).$
If one of the $z_j$'s is 0, we call it $z_0$ (so that $z_0 = 0$). 
Denote
$$\wt C =\{C[\bfs\sigma^+,\bfs\sigma^-]\,\|\,\id_\D\}.$$
We need to show that $\wt C = C_\D.$
By definition~\eqref{eq: def C on D}, 
$$\wt C =\{C[\bfs\sigma]\,\|\,(\id,\bar\iota)\},$$
where we use the chart $\bar\iota: z \mapsto 1/z$ at the nodes $z_j^*$ including $z_0^* = \infty.$
It follows that
$$C:= C_\C[\bfs\sigma] = \wt C \prod_{j>0} (-\bar z_j)^{2\lambda_j^-}.$$
Indeed, $C$ is the value of $C[\bfs\sigma]$ in the identity chart of $\C$ and the chart $\bar\iota$ at infinity, so
$$C = \wt C\prod_{j>0} \big(h'(z_j^*)\big)^{\lambda_j^-},$$
where $h = \bar\iota$ is the transition map. 
Recall the expression for $C_\C:$
$$C = \prod_{j<k}(z_j-z_k)^{\sigma_j^+\sigma_k^+} \prod_{0<j<k}(1/\bar z_j-1/\bar z_k)^{\sigma_j^-\sigma_k^-}\prod_{j,k;k\ne0} (z_j-1/\bar z_k)^{\sigma_j^+\sigma_k^-}.$$
Let us rewrite this as a fraction $C = N/D$ with 
$$N=\prod_{j<k}(z_j-z_k)^{\sigma_j^+\sigma_k^+}\prod_{0<j<k}(\bar z_j-\bar z_k)^{\sigma_j^-\sigma_k^-}\prod_{j,k} (1-z_j\bar z_k)^{\sigma_j^+\sigma_k^-},$$
and 
$$D = \prod_{0<j<k}(-\bar z_j\bar z_k)^{\sigma_j^-\sigma_k^-}\prod_{j,k;k\ne0}(-\bar z_k)^{\sigma_j^+\sigma_k^-}.$$
Comparing $C_\D$ to $N,$ we have 
\begin{align*}
C_\D &= N \prod_{j=0,k>0}(\bar z_j-\bar z_k)^{\sigma_j^-\sigma_k^-} = N\prod_{k>0}(-\bar z_k)^{\sigma_0^-\sigma_k^-}\\
&= CD\prod_{k>0}(-\bar z_k)^{\sigma_0^-\sigma_k^-} = \wt C D \prod_{j>0} (-\bar z_j)^{2\lambda_j^-} \prod_{k>0}(-\bar z_k)^{\sigma_0^-\sigma_k^-}.
\end{align*}
To verify $C_\D = \wt C ,$ it remains to show that
$$\prod_{j>0} (-\bar z_j)^{2\lambda_j^-} \prod_{j>0}(-\bar z_j)^{\sigma_0^-\sigma_j^-} \prod_{0<j<k}(-\bar z_j\bar z_k)^{\sigma_j^-\sigma_k^-}\prod_{j,k;k\ne0}(-\bar z_k)^{\sigma_j^+\sigma_k^-}= 1.$$
The exponent of $\bar z_j\,(j\ne0)$ on the left-hand side is 
\begin{align*}
2\lambda_j^- + \sigma_0^-\sigma_j^- &+ \sum_{k\ne j,0}\sigma_j^-\sigma_k^- + \sum_k \sigma_k^+\sigma_j^- = 2\lambda_j^- + \sigma_j^-\Big(\sum_{k\ne j}\sigma_k^- + \sum_k \sigma_k^+\Big)\\
&=2\lambda_ j^- + \sigma_j^-(2b-\sigma_j^-)=2\lambda_j^--2\lambda_b(\sigma_j^-)=0.
\end{align*}
\end{proof}
 
For the exterior $\Delta$ of the unit disc, we leave it to the readers to check that the Coulomb gas correlation functions in the $\Delta$-uniformization have the same expression as in the $\D$-uniformization but one must disregard the node at infinity (as in the $\H$-uniformization).

\section{Conformal field theory of Gaussian free field} \label{sec: CFT of GFF}

In this section we review a version of conformal field theory with central charge $c = 1$ (i.e., $b=0$) both in a simply-connected domain and on the Riemann sphere implemented in \cite{KM13,KM17} and present their connection in the context of Schottky double construction.
Coulomb gas correlations in the case $b = 0$ are represented as correlations of (formal) multi-vertex fields constructed from the Gaussian free field through Wick's calculus. 
In Sections~\ref{sec: O}~--~\ref{sec: GFF_beta}, we extend multi-vertex fields to the case $b\ne 0$ and interpret them as the OPE exponentials of background charge modifications of bosonic fields. 

\subsection{Bosonic field and its Wick's exponentials}
The chiral bosonic fields are described as the holomorphic part or anti-holomorphic part of the Gaussian free field in the physics literature.

\subsubsec{Gaussian free field}
The Gaussian free field $\Phi$ in a planar domain $D$ with Dirichlet boundary condition is an isometry
$\Phi:\EE(D) \to L^2(\Omega,\P)$
from the Dirichlet energy space $\EE(D)$ such that the image consists of centered Gaussian random variables.
Here $(\Omega,\P)$ is a probability space and $\EE(D)$ is the completion of smooth functions with compact supports in $D$ with respect to the norm
$$\|f\|^2_\EE=\iint 2G(\zeta,z)\,f(\zeta)\,\overline{f(z)}~\dd A(\zeta)\,\dd A(z),$$
where $A$ is the (normalized) area measure and $G\equiv G_D$ is the Dirichlet Green's function for $D.$
In the upper half-plane, we have 
$$G_\H(\zeta,z) = \log\Big|\frac{\zeta-\bar z}{\zeta-z}\Big|,$$
where $G(\zeta,\infty) = 0.$
In the unit disc, we have 
$$G_\D(\zeta,z) = \log\Big|\frac{1-\zeta\bar z}{\zeta-z}\Big|.$$
The same formula holds in the exterior $\Delta$ of the unit disc, but $G_\Delta(\zeta,\infty) = \log|\zeta|.$

The Gaussian free field can be viewed as a Fock space field with the $n$-point correlation function 
$$\E[\Phi(z_1) \cdots\Phi(z_n)] = \sum \prod_k 2G(z_{i_k},z_{j_k}),$$
where the sum is over all partitions of the set $\{1,\cdots,n\}$ into disjoint pairs $\{i_k,j_k\}.$
This correlation function is a unique continuous function on $D^n_{\mathrm{distinct}}$ such that
$$\E[\Phi(f_1)\cdots \Phi(f_n)]=\int f_1(z_1)\cdots f_n(z_n)~\E[\Phi(z_1)\cdots \Phi(z_n)]\,\dd A(z_1)\cdots \dd A(z_n)$$
for all test functions $f_j$ with disjoint supports. 

\subsubsec{Chiral bosonic fields}
We write $J := \pa \Phi$ for the current field. 
The chiral boson
$$\Phi^+(z,z_0) = \int_{z_0}^z J(\zeta)\,\dd\zeta$$
is a well-defined, ``multivalued", path dependent, ``generalized" centered Gaussian.
More precisely, 
$$\Phi^+(z,z_0)=\Big\{\Phi^+(\gamma)=\int_\gamma J(\zeta)\, \dd\zeta\Big\},$$
where $\gamma$ is a curve from $z_0$ to $z.$ 
Then the values of $\Phi^+$ are multivalued correlation functionals in the complement of the curve. 
For example, we have
\begin{equation} \label{eq: G^+}
\E[\Phi^+(z, z_0)\Phi(z_1)]=2(G^+(z,z_1)-G^+(z_0,z_1)),
\end{equation}
where $G^+$ is the \emph{complex} Dirichlet Green's function,
$2G^+(z,z_1)=G(z,z_1)+i\wt G(z,z_1).$
Here $\wt G$ is the harmonic conjugate of $G.$
The multivalued function $z\mapsto G^+(z,z_1)$ is defined up to constants.
Sometimes we work with a uniformization $w:(D,q)\to (\D,0).$ 
In such a case, it is convenient to choose the constant so that 
$$G^+(z,z_1)=\frac12\log\frac{1-w(z)\overline{w(z_1)}}{w(z)-w(z_1)}.$$
We write
$$\Phi^-(z,z_0)=\overline{\Phi^+(z,z_0)}.$$
Then we have 
$$\Phi(z) -\Phi(z_0)=\Phi^+(z,z_0)+\Phi^-(z,z_0).$$

\subsubsec{Wick's calculus}
For centered jointly Gaussian random variables $\xi_{jk},$ $(1\le j\le l, 1\le k \le m_j),$ let $X_j =\xi_{j1}\odot \cdots \odot \xi_{jm_j}$ be Wick's product of $\xi_{jk}, (1\le k \le m_j).$
Then by Wick's formula we have 
\begin{equation}\label{eq: tensor product} 
X_1\cdots X_l = \sum_\gamma\prod_{\{v,v'\}}\E[\xi_{v}\xi_{v'}]~\underset{v''}{\textstyle\bigodot} \xi_{v''},
\end{equation}
where the sum is taken over all Feynman diagrams labeled by $\xi_{jk}$'s without edges joining any $\xi_{jk_1}$ and $\xi_{jk_2}.$
Recall that the Feynman diagram labeled by $\xi_1,\cdots,\xi_n$ is a graph with vertices $1,\cdots,n$ such that edges $\{v,v'\}$ have no common endpoints.
(Such edges $\{v,v'\}$ are called ``Wick's contractions.")
We denote the unpaired vertices by $v''.$
For generalized Gaussians or Fock space fields like $\xi_{jk} = \Phi(z_{jk})$ or $\Phi^+(\zeta_{jk},z_{jk}),$ we set the tensor product $X_1\cdots X_l$ with disjoint sets $S(X_j)$ to be \eqref{eq: tensor product}.  
For example,
$$\Phi^{\odot2}(\zeta)\Phi^{\odot2}(z) = \Phi^{\odot2}(\zeta)\odot \Phi^{\odot2}(z) + 4\,\E[\Phi(\zeta)\Phi(z)]  \Phi(\zeta)\odot \Phi(z) + 2(\E[\Phi(\zeta)\Phi(z)])^2,$$
where each of terms on the right hand-side comes from $0, 1,$ or 2 contractions, respectively.  
Wick's exponentials of (generalized) Gaussian $\xi$ is defined by 
$$\ee^{\odot\xi} := \sum_{n=0}^\infty \frac{\xi^{\odot n}}{n!}.$$
For centered jointly (generalized) Gaussians $\xi_j$'s we have 
$$\ee^{\odot \alpha_1\xi_1} \cdots \ee^{\odot \alpha_n\xi_n} = \ee^{\sum_{j<k}\alpha_j\alpha_k\E[\xi_j\xi_k]} \ee^{\odot \sum_j \alpha_j\xi_j}.$$

\subsection{Formal bosonic fields} \label{ss: formal bosonic} 
To make the computation easy, it is convenient to consider a representation
$$\Phi^+(z,z_0) = \Phi^+(z)-\Phi^+(z_0)$$
such that 
$$\Phi=\Phi^++\Phi^-, \qquad \Phi^-=\overline{\Phi^+}.$$
It is not possible to define such $\Phi^\pm$ in a conformally invariant way, but linear combinations
$$\Phi[\bfs\sigma^+,\bfs\sigma^-]:=\sum\sigma_j^+\Phi^+(z_j)-\sigma_j^-\Phi^-(z_j)$$
satisfying the neutrality condition are well-defined as Fock space fields.

\begin{lem}
If a double divisor $(\bfs\sigma^+,\bfs\sigma^-)$ satisfies $(\mathrm{NC}_0)$, i.e.,
$$\int\bfs\sigma^++\bfs\sigma^- =0,$$
then the formal bosonic field $\Phi[\bfs\sigma^+,\bfs\sigma^-]$ can be represented as a linear combination of well-defined Fock space fields.
\end{lem} 

\begin{proof}
Let us choose any point $z_0\in D$. (It can be one of $z_j$'s.) 
Then
$$
 \Phi[\bfs\sigma^+,\bfs\sigma^-]=\Phi^+(z_0)\int\bfs\sigma^+-\Phi^-(z_0)\int\bfs\sigma^-+ \sum\sigma_j^+\Phi^+(z_j,z_0)-\sigma_j^-\Phi^-(z_j,z_0).
$$
Under the neutrality condition, the first two terms on the right-hand side become the Fock space correlation functional: 
$$\Phi^+(z_0)\int\bfs\sigma^+-\Phi^-(z_0)\int\bfs\sigma^- = \Phi(z_0)\int\bfs\sigma^+.$$ 
\end{proof}

\begin{rmk*}
The representation in the lemma is not unique, of course, but it is ``unique" in the sense of ``multivalued" fields.
For example,
$\Phi[1\cdot z_1-1\cdot z_2,\bfs0]=\Phi^+(z_1,z_0)-\Phi^+(z_2,z_0)$
is ``independent" of $z_0.$ 
If we specify the curves in the definition of bi-vertex fields with two different choices of $z_0,$ then the difference is an integral over a loop.
If we choose $z_0 =z_2,$ then $\Phi[1\cdot z_1-1\cdot z_2,\bfs0]=\Phi^+(z_1,z_2)-\Phi^+(z_2,z_2),$ where $z\mapsto\Phi^+(z,z)$ is a ``monodromy field:"
$$\Phi^+(z,z) = \Big\{\int_\gamma J(\zeta)\,\dd\zeta\,:\, \gamma \textrm{ is a loop rooted at } z\Big\}.$$ 
\end{rmk*}

\subsection{Correlations of formal fields in the unit disc and the upper half-plane}
We define $\Phi^\pm \equiv \Phi^\pm_\D$ $(\Phi^- = \overline{\Phi^+})$ in $\D$ as centered Gaussian (formal) fields with ``formal" correlations
\begin{equation}\label{eq: formal E}
\E[\Phi^+(z_1)\Phi^+(z_2)] = \log\frac1{z_1-z_2},\qquad \E[\Phi^+(z_1)\Phi^-(z_2)] = \log(1-z_1\bar z_2).
\end{equation}
Note that these formal correlations have no $\Aut(\D)$-invariance and depend on the order of particles.
It is easy to verify that
$$\E[\Phi(z_1)\Phi(z_2)] = \E [\Phi^+(z_1)+\Phi^-(z_1)][\Phi^+(z_2)+\Phi^-(z_2)],$$
where $\E$ on the left-hand side has the usual meaning, but we use ``formal" correlations on the right-hand side.
The same interpretation is applied for $\E[\Phi^+(z_1,z_2)\Phi^+(z_3,z_4)]$, etc.
If a double divisor $(\bfs\sigma^+,\bfs\sigma^-)$ in $D$ satisfies the neutrality condition $(\mathrm{NC}_0),$ then we can compute correlations of $\Phi[\bfs\sigma^+,\bfs\sigma^-]$ (or other functionals involving $\Phi[\bfs\sigma^+,\bfs\sigma^-],$ e.g., Wick's exponentials of $\Phi[\bfs\sigma^+,\bfs\sigma^-]$) with various Fock space fields by applying Wick's calculus to our formal fields in $\D.$
In the case of an arbitrary simply-connected domain $D,$ we fix a conformal map $w:D \to \D$ and define $\Phi^\pm(z) =\Phi^\pm_\D(w(z))$ so the correlations are $$\E[\Phi^+(z_1)\Phi^+(z_2)] = \log\frac1{w(z_1)-w(z_2)},\quad \E[\Phi^+(z_1)\Phi^-(z_2)] = \log(1-w(z_1)\overline{w(z_2)}).$$
The ``formal" correlations depend on the choice of the conformal map but this dependence disappears under the neutrality condition.
In particular, we can use this method to introduce formal bosonic fields in $\H$ but it is more convenient to define $\Phi^\pm \equiv \Phi^\pm_\H$ in $\H$ as follows:
$$\E[\Phi^+(z_1)\Phi^+(z_2)] = \log\frac1{z_1-z_2},\qquad \E[\Phi^+(z_1)\Phi^-(z_2)] = \log(z_1-\bar z_2)$$
for finite $z_j$'s and set $\Phi^+(\infty) = 0.$

\subsection{OPE exponentials and multi-vertex fields} \label{ss: V} 
The OPE product $X*Y$ of two fields $X$ and $Y$ is a generic notation for a coefficient in the operator product expansion $X(\zeta)Y(z)$ as $\zeta\to z,$ i.e., expansion with respect to a chart independent asymptotic scale.
Typically (but not always) we use $*$ for the coefficient of the first non-diverging term. 
So, the OPE product $X*Y$ of non-chiral fields is obtained by subtracting all divergent terms in the operator product expansion $X(\zeta)Y(z)$ and then taking the limit as $\zeta\to z.$
This is the case in notation $\Phi^{*2}:=\Phi*\Phi, \Phi^{*n} = \Phi*\Phi^{*n-1}$ in the definition of OPE powers 
of $\Phi.$
Let $c(z), z\in D$ denote the logarithm of conformal radius of $D,$ i.e.,
$$c(z) = u(z,z), \quad u(\zeta,z) := G_D(\zeta,z) + \log|\zeta-z|.$$ 
Then by Wick's calculus,
$$\Phi(\zeta)\Phi(z) = \log\frac1{|\zeta-z|^2} + \Phi^{*2}(z) + o(1), \quad \Phi^{*2} = 2c + \Phi^{\odot2}$$
as $\zeta\to z,\zeta\ne z.$
In general, unlike Wick's multiplication, the OPE multiplication is neither commutative nor associative. 
For example, if $f$ is a non-random holomorphic function and if $X,Y$ are holomorphic fields, then 
$f*X = X*f = fX,$ and $(fX)*Y \ne X*(fY).$

We define OPE exponentials of the Gaussian free field by
$$\VV^{(\sigma)}:=\ee^{*i\sigma\Phi} = \sum_{n=0}^\infty \frac{(i\sigma)^n}{n!}\,\Phi^{*n}.$$
Then we have 
$$\VV^{(\sigma)}(z) = C(z)^{-\sigma^2} \ee^{\odot i\sigma\Phi(z)}, \qquad C(z) = \ee^{c(z)},$$
see \cite[Proposition~3.3]{KM13}.
More generally, OPE exponentials (non-chiral multi-vertex fields) $\VV[\bfs\sigma]$ of 
$i\Phi[\bfs\sigma,-\bfs\sigma] = \sum i\sigma_j \Phi(z_j)$ can be defined in a similar way 
$$\VV[\bfs\sigma] := \ee^{*i\Phi[\bfs\sigma,-\bfs\sigma]}$$
and be computed as 
\begin{equation*} \label{eq: VV}
\VV[\bfs\sigma] = \mathcal{C}[\bfs\sigma]\ee^{\odot i\Phi[\bfs\sigma,-\bfs\sigma]},
\end{equation*}
where $\mathcal{C}[\bfs\sigma] = C[\bfs\sigma,-\bfs\sigma].$
See \cite[Section~5]{KM17} for more details. 

\subsubsec{Multi-vertex fields} 
We now extend the OPE exponentials or multi-vertex fields to chiral fields by stating the definition of the chiral multi-vertex field in terms of Coulomb gas correlation function and Wick's exponential. 
Suppose a double divisor $(\bfs\sigma^+,\bfs\sigma^-)$ in $D$ satisfies the neutrality condition $(\mathrm{NC}_0).$
Then the Wick's exponential
$$V^\odot[\bfs\sigma^+,\bfs\sigma^-]:=\ee^{\odot i \Phi[\bfs\sigma^+,\bfs\sigma^-]} \equiv \ee^{\odot i \sum\sigma_j^+\Phi^+(z_j)-\sigma_j^-\Phi^-(z_j)}$$
is a well-defined Fock space functional; $V^\odot[\bfs\sigma^+,\bfs\sigma^-]$ is a scalar, i.e., $(0,0)$-differential and
$$\E\,V^\odot[\bfs\sigma^+,\bfs\sigma^-]=1.$$
We define the multi-vertex field $V[\bfs\sigma^+,\bfs\sigma^-]$ by 
$$V[\bfs\sigma^+,\bfs\sigma^-] = C[\bfs\sigma^+,\bfs\sigma^-] V^\odot[\bfs\sigma^+,\bfs\sigma^-].$$
Thus a multi-vertex field consists of two parts -- its expectation given by
the correlation differential $C[\bfs\sigma^+,\bfs\sigma^-],$ and the Wick exponential $V^\odot[\bfs\sigma^+,\bfs\sigma^-].$
The conformal dimensions of $V[\bfs\sigma^+,\bfs\sigma^-]$ at $z_j$'s are 
$$\lambda_j^+ = \frac{(\sigma_j^+)^2}2, \qquad \lambda_j^- = \frac{(\sigma_j^-)^2}2.$$
(Recall that $b = 0$ in this section.)

\subsubsec{Formal representation of multi-vertex fields} 
For a given conformal map $w:D\to\D,$ we formally set
$$V^{\sigma}(z) = (w'(z))^{\sigma^2/2} \ee^{\odot i\sigma\Phi^+(z)}, \qquad \bar V^{\sigma}(z) = (\overline{w'(z)})^{\sigma^2/2} \ee^{\odot -i\sigma\Phi^-(z)}.$$
Then under the neutrality condition $(\mathrm{NC}_0),$
$$\E\,V[\bfs\sigma^+,\bfs\sigma^-] = \E\, V^{\sigma_1^+}(z_1) \bar V^{\sigma_1^-}(z_1) \cdots V^{\sigma_n^+}(z_n) \bar V^{\sigma_n^-}(z_n),$$
where the left-hand side is $C[\bfs\sigma^+,\bfs\sigma^-]$ by the definition above and the right-hand side is computed by Wick's calculus from formal correlations in $\D.$
The result does not depend on the choice of the conformal map $w,$ and is the same if we use formal correlations in $\H$ or $\Delta.$

\subsection{Conformal field theory on the Riemann sphere} \label{ss: CFT on C}
In this subsection we define the Gaussian free field on $\wh\C$ both as a Gaussian field indexed by the energy space and as a bi-variant Fock space field. 
After introducing formal bosonic fields, we define multi-vertex fields on $\wh\C.$ 

\subsubsec{Energy space}
For a compact Riemann surface $S$ of genus zero, we denote 
$\mathcal{W}(S) \equiv \EE^{(0)}(S) = W^{1,2}(S)/\C$
with the scalar product
$$(f,g)_\nabla = \int_S \dd f \wedge \overline{\dd g}.$$
The \emph{Dirichlet energy space}, $\EE(S) = \EE^{(1,1)}(S)$ is the Hilbert space defined as the completion of the smooth $(1,1)$ forms with respect to the scalar product such that the Laplacian operator 
$$\pa\bp:\, \mathcal{W}(S)\to\EE(S)$$
is unitary. 
Note that if $\mu\in\EE(S)$, then by Green's theorem $\mu$ satisfies the neutrality condition $(\NC_0):$
\begin{equation}
\int\mu=0. 
\end{equation}
Furthermore, $\|\mu\|_\EE^2$ is represented as 
$$\|\mu\|_\EE^2 = \iint_{\C\times\C}\log\frac1{|z-w|^2}\,\mu(z)\overline{\mu(w)}$$
for \emph{any} $\wh\C$-uniformization of $S.$ 
The $\Aut(\wh\C)$-invariance of the right-hand side means that 
$$\iint_{\C\times\C}\log\frac1{|\tau z-\tau w|^2}\,\mu(z)\overline{\mu(w)} $$
is independent of $\tau\in\Aut(\wh\C).$ 
Since translation invariance is trivial, it suffices to show this for $\tau(z) = az\, (a\ne 0)$ and $\tau(z) = 1/z.$ 
It follows from the neutrality condition $(\NC_0)$ on $\mu.$

\subsubsec{Gaussian free field}
The Gaussian free field $\Psi$ on $S$ is a (real, centered) Gaussian field indexed by $\EE(S)$, i.e.,
$$\Psi:\,\EE(S)\to L^2(\Omega,\P)$$
is an isometry such that $\Psi(\mu)$ is a (real, centered) Gaussian random variable for each $\mu\in\EE(S).$ 
We introduce the Fock space functionals (``generalized" Gaussians)
$$\Psi(z,z_0) = \Psi(\delta_z-\delta_{z_0}).$$
(Note that $\mu = \delta_z-\delta_{z_0}$ is not in $\EE(S)$ but as a signed/complex measure it can be approximated by $\mu_n$'s in $\EE$ because it satisfies $(\NC_0);$ here $\delta$'s are $\delta$-\emph{functions}, i.e., $(1,1)$ forms.)
Then the Fock space functionals $\Psi(z,z_0)$ have the following properties:
\renewcommand{\theenumi}{\alph{enumi}}
\begin{enumerate}
\ss\item $\Psi(z,z_0) = -\Psi(z_0,z),$ in particular $\Psi(z,z)=0;$
\ss\item $\Psi(z_1,z_3) = \Psi(z_1,z_2)+\Psi(z_2,z_3).$
\end{enumerate}

\ms
For $\mu\in\EE(S),$ somewhat symbolically we have 
$$\Psi(\mu) = \int \Psi(z,z_0) \,\mu(z).$$
The integral is independent of the reference point $z_0$ by (b) and the neutrality condition $(\NC_0).$
Thus we can think of the Gaussian free field as a bi-variant field with 
$$\E\,\Psi(z,z_0)\,\Psi(z',z_0') = 2\log\bigg|\frac{(z-z_0')(z_0-z')}{(z-z')(z_0-z_0')}\bigg|$$
on $\wh\C.$
As a scalar, conformal invariance is manifest in the cross-ratio on the right-hand side.
As in a simply-connected domain, we define the chiral boson $\Psi^+$ as a bi-variant multivalued field by 
$$\Psi^+(z,z_0) = \int_{z_0}^z \pa_\zeta\Psi(\zeta,z_1)\,\dd\zeta.$$
It is obvious that $\Psi^+(z,z_0)$ does not depend on $z_1.$ 
On $\wh\C,$ we have 
$$\E\,\Psi^+(z,z_0)\Psi^+(z',z_0') = \log \frac{(z-z_0')(z_0-z')}{(z-z')(z_0-z_0')}.$$

\subsubsec{Formal 1-point field}
We introduce $\Psi$ on $\wh\C$ as a centered Gaussian formal field with formal correlation
$$\E\,\Psi(z_1)\,\Psi(z_2) = 2\log\frac1{|z_1-z_2|}.$$
For a compact Riemann surface $S$ of genus zero, we fix a uniformization $w:S\to \wh\C$ and define $\Psi(z) = \Psi_{\wh\C}(w(z))$ so that 
$$\E\,\Psi(z_1)\,\Psi(z_2) = 2\log\frac1{|w(z_1)-w(z_2)|}. $$
The dependence of the formal correlations on the choice of $w$ disappears if we apply this formalism only to the linear combinations satisfying the neutrality condition $(\NC_0):$
$$\sum_j\sigma_j\Psi(z_j),\qquad \sum\sigma_j = 0.$$
For example, we have a representation
$\Psi(z,z_0) = \Psi(z) - \Psi(z_0).$
Next we introduce the formal bosonic fields $\Psi^\pm\equiv\Psi^\pm_{\wh\C}$ on $\wh\C$ as centered Gaussian formal fields satisfying $\Psi = \Psi^++\Psi^-,\Psi^- = \overline{\Psi^+}$ and 
$$\E\,\Psi^+(z_1)\,\Psi^+(z_2)= \log\frac1{z_1-z_2},\qquad \E\,\Psi^+(z_1)\,\Psi^-(z_2) =0.$$
For a given uniformizing map $w:S\to\wh\C,$ we define $\Psi^\pm(z) \equiv\Psi^\pm_S(z) := \Psi^\pm_{\wh\C}(w(z)).$
If both $\bfs\sigma^+ = \sum \sigma^+_j\cdot z_j$ and $\bfs\sigma^- = \sum \sigma^-_j\cdot z_j $ satisfy the neutrality condition $(\NC_0),$ then 
the linear combinations 
$$\Psi^+[\bfs\sigma^+]:=\sum \sigma^+_j\Psi^+(z_j),\qquad \Psi^-[\bfs\sigma^-]:=\sum \sigma^-_j\Psi^-(z_j),$$
and 
$$\Psi[\bfs\sigma^+,\bfs\sigma^-]:= \Psi^+[\bfs\sigma^+] - \Psi^-[\bfs\sigma^-] = \sum \sigma^+_j\Psi^+(z_j) - \sigma^-_j\Psi^-(z_j)$$ 
are well-defined Fock space fields.
For a divisor $\bfs\sigma^- = \sum \sigma^-_j\cdot z_j$ satisfying the neutrality condition $(\NC_0),$ we have 
$$ \Psi^-[\bfs\sigma^-] = \overline{\Psi^+[\overline{\bfs\sigma^-}]},$$
where $\overline{\bfs\sigma} = \sum \overline{\sigma_j}\cdot z_j$ for $\bfs\sigma= \sum \sigma_j\cdot z_j.$

\subsubsec{Vertex fields}
Suppose that both $\bfs\sigma^+$ and $\bfs\sigma^-$ satisfy the neutrality condition $(\NC_0).$ 
We define the multi-vertex field $V[\bfs\sigma^+,\bfs\sigma^-]$ by
$$V[\bfs\sigma^+,\bfs\sigma^-] = C[\bfs\sigma^+,\bfs\sigma^-] V^\odot[\bfs\sigma^+,\bfs\sigma^-], \qquad V^\odot[\bfs\sigma^+,\bfs\sigma^-]:=\ee^{\odot i\Psi[\bfs\sigma^+,\bfs\sigma^-]},$$
where the Coulomb gas correlation function $C[\bfs\sigma^+,\bfs\sigma^-]\equiv C_S[\bfs\sigma^+,\bfs\sigma^-]$ is given by 
$$C[\bfs\sigma^+,\bfs\sigma^-]=C[\bfs\sigma^+]\,\overline{C[\overline{\bfs\sigma^-}]}.$$
We remark that $\bfs\sigma^+$ does not interact with $\bfs\sigma^-$ in $C[\bfs\sigma^+,\bfs\sigma^-]$ due to the independence of $\Psi^+$ and $\Psi^-$:
$$ \E\,\Psi^+(z_1)\,\Psi^-(z_2) =0.$$
In the identity chart of $\C$ and the chart $z\mapsto -1/z$ at infinity, the value of the differential $C[\bfs\sigma^+,\bfs\sigma^-]$ is 
$$\prod_{j<k}(z_j-z_k)^{\sigma_j^+\sigma_k^+} (\bar z_j-\bar z_k)^{\sigma_j^-\sigma_k^-},$$
where the product is taken over finite $z_j$'s and $z_k$'s, and as usual $0^0:=1.$

\begin{rmk*} Recall that we may assume that $\sigma_j^- = 0 $ if $z_j \in \pa D.$
There is a 1-1 correspondence between a divisor $\bfs\sigma$ on $S = D^\mathrm{double}$ and a double divisor $(\bfs\sigma^+,\bfs\sigma^-)$ with $\supp\,\bfs\sigma^+\subseteq D\cap \pa D,$ $\supp\,\bfs\sigma^-\subseteq D:$ 
$$ \bfs\sigma=\bfs\sigma^++\bfs\sigma_*^-.$$
Sometimes it is convenient to write $\Phi[\bfs\sigma]$ for $\Phi[\bfs\sigma^+,\bfs\sigma^-],$ 
$C[\bfs\sigma]$ for $C[\bfs\sigma^+,\bfs\sigma^-],$ and $V[\bfs\sigma]$ for $V[\bfs\sigma^+,\bfs\sigma^-],$ etc. 
\end{rmk*}

\section{Modified multi-vertex fields} \label{sec: O}
In this section we extend the concept of multi-vertex fields to the case $b \ne 0.$
Their correlation functions are defined in terms of correlation functions with the neutrality condition $(\NC_b)$ and their Wick's parts are Wick's exponentials of bosonic fields with charges $\bfs\tau$ satisfying the neutrality condition $(\NC_0).$
To reconcile these two neutrality conditions, background charges $\bfs\beta$ with the neutrality condition $(\NC_b)$ are placed.
In the next section we view the modified multi-vertex fields $\OO_{\bfs\beta}[\bfs\tau]$ as the OPE exponentials of background charge modification of $i\Phi[\bfs\tau].$

\subsection{Definition} 
Let us fix the (background charge) parameter $b\in\R.$ 
This parameter $b$ is related to the central charge $c$ in the following way:
$$c = 1 - 12b^2.$$
As we mentioned in the last remark in the previous section, we often use the 1-1 correspondence between a divisor $\bfs\sigma$ on $S = D^\mathrm{double}$ and a double divisor $(\bfs\sigma^+,\bfs\sigma^-)$ with $\supp\,\bfs\sigma^+\subseteq D\cap \pa D,$ $\supp\,\bfs\sigma^-\subseteq D:$ 
$ \bfs\sigma=\bfs\sigma^++\bfs\sigma_*^-.$
Recall that the Coulomb gas correlation functions
$C[\bfs\sigma] \equiv C_{(b)}[\bfs\sigma^+,\bfs\sigma^-]$ 
are well-defined $[\bfs\lambda^+,\bfs\lambda^-]$-differentials under the neutrality condition $(\NC_b):$ 
$$\int \bfs\sigma=2b,$$ 
see Theorem~\ref{C under NCb}, and that Wick's exponentials $V^\odot[\bfs\tau] \equiv V^\odot[\bfs\tau^+,\bfs\tau^-] = \ee^{\odot i\Phi[\bfs\tau^+,\bfs\tau^-]}$ $(\bfs\tau=\bfs\tau^++\bfs\tau_*^-)$ are well-defined Fock space fields under the neutrality condition $(\NC_0):$
$$\int \bfs\tau=0,$$
where $\Phi[\bfs\tau]\equiv \Phi[\bfs\tau^+,\bfs\tau^-]=\sum \tau_j^+\Phi^+(z_j) - \tau_j^-\Phi^-(z_j),$ see Subsection~\ref{ss: V}.

To reconcile these two neutrality conditions $(\NC_b)$ on $C_{(b)}[\bfs\sigma^+,\bfs\sigma^-]$ and $(\NC_0)$ on $V^\odot[\bfs\tau^+,\bfs\tau^-],$
we need at least one marked point in $D \cup \pa D$ to place the ``background charge" there.
Let us consider the case of one marked point and denote by $q$ this marked point. 
It can be one of the nodes of $(\bfs\sigma^+,\bfs\sigma^-).$
The case $q\in D$ is called \emph{radial}, and the case $q\in\pa D$ \emph{chordal}.

\subsubsec{Standard chordal case} For a divisor $\bfs\sigma( = \bfs\sigma^++\bfs\sigma_*^-)$ on $S$ satisfying the neutrality condition $(\mathrm{NC}_b)$ we define $V[\bfs\sigma] \equiv V[\bfs\sigma^+,\bfs\sigma^-]$ by 
$$V[\bfs\sigma] = C_{(b)}[\bfs\sigma]~V^\odot[\bfs\sigma-2b\cdot q].$$
Sometimes we write $V_q$ or $V_{(b,q)}$ to indicate the location of the marked point and/or the background charge. 
Let us emphasize that the expectation of the multi-vertex (the correlation $C_{(b)}[\bfs\sigma]$) does not depend on $q$ but the Wick part does. 
Multi-vertex functionals/fields are differentials with the same conformal dimensions as in $C_{(b)}[\bfs\sigma].$

\subsubsec{Standard radial case}
For a divisor $\bfs\sigma( = \bfs\sigma^++\bfs\sigma_*^-)$ on $S$ satisfying the neutrality condition $(\mathrm{NC}_b)$ we define $V[\bfs\sigma] \equiv V[\bfs\sigma^+,\bfs\sigma^-]$ by 
$$V[\bfs\sigma] = C_{(b)}[\bfs\sigma]~V^\odot[\bfs\sigma-b\cdot q-b\cdot q^*].$$
Again, $V[\bfs\sigma]$ is $\Aut(D,q)$-invariant but its expectation $C_{(b)}[\bfs\sigma]$ is $\Aut(D)$-invariant.
(This is not the only way to satisfy neutrality condition in Wick's part of the functional.)

Both cases can be generalized to several marked points. 
Additional marked points $q_k$ are not necessarily on the boundary $\pa D.$
Suppose we have a background charge $\bfs\beta (= \bfs\beta^++\bfs\beta_*^-)$ on $S$ with the neutrality condition $(\mathrm{NC}_b).$ 
For a divisor $\bfs\sigma( = \bfs\sigma^++\bfs\sigma_*^-)$ on $S$ satisfying the neutrality condition $(\mathrm{NC}_b)$ we define $V_{\bfs\beta}[\bfs\sigma] \equiv V_{\bfs\beta^+,\bfs\beta^-}[\bfs\sigma^+,\bfs\sigma^-]$ by 
$$V_{\bfs\beta}[\bfs\sigma] = C_{(b)}[\bfs\sigma]~V^\odot[\bfs\sigma-\bfs\beta].$$

\subsection{Background charge operators} 
Let us introduce the background charge operators. 
For a background charge $\bfs\beta$ on $S,$ we define the background charge operator $\PP_{\bfs\beta}$ associated with $\bfs\beta$ by
$$\PP_{\bfs\beta} := V_{\bfs\beta}[\bfs\beta] = C_{(b)}[\bfs\beta].$$ 
Note that the Wick part of $V[\bfs\beta]$ is the constant field $1.$

\subsubsec{Standard chordal case} 
We write $\PP_q$ for $\PP_{2b\cdot q}= C_{(b)}[2b\cdot q].$
Then we would have $\PP_q = 1$ in all charts because
$$\lambda_q^+ =\lambda_b(2b) = 0$$
and clearly $\PP_q = 1$ in the $(\H,\infty)$-uniformization.

\subsubsec{Standard radial case} 
We write $\PP_q$ for $\PP_{b\cdot q+b\cdot q^*}= C_{(b)}[b\cdot q+b\cdot q^*].$
This is a non-random differential with dimensions
$$\lambda_q^+ = \lambda_q^- =\lambda_b(b) = -\frac{b^2}2$$
at $q.$ 
Note that $\PP_q= 1$ in the $(\D,0)$-uniformization.

\subsection{OPE exponentials} 
We use the background charge operators to modify the multi-vertex fields so that the OPE calculus of the modified multi-vertex fields has a very simple and natural form.
We call the modified multi-vertex fields the OPE exponentials; the reason for this terminology becomes apparent in the next section.

Suppose $\bfs\tau (= \bfs\tau^++\bfs\tau_*^-)$ is a divisor on $S$ satisfying the neutrality
condition $(\mathrm{NC}_0).$
For this divisor $\bfs\tau$ and a background charge $\bfs\beta(=\bfs\beta^++\bfs\beta_*^-)$ on $S$ (satisfying the neutrality condition $(\NC_b)$), we define the OPE exponentials (the modified multi-vertex fields) $\OO_{\bfs\beta}[\bfs\tau]\equiv \OO_{\bfs\beta^+,\bfs\beta^-}[\bfs\tau^+,\bfs\tau^-] $ by 
$$\OO_{\bfs\beta}[\bfs\tau]:= \PP_{\bfs\beta}^{-1}V_{\bfs\beta}[\bfs\tau + \bfs\beta]
= \frac{C_{(b)}[\bfs\tau + \bfs\beta]}{C_{(b)}[\bfs\beta]}~V^\odot [\bfs\tau].$$
We remark that all charges in the correlation
$$\E\,\OO_{\bfs\beta}[\bfs\tau] = \frac{C_{(b)}[\bfs\tau + \bfs\beta]}{C_{(b)}[\bfs\beta]}$$
interact except that background charges do not interact with each other.

\subsubsec{Standard radial case} 
If $\bfs\beta = b\cdot q + b\cdot q^*,$ then
$$\OO[\bfs\tau]\equiv\OO_{\bfs\beta}[\bfs\tau] = \frac{C_{(b)}[\bfs\tau + b\cdot q + b\cdot q^*]}{C_{(b)}[b\cdot q+b\cdot q^*]}~V^\odot [\bfs\tau].$$
We call $( \bfs\tau^+,\bfs\tau^-)$ the (double) divisor of Wick's exponents and $(\tau_q^++b,\tau_q^-+b)$ the \emph{effective} charges at $q.$ 
Thus the sum of exponents is zero but the sum of (effective) charges is $2b,$ i.e., $\bfs\tau$ satisfies $(\NC_{0})$ and $\bfs\tau+\bfs\beta$ satisfies $(\NC_{b}).$
Sometimes we omit the subscript $\bfs\beta$ in $\OO_{\bfs\beta}[\bfs\tau]$ when $\bfs\beta = b\cdot q+b\cdot q^*.$ 

\begin{prop} \label{prop: hq}
The conformal dimensions of $\,\OO[\bfs\tau]$ at $q$ are
$$h_q^\pm = \frac{(\tau_q^\pm)^2}2.$$
\end{prop}

\begin{proof}
The Coulomb gas correlations $C_{(b)}[\bfs\tau^+ + b\cdot q,\bfs\tau^- + b\cdot q], C_{(b)}[b\cdot q,b\cdot q]$ have conformal dimensions $[\lambda_b(\tau_q^++b),\lambda_b(\tau_q^-+b)],$ $[\lambda_b(b),\lambda_b(b)],$ respectively at $q.$ 
Since Wick's exponentials $V^\odot [\bfs\tau^+,\bfs\tau^-]$ are scalars or $(0,0)$-differentials, we have 
$$h_q^\pm = \lambda_b(\tau_q^\pm+b)- \lambda_b(b) = \frac{(\tau_q^\pm)^2}2.$$
\end{proof} 

We call the numbers $\lambda_b(\tau_q^\pm + b)$ \emph{effective} conformal dimensions at $q$; they are dimensions of $\PP_q\,\OO[\bfs\tau] = V[\bfs\tau+\bfs\beta].$
As we have just seen,
$$\lambda_b(\tau_q^\pm+b) = h_q^\pm - \frac{b^2}2.$$

\begin{prop} \label{EO}
In the $(\D,0)$-uniformization, we have
$$\E\,\OO[\bfs\tau] = \prod_j (z_j)^{\nu_j^+}(\bar z_j)^{\nu_j^-} \prod_{j<k}(z_j-z_k)^{\tau_j^+\tau_k^+}(\bar z_j-\bar z_k)^{\tau_j^-\tau_k^-}\prod_{j,k} (1-z_j\bar z_k)^{\tau_j^+\tau_k^-},$$
where $\nu_j^\pm = \tau_j^\pm(\tau_q^\pm+b)$ and $z_j$'s are non-zero nodes.
\end{prop}

\begin{proof}
Since $\PP_q= 1$ in the $(\D,0)$-uniformization and $\E\,V^\odot [\bfs\tau]=1,$
we have 
$$\E\,\OO[\bfs\tau] = C_{(b)}[\bfs\tau + b\cdot q+ b\cdot q^*].$$
Proposition now follows from Theorem~\ref{thm: C in D}.
We have $\E\,\OO[\bfs\tau] = C_\D[\bfs\tau^+ + b\cdot q,\bfs\tau^- + b\cdot q]$ in the identity chart of $\D.$
\end{proof}

In other words, $\E\,\OO[\bfs\tau]$ in $(\D,0)$ is given by the usual Coulomb gas correlation function for effective charges.

\subsubsec{Standard chordal case} 
If $\bfs\beta = 2b\cdot q,$ then
$$\OO[\bfs\tau]\equiv\OO_{\bfs\beta}[\bfs\tau] := V[\bfs\tau + 2b\cdot q] = C_{(b)}[\bfs\tau + 2b\cdot q ]\,V^\odot[\bfs\tau].$$
Recall that $\PP_q = 1$ in all charts. 
The (effective) charge at $q$ is $\tau_q^++ 2b,$ so the sum of (effective) charges is $2b.$
Sometimes we omit subscript $\bfs\beta$ in $\OO_{\bfs\beta}[\bfs\tau]$ when $\bfs\beta =2b\cdot q.$ 
The conformal dimension of $\OO[\bfs\tau] $ at $q$ is
$$h_q^+ = \lambda_b(\tau_q^++ 2b) = \lambda_b(-\tau_q^+) = \frac{(\tau_q^+)^2}2+\tau_q^+ b.$$

\begin{rmk*}
In the standard chordal case, 
there is a 1-to-1 correspondence between the double divisors $(\bfs\sigma^+,\bfs\sigma^-)$ in $\overbar{D} \setminus\{q\}$ and the double divisors $( \bfs\tau^+,\bfs\tau^-)$ in $\overbar{D}$ satisfying $(\mathrm{NC}_0)$ with $\tau_q^-=0:$
$$\bfs\sigma^+ = \bfs\tau^+ - \tau_q^+\cdot q, \quad \bfs\sigma^-= \bfs\tau^-, \quad \tau_q^+ = -\int (\bfs\sigma^++\bfs\sigma^-).$$ 
In \cite{KM13} we use the notation 
$$\Sigma = \int (\bfs\sigma^++\bfs\sigma^-),$$
so $\Sigma = -\tau_q^+$ (by $(\mathrm{NC}_0)$) and $h_q =\lambda_b(\Sigma).$
Also, in \cite{KM13} we use the notation
$$\OO^{(\sigma_1^+,\sigma_1^-)}(z_1) \star \cdots \star\OO^{(\sigma_n^+,\sigma_n^-)}(z_n)$$
(which can be shortened to $\OO^{(\bfs\sigma^+,\bfs\sigma^-)}$) instead of our present notation $\OO[\bfs\tau^+,\bfs\tau^-]$ for the OPE exponentials.
\end{rmk*}

\subsection{Example: one-leg operators} \label{ss: Example: one-leg operators}
As a special case of Theorem~\ref{main}, under the insertion of Wick's part $\leg_p/\E\,\leg_p$ of the radial one-leg operator $\leg_p$ (rooted at $q$) with $p\in\pa D$ all correlation functions of the fields in the extended OPE family of $\Phi_{\bfs\beta}$ ($\bfs\beta=b\cdot q+b\cdot q^*$) are radial $\SLE(\kappa)$ martingale-observables (see Subsection~\ref{ss: Phib} for the modified Gaussian free $\Phi_{\bfs\beta}$ and Subsection~\ref{ss: extended OPE family} for its extended OPE family) if the parameters $a$ (the charge of $\leg_p$ at $p$) and $b$ are related to the SLE parameter $\kappa$ as 
$$a = \pm\sqrt{2/\kappa},\qquad b = a(\kappa/4-1).$$

\no (a) We define the \emph{radial one-leg operator} $\leg$ (rooted at $q$) by 
$$\leg_z\equiv\leg(z):=\OO\big[a\cdot z - \frac{a}2\cdot q- \frac{a}2\cdot q^*\big].$$
Its conformal dimensions $\lambda_z$ at $z$ and $h_q^\pm$ at $q$ are
\begin{equation} \label{eq: h12}
\lambda_z = h \equiv h_{1,2} := \frac{a^2}2-ab = \frac{6-\kappa}{2\kappa} 
\end{equation}
and 
$$h_q^+ = h_q^- = \frac{a^2}8, \qquad H_q (:= h_q^+ + h_q^-) = \frac{a^2}4.$$
The effective charges at $q$ are $(b-a/2,b-a/2)$ and the effective dimensions at $q$ are
\begin{equation} \label{eq: h01/2}
H_q^\eff(\leg)=2h_{0,1/2},\quad h_{0,1/2}:=\frac{a^2}8-\frac{b^2}2 = \frac{(\kappa-2)(6-\kappa)}{16\kappa}.
\end{equation}
Indeed, $H_q^\eff(\leg)=H_q(\leg^\eff),$ where $\leg^\eff$ is the ``effective" one-leg operator 
$$\leg^\eff=\PP_q\,\leg = V\big[a\cdot z +\big(b- \frac{a}2\big)\cdot q+\big(b- \frac{a}2\big)\cdot q^*\big].$$
In $(\D,0)$ we have 
$\E\,\leg(z) =\E\,\leg^\eff(z) = z^{a(b-a/2)} = z^{-h},$
so in $(D,q)$
$$\E\,\leg(z) = \Big(\frac{w'(z)}{w(z)}\Big)^{h} |w_q'|^{H_q} ,$$ 
where $w: (D,q)\to(\D,0)$ is a conformal map and $w_q' = w'(q).$ 

\ss \no (b) Let $\eta\in\R$ be the parameter in the definition of radial $\SLE_\eta(\kappa,\bfs\rho),$ see \eqref{eq: d theta}. 
We define the \emph{radial one-leg operator} $\leg_{(s)}$ with spin $s = i\eta a^2/2$ by 
$$\leg_{(s)}(z):=\OO\big[a\cdot z - \frac{a+i\delta}2\cdot q - \frac{a-i\delta}2\cdot q^*\big],$$
where $\delta = \eta a.$
Its conformal dimensions $h_q^\pm$ at $q$ are
$$h_q^\pm = \frac{(a\pm i\delta)^2}8,$$ 
so the spin $s:=h_q^+ - h_q^-$ and the conformal dimension $H_q := h_q^+ + h_q^-$ at $q$ are 
$$s = \frac{ia\delta}2,\qquad H_q = \frac{a^2+\delta^2}4.$$
In $(\D,0)$ we have 
$\E\,\leg_{(s)}(z) =\E\,\leg_{(s)}^\eff(z)= z^{a(b-a/2-i\delta/2)} = z^{-h-s},$ and
$$\kappa \pa_\theta \log \frac{\E\,\leg_{(s)}^\eff(\ee^{i\theta})}{\E\,\leg^\eff(\ee^{i\theta})} = \eta.$$
See \eqref{eq: driving for radial SLE[beta]} and \eqref{eq: d theta}. 
Later we define the partition function in terms of the correlation function of the effective one-leg operator.

\ss \no (c) We define the \emph{chordal one-leg operator} $\leg$ by 
$$\leg(z):=\OO\big[a\cdot z - a\cdot q].$$
Its conformal dimension $\lambda_z$ at $z$ and effective dimension $h_q^\eff$ at $q$ are
$$\lambda_z =h, \qquad h_q^\eff = \lambda_b(2b-a) = \lambda_b(a) = h.$$
In $(\H,\infty)$ we have 
$\E\,\leg^\eff(z) = 1,$
so
$\E\,\leg^\eff(z) = (w')^h(w_q')^h$ in $(D,q).$

\subsection{Algebra of multi-vertex functionals and OPE exponentials}
In the general case, given a background charge $\bfs\beta$ we define the multiplication of OPE exponentials by 
$$\OO_{\bfs\beta}[\bfs\tau_1]\OO_{\bfs\beta}[\bfs\tau_2]=\OO_{\bfs\beta}[\bfs\tau_1+\bfs\tau_2].$$
In the next section we relate this operation to OPE/tensor multiplication.
Multiplication of OPE exponentials is commutative and associative (if we ignore the order of particles).
Formally, we have the following representation
$$\OO_{\bfs\beta}[\bfs\tau^+,\bfs\tau^-] = \OO_{\bfs\beta}^{(\tau_1^+)}(z_1)\overline{\OO_{\bfs\beta}^{(\bar\tau_1^-)}(z_1)}\cdots\OO_{\bfs\beta}^{(\tau_n^+)}(z_n)\overline{\OO_{\bfs\beta}^{(\bar\tau_n^-)}(z_n)},$$
where $\bar\tau_j^- = \overline{\tau_j^-}.$

\section{Background charge modifications of Gaussian free field} \label{sec: GFF_beta}
In this section we discuss background charge modifications of the Gaussian free field in a simply-connected domain $D$ with marked boundary/interior points.
A special type of such modifications with a marked boundary point appeared in \cite{RBGW07} and \cite{KM13} to present their connections to chordal SLE theory.
Similar constructions had been well known both in the physics literature and in the algebraic literature. 
For example, see \cite[Chapter~9]{DFMS97} for Coulomb gas formalism, and \cite{KR87} for Fairlie's modifications of Virasoro generators.

The background charge modifications of the radial conformal field theory with a marked interior point $q$ under our consideration equip the Gaussian free field with the additive monodromy around the puncture $q\in D.$
We define the OPE exponentials of modified bosonic fields with nodes in $D^*=D\sm\{q\}$ as differentials in the OPE family.
The OPE exponentials extend to the puncture through operator product expansion with the constant field $1_q$ or the rooting procedure. 
We discuss these extensions in the next section.

\subsection{Modification of Gaussian free field on the Riemann sphere} \label{ss: Psib}

In this subsection we borrow the concepts of the background charge modifications of Gaussian free field on a compact Riemann surface from \cite{KM17}. Later we adapt them in a simply-connected domain employing Schottky double construction.

Recall that a non-random field $\psi$ is called a \emph{pre-pre-Schwarzian form} of order $(\mu,\nu)$ (or $\PPS(\mu,\nu)$) if the transformation law is
$$\psi = \ti\psi\circ h + \mu \log h' + \nu \log \overline{h'},$$ 
where $\psi = (\psi\,\|\,\phi)$ in a chart $\phi,$ $\ti\psi = (\psi\,\|\,\ti\phi)$ in a chart $\ti\phi,$ and $h$ is the transition map between two overlapping charts $\phi,\ti\phi.$
We consider a holomorphic/harmonic PPS form $\psi$ such that $\pa \psi$ is meromorphic and $\bfs\beta\equiv\bfs\beta_\psi:= i/\pi \bp\pa \psi$ is a finite linear combination of $\delta$-measures, $\bfs\beta = \sum_k \beta_k \delta_{q_k}.$ 
Such $\psi$ is called a \emph{simple} PPS form. 
We often write $\bfs\beta_\psi = \sum_k \beta_k \cdot q_k$ as a divisor and call it the \emph{background charge} of $\psi.$
\begin{prop} \label{Gauss-Bonnet}
On a compact Riemann surface $S$ of genus zero, we have the neutrality condition $(\NC_b)$ for the background charge $\bfs\beta_\psi$ of a simple $\PPS(ib,0)$ form $\psi:$ 
\begin{equation}\label{eq: NCb}
\sum_k \beta_k = 2b.
\end{equation}
\end{prop}
It is a consequence of a version of Gauss-Bonnet theorem, e.g., see \cite[Corollary~6.3]{KM17}.
For the reader's convenience, we present its proof.

\begin{proof}[Proof of Proposition~\ref{Gauss-Bonnet}]
Let us choose a conformal metric $\rho = |\omega_q|^2$ (a positive $(1,1)$-differential) on $S,$ where $\omega_q$ is a meromorphic differential with a sole double pole at $q \in S.$ 
In the identity chart of $\C$ with $q=0,$ one can take $\rho(z) = |z|^{-4},$ so
$\bp \pa\log \rho = -2\pi \delta_0$
and 
$$\int_S \bp\pa \log \rho = -2\pi = -\pi \chi(S),$$
where $\chi(S)$ is the Euler characteristic of $S.$
Let $\varphi = \psi_{\bfs\beta}/(ib).$ 
Then $\varphi$ is a simple $\PPS(1,0)$ form. 
We now consider the harmonic $\PPS(1,1)$ form, $\varphi_* = \log \rho = \log |\omega_0|^2.$ 
Proposition now follows from 
\begin{equation} \label{eq: int PPS}
\int_S\bp\pa(\varphi-\varphi_*) = 0
\end{equation}
since
$$\int\bfs\beta = \frac i\pi \int\bp \pa\psi_{\bfs\beta} = \frac i\pi ib \int\bp \pa\varphi = \frac i\pi ib \int\bp \pa\varphi_* = 2b . $$
By means of Green's theorem, the integral in \eqref{eq: int PPS} is the sum of all residues of the meromorphic differential $\pa(\varphi-\varphi_*)/(2i).$ 
It is well known that such a sum vanishes. 
\end{proof}

Given a background charge $\bfs\beta = \sum_k \beta_k \cdot q_k$ with the neutrality condition $(\NC_b)$, there is a unique (up to an additive constant) simple $\PPS(ib,0)$ form $\psi_{\bfs\beta}^+$ with the background charge $\bfs\beta,$
$$\psi_{\bfs\beta}^+ = i\sum_k\beta_k \psi_{q_k}^+,$$
where 
$\psi_{q}^+ = \frac12 \log \omega_q$ and $\omega_q$ is a meromorphic differential with a sole double pole at $q \in S.$ 
On the Riemann sphere $\wh\C,$ $\omega_\infty = 1$ and $\omega_q = 1/(z-q)^2$ in the identity chart of $\C.$ 
In terms of a uniformization $w:S\to\wh\C,$ 
$$\psi_{q_k}^+(z) = 
\begin{cases} \frac12 \log w'(z) - \log(w(z)-w(q_k)),\qquad &w(q_k) \ne \infty \\ \frac12 \log w'(z), &w(q_k) = \infty.\end{cases}$$

We now define the background charge modifications of the formal bosonic fields. 
Given two background charges 
$\bfs\beta^\pm = \sum_k \beta_k^\pm\cdot q_k$ with the neutrality conditions 
$\sum\beta_k^\pm = 2b^\pm,$ we define 
$$\Psi_{\bfs\beta^\pm}^\pm = \Psi^\pm + \psi_{\bfs\beta^\pm}^\pm, \qquad \psi_{\bfs\beta^-}^- = \overline{\psi_{\overbar{\bfs\beta^-}}^+}, \qquad \Psi_{\bfs\beta^+,\bfs\beta^-} = \Psi_{\bfs\beta^+}^+ +\Psi_{\bfs\beta^-}^-, \qquad \psi_{\bfs\beta^+,\bfs\beta^-} = \psi_{\bfs\beta^+}^+ +\psi_{\bfs\beta^-}^-$$
so that $\Psi_{\bfs\beta^+,\bfs\beta^-} = \Psi+ \psi_{\bfs\beta^+,\bfs\beta^-}.$
Then $\Psi_{\bfs\beta^+,\bfs\beta^-}$ is a $\PPS(ib^+,-ib^-)$ form. 
In the $\C$-uniformization, we have 
\begin{align*}
\psi_{\bfs\beta^+,\bfs\beta^-}(z) &= i\sum_{k}\beta_k^+\log\frac1{z-q_k} - i\sum_{k} \beta_k^-\log\frac1{\bar z -\bar q_k}\\
&= i \sum_{k}(\beta_k^+-\beta_k^-)\log\frac1{|z-q_k|} + \sum_{k}(\beta_k^++\beta_k^-)\arg(z-q_k). 
\end{align*}
Here the summation is taken over all $k$'s but $q_k = \infty.$
In particular, if $\bfs\beta^+ = \overline{\bfs\beta^-} \equiv \bfs\beta,$ then 
\begin{equation} \label{eq: psi beta}
\psi_{\bfs\beta,\overline{\bfs\beta}}(z) = -2 \sum_k \Im\,\beta_k \log\frac1{|z-q_k|} + 2\sum_k \Re\, \beta_k \arg(z-q_k)
\end{equation}
in the $\C$-uniformization.

For two divisors $\bfs\tau^\pm = \sum \tau_j^\pm\cdot z_j$ satisfying the neutrality condition $(\NC_0),$ we define 
$$\Psi_{\bfs\beta^+,\bfs\beta^-}[\bfs\tau^+,\bfs\tau^-]:= \Psi_{\bfs\beta^+}^+[\bfs\tau^+]-\Psi_{\bfs\beta^-}^-[\bfs\tau^-],$$
where 
$$\Psi_{\bfs\beta^+}^+[\bfs\tau^+]:=\sum_j \tau_j^+\Psi_{\bfs\beta^+}^+(z_j),\qquad\Psi_{\bfs\beta^-}^-[\bfs\tau^-]:= \sum_j \tau_j^-\Psi_{\bfs\beta^-}^-(z_j).$$
Then we have 
$$\Psi_{\bfs\beta^+,\bfs\beta^-}[\bfs\tau^+,\bfs\tau^-]=\Psi[\bfs\tau^+,\bfs\tau^-] + i\sum_{j,k}\tau_j^+\beta_k^+\log\frac1{z_j-q_k} + i\sum_{j,k} \tau_j^-\beta_k^-\log\frac1{\bar z_j -\bar q_k}$$
in the $\C$-uniformization.

Sometimes it is convenient to write $\bfs\beta = [\bfs\beta^+,\bfs\beta^-],$ $\bfs\tau = [\bfs\tau^+,\bfs\tau^-],$ and 
$$\Psi_{\bfs\beta}[\bfs\tau]=\Psi_{\bfs\beta^+,\bfs\beta^-} [\bfs\tau^+,\bfs\tau^-].$$

\subsection{Stress tensors and Virasoro fields}

For the reader's convenience, we briefly review the definitions of a stress tensor and the Virasoro field. See \cite[Lectures~4 and 5]{KM13} for more details. Suppose $A^+$($A^-,$ respectively) is a Fock space \emph{holomorphic} (\emph{anti-holomorphic}, respectively) quadratic differential. 
Recall that a Fock space field $X$ is (anti-)holomorphic if the correlation $z\mapsto \E\,X(z)\XX$ is (anti-)holomorphic in the complement of $S_\XX$ for any tensor product $\XX$ of the Gaussian free fields.
Let $v$ be a non-random holomorphic vector field defined in some neighborhood of $\zeta.$
We define the residue operator $A_v^+$($A_v^-$, respectively) as an operator on Fock space fields:
$$
(A_v^+X)(z)= \frac1{2\pi i}\oint_{(z)} vA^+\,X(z), \qquad
(A_v^-X)(z)= -\frac1{2\pi i}\oint_{(z)} \bar vA^-\,X(z)
$$
in a given chart $\phi,\;\phi(\zeta)=z.$
A pair $W=(A^+,A^-)$
is called a \emph{stress tensor} for a Fock space field $X$ if for all non-random local vector fields $v,$ the so-called ``residue form of Ward's identity"
\begin{equation} \label{eq: LvWv}
\LL_vX= A^+_vX+ A^-_vX
\end{equation}
holds in the maximal open set $D_\hol (v)$ where $v$ is holomorphic. 
We recall the definition of Lie derivatives (see \cite[Section~3.4]{KM13}):
$$(\LL_v X\,\|\, \phi) = \frac{\dd}{\dd t}\Big|_{t=0} (X\,\|\, \phi\circ\psi_{-t}),$$
where $\psi_t$ is a local flow of $v,$ and $\phi$ is an arbitrary chart. 
The Lie derivative operator $\LL_v$ depends $\R$-linearly on vector fields $v$ and it is convenient to consider its $\C$-linear part $\LL_v^+$ and anti-linear part $\LL_v^-:$ 
$$2\LL_v^+ = \LL_v-i\LL_{iv},\qquad 2\LL_v^- = \LL_v+i\LL_{iv}.$$
The $\C$-linear part and the anti-linear part are related as $\LL_v^- = \overline{\LL_v^+}$ and the conjugation means that 
$\overline{\LL_v^+}X = \overline{\LL_v^+ \bar X}.$
If $X$ is a $(\lambda,\lambda_*)$-differential, then 
\begin{equation} \label{eq: LvDiff}
\LL_v^+X = (v\pa + \lambda v')X, \qquad \LL_v^-X = (\bar v\bp + \lambda_* \overline{v'})X;
\end{equation}
if $X$ is a $\PPS(\mu^+,\mu^-)$ form, then
\begin{equation} \label{eq: LvPPS}
\LL_v^+X = v\pa X + \mu^+ v', \qquad \LL_v^-X = \bar v\bp X + \mu^- \overline{v'};
\end{equation}
if $X$ is a pre-Schwarzian form of order $\mu,$ then
\begin{equation} \label{eq: Lie4PS-form}
\LL_vX=\left(v\pa+v'\right)X +\mu v'';
\end{equation}
if $X$ is a Schwarzian form of order $\mu,$ then
\begin{equation} \label{eq: Lie4S-form}
\LL_vX=\left(v\pa+2v'\right)X +\mu v'''.
\end{equation}

We denote by $\FF(W)\equiv\FF(A^+,A^-)$ \emph{Ward's family} of $W,$ the linear space of all Fock space fields $X$ with a stress tensor $W$ in common. 
If $\FF(A^+,A^-)$ is closed under complex conjugation, we can choose $A^+ = A, A^-=\bar A.$ 
In this symmetric case, the Lie derivatives operators $\LL^\pm_v$ acts on $\FF(A,\bar A)$ as the residue operators $A^\pm_v.$
In this case $X\in\FF(A,\bar A)$ if and only if the following Ward's OPEs hold in every local chart $\phi:$
\begin{equation}\label{eq: Ward's OPEs}
\Sing_{\zeta\to z}[A(\zeta)X(z)]= (\LL_{k_\zeta}^+X)(z), \quad \Sing_{\zeta\to z}[A(\zeta)\bar X(z)]= (\LL_{k_\zeta}^+\bar X)(z),
\end{equation}
where $\Sing_{\zeta\to z}$ means the singular part of the operator product expansion in chart $\phi$ as $\zeta\to z$ and $k_\zeta$ is the local vector field given by 
$$(k_\zeta\|\phi)(\eta) = \frac1{\zeta-\eta}.$$
See \cite[Proposition~5.3]{KM13}.
We often use the notation $\sim$ for the singular part of the operator product expansion.

For example, in a simply-connected domain, the Gaussian free field $\Phi$ (with Dirichlet boundary condition) has a stress tensor 
$$(A,\bar A),\qquad A = -\frac 12 J\odot J, \qquad J = \pa \Phi.$$
It follows from Ward's OPE
$$A(\zeta)\Phi(z) \sim \frac{\pa\Phi(z)}{\zeta-z},\qquad \zeta\to z.$$
It is well known that Ward's family is closed under the OPE product $*$ and differentiations, see \cite[Proposition~5.8]{KM13}.
This fact implies that $(A,\bar A)$ is also a stress tensor for $J$ and $T:=-\frac12 J*J.$ 

On the Riemann sphere, the current fields $J, \bar J$ are defined as well-defined single-variable Fock space fields,
$$J(z) = \pa_z\Psi(z,z_0), \qquad \bar J(z) = \bp_z\Psi(z,z_0).$$
It is easy to see that this definition does not depend on the choice of $z_0.$
Furthermore, $J$ is holomorphic and $\bar J$ is anti-holomorphic. 
As in a simply-connected domain, the Gaussian free field $\Psi$ has a stress tensor 
$$(A,\bar A),\qquad A = -\frac 12 J\odot J.$$
In this case, Ward's OPE reads as
$$A(\zeta)\Psi(z,z_0) \sim \frac{\pa_z\Psi(z,z_0)}{\zeta-z},\qquad \zeta\to z.$$

A pair $(T^+,T^-)$ of Fock space fields is called the Virasoro pair for the family $\FF(A^+,A^-)$ if 
$T^\pm \in \FF(A^+,A^-)$ and if $T^+-A^+, \overline{T^--A^-}$ are non-random holomorphic (or meromorphic with poles where background charges are placed, see Section~\ref{sec: GFF_beta}) Schwarzian forms.
In the symmetric case, $T$ is called the \emph{Virasoro field} for Ward's family $\FF(A,\bar A).$ 
Both in a simply-connected domain and on the Riemann sphere, the Virasoro field $T$ is given by 
$$T = -\frac12 J*J $$
or by the operator product expansion
$$J(\zeta) J(z) = -\frac1{(\zeta-z)^2} -2T(z) + o(1), \qquad \zeta\to z.$$
It is easy to see that $T$ is a Schwarzian form of order $1/12.$ 
See \cite[Proposition~3.4]{KM13} in a simply-connected domain or \cite[Proposition~4.1]{KM17} on a compact Riemann surface.

\begin{thm} \label{SETinS} 
The bosonic field $\Psi_{\bfs\beta}\equiv \Psi_{\bfs\beta^+,\bfs\beta^-}$ has a stress tensor $W_{\bfs\beta}\equiv(A_{\bfs\beta^+}^+,A_{\bfs\beta^-}^-),$
$$A_{\bfs\beta^+}^+ = A + ib^+ \pa J - \pa\psi_{\bfs\beta^+}^+J, \qquad A_{\bfs\beta^-}^- = \overline{A_{\overbar{\bfs\beta^-}}^+ }(=\bar A - ib^- \bp \bar J - \overline{\pa\psi_{\overbar{\bfs\beta^-}}^+} \bar J).$$
The Virasoro pair $(T_{\bfs\beta^+}^+, T_{\bfs\beta^-}^-)$ for the family $\FF(W_{\bfs\beta})$ is given by 
\begin{align*}
T_{\bfs\beta^+}^+ &= -\frac12 (\pa\Psi_{\bfs\beta^+,\bfs\beta^-}*\pa\Psi_{\bfs\beta^+,\bfs\beta^-}) + ib^+\pa^2\Psi_{\bfs\beta^+,\bfs\beta^-}, \\
T_{\bfs\beta^-}^- &= -\frac12 (\bp\Psi_{\bfs\beta^+,\bfs\beta^-}*\bp\Psi_{\bfs\beta^+,\bfs\beta^-}) - ib^-\bp^2\Psi_{\bfs\beta^+,\bfs\beta^-}.
\end{align*}
\end{thm}

\begin{proof}
We first observe that $A_{\bfs\beta^+}^+$ is a holomorphic quadratic differential. 
Indeed, $ib^+ \pa J$ and $\pa\psi_{\bfs\beta^+}^+J$ satisfy
\begin{align*}
ib^+ \pa J &= ib^+ h'' \tilde J\circ h + ib(h')^2 \pa \tilde J\circ h, \\
\pa\psi_{\bfs\beta^+}^+J&=(h'\pa\tilde\psi_{\bfs\beta^+}^+ \circ h + ib^+\frac{h''}{h'})h' \ti J \circ h.
\end{align*}
Similarly, $A_{\bfs\beta^-}^-$ is an anti-holomorphic quadratic differential. 

Next we check Ward's OPE on $\wh\C$
\begin{align*}
A_{\bfs\beta^+}^+(\zeta)\Psi_{\bfs\beta^+,\bfs\beta^-}(z)&= \big( A(\zeta) + ib^+ \pa J(\zeta) - \pa\psi_{\bfs\beta^+}^+(\zeta)J(\zeta)\big)\big(\Psi(z) + \psi_{\bfs\beta^+}^+(z) + \overline{\psi_{{\bfs\beta^-}}^+(z)}\big)\\
&\sim \frac{J(z)}{\zeta-z} + ib^+\frac1{(\zeta-z)^2} + \frac{\pa\psi_{\bfs\beta^+}^+(z)}{\zeta-z} = \frac{\pa \Psi_{\bfs\beta^+,\bfs\beta^-} (z)}{\zeta-z} + ib^+\frac1{(\zeta-z)^2}.
\end{align*}
Similar operator product expansion holds for $A_{\bfs\beta^-}^-(\zeta)\Psi_{\bfs\beta^+,\bfs\beta^-}(z)$ as $\zeta\to z.$

Finally we want to show that $T_{\bfs\beta^+}^+$ is a Schwarzian form of order $c^+/12$ $(c^+ = 1-12(b^+)^2).$ 
We find 
$$T_{\bfs\beta^+}^+ = T +(A_{\bfs\beta^+}^+-A)- \frac12 (\pa\psi_{\bfs\beta^+})^2 + ib^+\pa^2 \psi_{\bfs\beta^+}$$
from the expressions of $A_{\bfs\beta^+}^+$ and $T_{\bfs\beta^+}^+.$ 
All we need to check is that $- \frac12 (\pa\psi_{\bfs\beta^+})^2 + ib\pa^2 \psi_{\bfs\beta^+}$ is a Schwarzian form of order $-(b^+)^2.$
It follows from the transformation laws for $- \frac12 (\pa\psi_{\bfs\beta^+})^2 $ and $ ib^+\pa^2 \psi_{\bfs\beta^+}:$ 
\begin{align*}
- \frac12 (\pa\psi_{\bfs\beta^+})^2&= -ib^+ h'' \pa\ti\psi_{\bfs\beta^+}\circ h -\frac12 (h')^2(\pa\ti\psi_{\bfs\beta^+}\circ h)^2 + \frac12 (b^+)^2 \Big(\frac{h''}{h'}\Big)^2,\\
 ib^+\pa^2 \psi_{\bfs\beta^+}&= \phantom{-} ib^+ h'' \pa\ti\psi_{\bfs\beta^+}\circ h + ib^+ (h')^2 (\pa^2 \ti\psi_{\bfs\beta^+}\,) \circ h - (b^+)^2 \Big(\frac{h''}{h'}\Big)'.
\end{align*}
Similarly, $\overline{T_{\bfs\beta^-}^-}$ is a Schwarzian form of order $\overline{c^-}/12$ $(c^- = 1-12(b^-)^2).$ 
\end{proof}

\begin{eg*} 
In the $\wh\C$-uniformization, we have 
$$\E\,T_{\bfs\beta^+}^+(z) = \sum_k \frac{\lambda_{b^+}(\beta_k^+)}{(z-q_k)^2} + \sum_{j<k} \frac{\beta_j^+\beta_k^+}{(z-q_j)(z-q_k)}.$$
\end{eg*}

Denote by $\FF_{\bfs\beta}(S')$ the \emph{OPE family} of $\Psi_{\bfs\beta}$ on $S'\subseteq S,$ the algebra (over $\C$) spanned by $1$ and the derivatives of $\Psi_{\bfs\beta},$ $\OO_{\bfs\beta}$ with nodes in $S'$ under the OPE multiplication $*.$
Let $S^*:=S\setminus (\supp\,\bfs\beta_1^\pm\cup\,\supp\,\bfs\beta_2^\pm).$

\begin{thm} \label{change of beta's}
Given background charges $\bfs\beta^\pm_1, \bfs\beta^\pm_2$ with the neutrality conditions $(\NC_{b^\pm}),$ 
the image of $\FF_{\bfs\beta_1}(S^*)$
under the insertion of 
$\ee^{\odot i\Psi[\bfs\beta_2-\bfs\beta_1]}$ is 
$\FF_{\bfs\beta_2}(S^*)$ within correlations.
\end{thm}
\begin{proof}

We first show that the image of $\Psi_{\bfs\beta_1}$ under the insertion of 
$\ee^{\odot i\Psi[\bfs\beta_2-\bfs\beta_1]}$ is $\Psi_{\bfs\beta_2}$ within correlations:
$$\E\, \Psi_{\bfs\beta_1}(\zeta,z) \, \ee^{\odot i\Psi[\bfs\beta_2-\bfs\beta_1]} = \E\, \Psi_{\bfs\beta_2}(\zeta,z).$$
By Wick's calculus, we have
\begin{align} \label{eq: Psi under insertion}
\E\, \Psi_{\bfs\beta_1}(\zeta,z)\, \ee^{\odot i\Psi[\bfs\beta_2-\bfs\beta_1]}
&= \psi_{\bfs\beta_1}(\zeta)- \psi_{\bfs\beta_1}(z) + \E\, \Psi(\zeta,z) \, \ee^{\odot i\Psi[\bfs\beta_2-\bfs\beta_1]} \\
&= \psi_{\bfs\beta_1}(\zeta)- \psi_{\bfs\beta_1}(z) + i\,\E\, \Psi(\zeta,z) \Psi[\bfs\beta_2-\bfs\beta_1] = \psi_{\bfs\beta_2}(\zeta)-\psi_{\bfs\beta_2}(z). \nonumber
\end{align}
We define the correspondence $\XX\mapsto \wh\XX$ by the formula $\Psi_{\bfs\beta_1} \mapsto \Psi_{\bfs\beta_2}$
and the rules
$$\pa\XX\mapsto\pa\wh\XX,\quad \bp\XX\mapsto\bp\wh\XX,\quad \alpha\XX+\beta\YY\mapsto \alpha\wh\XX+\beta\wh\YY,\quad \XX\odot\YY\mapsto\wh\XX\odot\wh\YY.$$
Denote 
$$\wh\E\, \XX = \E\, \XX \ee^{\odot i\Psi[\bfs\beta_2-\bfs\beta_1]}.$$
It suffices to show 
$$\wh\E\, \XX =\E\, \wh\XX$$
for $\XX= X_1(\zeta_1,z_1) \odot \cdots \odot X_n(\zeta_n,z_n),$ $X_j = \pa_{\zeta_j}^{m_j}\bp_{\zeta_j}^{\wt m_j}\pa_{z_j}^{n_j}\bp_{z_j}^{\wt n_j} \Psi(\zeta_j,z_j).$ 
Differentiating \eqref{eq: Psi under insertion},
it follows from Wick's calculus that
$$ \E\, X_1(\zeta_1,z_1) \odot \cdots \odot X_n(\zeta_n,z_n) \Psi^{\odot k}[\bfs\beta_2-\bfs\beta_1] = 
\begin{cases}
k!\prod_{j=1}^n \,\E\,X_j(\zeta_j,z_j)\Psi[\bfs\beta_2-\bfs\beta_1] &\textrm{if } n=k,\\
0 &\textrm{otherwise.}
\end{cases}$$
We now compute 
\begin{align*}
\wh \E\, \XX & = \sum_{k=0}^\infty \frac{i^k}{k!}\, \E\, X_1(\zeta_1,z_1) \odot \cdots \odot X_n(\zeta_n,z_n) \Psi^{\odot k}[\bfs\beta_2-\bfs\beta_1] \\
&= \prod_{j=1}^n i\E\,X_j(\zeta_j,z_j)\Psi[\bfs\beta_2-\bfs\beta_1] = \prod_{j=1}^n \E\,\wh X_j(\zeta_j,z_j) = \E\, \wh\XX,
\end{align*} 
which completes the proof.
\end{proof}

\subsection{Embeddings} 
In this subsection we construct the Gaussian free field in a simply-connected domain $D$ with Dirichlet boundary condition from the Gaussian free field on its Schottky double $S = D^{\mathrm{double}}.$ 
Recall that the Gaussian free field can be defined as a Gaussian field indexed by the energy space. 
See Subsection~\ref{ss: CFT on C} for this definition on $\wh\C.$
The energy space $\EE(D)$ can be embedded isometrically into $\EE(S)$ in a natural way. 
For example, for $\mu$ in the energy space $\EE(\H),$ 
\begin{align*}
\|\mu\|_{\EE(\H)}^2 &= \int_{\H\times\H} \log\frac1{|z-w|^2} \, \mu(z) \overline{\mu(w)} -\int_{\H\times\H} \log\frac1{|z-\bar w|^2} \, \mu(z) \overline{\mu(w)}\\
&= \frac12 \int_{\C\times\C} \log\frac1{|z-w|^2} \, \nu(z) \overline{\nu(w)} = \frac12 \|\nu\|_{\EE(\wh\C)}^2,
\end{align*}
where 
$$\nu := \begin{cases} \phantom{-} \mu\phantom{^*} \textrm{ in } \H, \\
- \mu^* \textrm{ in } \H^*,
\end{cases}$$
and $E^* := \{\bar z\,|\, z\in E\}, \mu^*(E) := \mu(E^*).$
As a Fock space field, the Gaussian free field $\Phi \equiv \Phi_D$ in $D$ with Dirichlet boundary condition can be constructed from the Gaussian free field $\Psi \equiv \Psi_{S}$ on $S:$
$$\Phi_D(z) = \frac1{\sqrt 2}\Psi_{S}(z,z^*),$$
where $\iota: S\to S , z \mapsto z^*$ is the canonical involution in the Schottky double $S$ of $D.$ 
For example,
$$\Phi_\H(z) = \frac1{\sqrt 2}\Psi_{\wh\C}(z,\bar z),$$
and 
$$\Phi_\D(z) = \frac1{\sqrt 2}\Psi_{\wh\C}(z,z^*), \qquad z^* = 1/\bar z.$$
In the chordal case, the current fields and the stress energy tensors are related as 
$$J_\H(z) = \frac1{\sqrt2} \big(J_\C(z) - {\bar J_\C(\bar z)}\big)$$ 
and 
$$A_\H(z) = \frac12\big(A_\C(z)+ {\bar A_\C(\bar z)} + J_\C(z) \odot {\bar J_\C(\bar z)} \big).$$ 
Similar statements hold in the radial case. 
Furthermore, the formal bosonic fields are related as 
$$\Phi^\pm(z) = \frac1{\sqrt 2}\big(\Psi^\pm(z)-\Psi^\mp(z^*)\big).$$
More generally, for a double divisor $(\bfs\sigma^+,\bfs\sigma^-)$ satisfying the neutrality condition $(\NC_0),$ we have 
\begin{equation} \label{eq: Schottky4Phi}
\Phi[\bfs\sigma^+,\bfs\sigma^-] = \frac1{\sqrt2} \Psi[\bfs\sigma^++\bfs\sigma_*^-, \bfs\sigma_*^++\bfs\sigma^-].
\end{equation}

In Subsection~\ref{ss: GFF N} we construct the Gaussian free field in $D$ with Neumann boundary condition from $\Psi_{S}.$

\subsection{Modification of Gaussian free field in a simply-connected domain} \label{ss: Phib}
In Subsection~\ref{ss: Psib} we introduce background charge modifications $\Psi_{\bfs\beta^+,\bfs\beta^-}$ of the Gaussian free field on the Riemann sphere. 
Here the background charges $\bfs\beta^\pm$ satisfy the neutrality condition $(\NC_{b^\pm}) (b^\pm\in\C).$
In this subsection, from $\Psi_{\bfs\beta^+,\bfs\beta^-},$ we construct background charge modifications of the Gaussian free field with Dirichlet boundary condition in a simply-connected domain. 

Let us fix a background charge parameter $b\in\mathbb{R}.$
This parameter $b$ is related to the central charge $c$ as $c = 1 - 12b^2.$
For a double background charge $(\bfs\beta^+,\bfs\beta^-)$ on $\overbar{D}\times D$ satisfying the neutrality condition $(\NC_b):$
$$\int (\bfs\beta^++ \bfs\beta^-) = 2b,$$
we set $\bfs\beta = \bfs\beta^++ \bfs\beta^-_*$ and define 
$$ \Phi_{\bfs\beta}(z) \equiv \Phi _{\bfs\beta^+,\bfs\beta^-}(z) = \frac1{\sqrt2} \Psi_{\sqrt2\,\bfs\beta^+,\sqrt2\,\bfs\beta_*^-}(z,z^*).$$ 
Then $\Phi _{\bfs\beta}$ is a $\PPS(ib,-ib)$ form and 
$$\varphi_{\bfs\beta} (z):=\E\,\Phi_{\bfs\beta}(z) = \sum (\beta_k^++\beta_k^-)\big(\arg(z-q_k)+\arg(z-\bar q_k)\big) + i \sum(\beta_k^+-\beta_k^-)\log \Big|\frac{z-\bar q_k}{z-q_k}\Big|$$
in the $\H$-uniformization. 
 
We also define background charge modifications of formal bosonic fields as 
$$\Phi^\pm_{\bfs\beta}(z) = \frac1{\sqrt 2}\big(\Psi^\pm_{\sqrt 2\bfs\beta^+,\sqrt 2\bfs\beta_*^-}(z) -\Psi^\mp_{\sqrt 2\bfs\beta^+,\sqrt 2\bfs\beta_*^-} (z^*)\big).$$
For a double divisor $( \bfs\tau^+,\bfs\tau^-)$ satisfying the neutrality condition $(\NC_0),$ we set $\bfs\tau = \bfs\tau^++\bfs\tau_*^-.$
Then we have 
\begin{equation} \label{eq: Schottky4Phibeta}
\Phi_{\bfs\beta}[\bfs\tau] = \frac1{\sqrt2} \Psi_{\sqrt2\,\bfs\beta^+,\sqrt2\,\bfs\beta_*^-}[\bfs\tau, \bfs\tau_*].
\end{equation}

\subsubsec{Chordal case without spins}
For a simply-connected domain $D$ with marked boundary points $q, q_k\in \pa D,$ a background charge $\bfs\beta = (2b-\sum\beta_k)\cdot q + \sum \beta_k\cdot q_k, (\beta_k\in\R),$ and a conformal map $w\equiv w_{D,q}:~(D,q)\to(\mathbb{H},\infty)$
we have 
$$\Phi_{\bfs\beta} = \Phi + \varphi_{\bfs\beta},\qquad \varphi_{\bfs\beta}= -2b\arg w' + 2\sum\beta_k\arg(w-w(q_k)).$$
We also define $J_{\bfs\beta} := \pa \Phi_{\bfs\beta}.$ 
Then it is a $\PS(ib)$ form and 
$$ J_{\bfs\beta} = J + j_{\bfs\beta}, \qquad j_{\bfs\beta} = ib\frac{w''}{w'} - i\sum \beta_k \frac{w'}{w-w(q_k)}.$$
In particular, if $\bfs\beta= 2b\cdot q,$ then we have 
\begin{equation} \label{eq: GFFbH}
\Phi_{\bfs\beta} =\Phi -2b\arg w', \qquad J_{\bfs\beta}=\pa \Phi_{\bfs\beta}=J+ib\frac{w''}{w'},
\end{equation}
and the 1-point function $\varphi_{\bfs\beta} = \E\,\Phi_{\bfs\beta} = -2b\arg w'$ does not depend on the choice of the conformal map.
Due to this property, these modifications~\eqref{eq: GFFbH} are well-defined. 
As a $\PS(ib)$ form, the current field $J_{\bfs\beta}$ is conformally invariant with respect to $\Aut(D,q).$

\subsubsec{Radial case with a spin only at a marked interior point}
For a simply-connected domain $D$ with a marked interior point $q\in D,$ we consider a conformal map $$w\equiv w_{D,q}:~(D,q)\to(\D,0),$$
from $D$ onto the unit disc $\D = \{z\in\C:|z|<1\}.$
Let $q_k\in\pa D$ be marked boundary points. 
Given a background charge $\bfs\beta = \beta_q\cdot q + \overline{\beta_q}\cdot q^* + \sum\beta_k\cdot q_k (\beta_k\in\R)$ with the neutrality condition $(\NC_b):$
$$2\,\Re\,\beta_q + \sum \beta_k = 2b,$$
the 1-point function $\varphi_{\bfs\beta} = \E\,\Phi_{\bfs\beta}$ is given by 
$$\varphi_{\bfs\beta}(z) = 2\,\Re\, \beta_q\arg z +2\,\Im\, \beta_q \log|z| + 2\sum\beta_k \arg(z-q_k).$$
in the $\D$-uniformization. 
Also we define $J_{\bfs\beta} := \pa \Phi_{\bfs\beta}.$ 
Then it is a $\PS(ib)$ form and 
$$ J_{\bfs\beta} = J +ib\frac{w''}{w'} - i{\beta_q}\frac{w'}w - i\sum \beta_k \frac{w'}{w-w(q_k)}.$$

In particular, if $\bfs\beta = b\cdot q + b\cdot q^*,$ then 
\begin{equation} \label{eq: GFFb}
\Phi_{\bfs\beta} = \Phi + \varphi_{\bfs\beta},\qquad \varphi_{\bfs\beta}=-2b\arg \frac{w'}w.
\end{equation}
Since the multivalued function $\varphi_{\bfs\beta}$ does not depend on the choice of the conformal map, the ``bosonic" field, $\Phi_{\bfs\beta}$ in $D^*:= D\setminus\{q\}$ is invariant with respect to $\Aut(D,q).$ 
Also we have 
\begin{equation} \label{eq: Jb}
J_{\bfs\beta}=J+j_{\bfs\beta},\qquad j_{\bfs\beta}=\pa\varphi_{\bfs\beta}=ib\Big(\frac{w''}{w'}-\frac{w'}{w}\Big).
\end{equation}

\begin{rmks*}
(a) For $b\ne0,$ the 1-point function $\varphi_{\bfs\beta}$ has monodromy $4\pi b$ around $q.$

\ss \no (b) For $b\ne0,$ the 1-point function $j_{\bfs\beta}$ has a simple pole at $q,$ and $J$ is holomorphic in $D^*.$

\ss \no (c) As a $\PS(ib)$ form, the current $J_{\bfs\beta}$ is conformally invariant with respect to $\Aut(D,q).$
\end{rmks*}

For a symmetric background charge $\bfs\beta = \sum \beta_k \cdot q_k + \overbar{\beta_k}\cdot q_k^*$
on $S,$ 
we have 
$$\varphi_{\bfs\beta} (z):=\E\,\Phi_{\bfs\beta}(z) = 2\sum \Re\,\beta_k\big(\arg(z-q_k)+\arg(z-\bar q_k)\big) - 2\sum\Im\,\beta_k\log \Big|\frac{z-\bar q_k}{z-q_k}\Big|$$
and 
\begin{equation} \label{eq: j_beta} 
j_{\bfs\beta} (z):=\E\,J_{\bfs\beta}(z) = -i \sum \Big(\frac{\beta_k}{z-q_k} + \frac{\bar\beta_k}{z-\bar q_k}\Big)\end{equation}
in the $\H$-uniformization if all $q_k$'s are in $D.$
Later, we present the connection between boundary conformal field theory with symmetric background charges and the chordal/radial SLE theory with forces and spins. 
The following theorem is parallel with Theorem~\ref{SETinS}. 

\begin{thm} \label{SETinD} 
For a symmetric background charge $\bfs\beta$ on $S,$ the bosonic field $\Phi_{\bfs\beta}$ in $D$ has a stress tensor $(A_{\bfs\beta},\bar A_{\bfs\beta}):$
$$A_{\bfs\beta} = A + ib\pa J - j_{\bfs\beta}J$$
and its Virasoro field is
$$T_{\bfs\beta}=-\frac12 J_{\bfs\beta}*J_{\bfs\beta}+ib\pa J_{\bfs\beta}.$$
\end{thm} 

\begin{eg*} For a background charge $\bfs\beta = \beta_q\cdot q + \overbar{\beta_q}\cdot q^* + \sum \beta_k\cdot q_k (q\in D, q_k\in \pa D, \beta_k\in\R)$ satisfying the neutrality condition $(\NC_{b})$, 
we have 
$$\E\,T_{\bfs\beta}(z) = \sum_k \frac{\lambda_k}{(z-q_k)^2} + \sum_{j<k} \frac{\beta_j\beta_k}{(z-q_j)(z-q_k)} + \frac{\lambda_q}{z^2} + \sum_{k} \frac{\beta_q\beta_k}{z(z-q_k)}$$
in the $(\D,0)$-uniformization, 
where $\lambda_k = \lambda_b(\beta_k)$ and $\lambda_q = \lambda_b(\beta_q).$
\end{eg*}

\subsection{OPE exponentials in a punctured domain}
For simplicity, we consider the standard radial case $\bfs\beta = b\cdot q + b\cdot q^*.$ 
Our goal in this subsection is to explain the statement that under the neutrality condition $(\mathrm{NC}_0)$ on $\bfs\tau = \bfs\tau^++\bfs\tau^-_*$ the modified multi-vertex field $\OO_{\bfs\beta}[\bfs\tau]\equiv \OO_{\bfs\beta}[\bfs\tau^+,\bfs\tau^-]$ (with $\tau_q^+=\tau_q^-=0$) in $D^*$ can be viewed as the OPE exponential of $i\Phi_{\bfs\beta}[\bfs\tau^+,\bfs\tau^-].$
Non-chiral vertex fields $\VV_{\bfs\beta}^{(\sigma)}$ are defined by 
$$\VV_{\bfs\beta}^{(\sigma)}:=\ee^{*i\sigma\Phi_{\bfs\beta}} = \sum_{n=0}^\infty \frac{(i\sigma)^n}{n!}\,\Phi_{\bfs\beta}^{*n}.$$
We have
$$\VV_{\bfs\beta}^{(\sigma)}
=\ee^{i\sigma\varphi_{\bfs\beta}} \VV^{(\sigma)}=\ee^{i\sigma\varphi_{\bfs\beta}}C^{-\sigma^2}\ee^{\odot i\sigma\Phi},\qquad(\varphi_{\bfs\beta}:=\E\,\Phi_{\bfs\beta}).$$
The non-random field $\varphi_{\bfs\beta}$ is a $\PPS(ib,-ib)$ form and $C$ is the conformal radius. 
In terms of a conformal map $w:(D,q)\to (\D,0),$ we have
$$(C\,\|\,\id_D)(z) = \frac{1-|w(z)|^2}{|w'(z)|},$$ 
where $\id_D$ is the identity chart of $D.$ 
Thus the conformal radius is a $[-1/2,-1/2]$-differential and $\VV^{(\sigma)}$ is a $[\lambda,\lambda_*]$-differential with $\lambda = \lambda_b(\sigma), \lambda_* =\lambda_b(-\sigma).$
Its 1-point function in the $(\D,0)$-uniformization is 
$$\E\,\VV_{\bfs\beta}^{(\sigma)}(z) = (1-|z|^2)^{-\sigma^2}z^{\sigma b}(\bar z)^{-\sigma b}.$$
Since $\VV_{\bfs\beta}^{(\sigma)}(z)$ and $\OO_{\bfs\beta}[\sigma\cdot z,-\sigma\cdot z]$ contain the same Wick's exponential and they are differentials with the same conformal dimensions and the same correlation functions in the $(\D,0)$-uniformization, 
we conclude that 
$$\OO_{\bfs\beta}[\sigma\cdot z,-\sigma\cdot z] = \ee^{*i\sigma\Phi_{\bfs\beta}(z)}.$$

Next, we explain the OPE product of OPE exponentials and why it corresponds to the addition of divisors. 
Here is the typical example:
\begin{equation}\label{eq: V*V}
\VV_{\bfs\beta}^{(\sigma_1)}*\VV_{\bfs\beta}^{(\sigma_2)}=\VV_{\bfs\beta}^{(\sigma_1+\sigma_2)}.
\end{equation}
Here $*$ means the coefficients of the leading term in OPE expansion.
An alternative notation for such OPE multiplication is $\VV_{\bfs\beta}^{(\sigma_1)}\star\VV_{\bfs\beta}^{(\sigma_2)},$ see \cite[Section~15.2]{KM13}.
To see \eqref{eq: V*V}, we need to compute asymptotic behavior of $\exp(-\sigma_1\sigma_2\E[\Phi(\zeta)\Phi(z)]):$
$$\bigg|\frac{w(\zeta)-w(z)}{1-w(\zeta)\overline{w(z)}}\bigg|^{2\sigma_1\sigma_2}\sim \bigg|\frac{w'(z)(\zeta-z)}{1-|w(z)|^2}\bigg|^{2\sigma_1\sigma_2},\qquad \zeta\to z.$$
This implies $\ee^{\odot i\sigma_1\Phi}*\ee^{\odot i\sigma_2\Phi}= C^{-2\sigma_1\sigma_2}\ee^{\odot i(\sigma_1+\sigma_2)\Phi}.$
The same argument shows that if $f_1 $ and $f_2$ are non-random fields, then
$$(f_1\ee^{\odot i\sigma_1\Phi})*(f_2\ee^{\odot i\sigma_2\Phi})= f_1f_2C^{-2\sigma_1\sigma_2}\ee^{\odot i(\sigma_1+\sigma_2)\Phi}.$$
Taking $f_j = \E\,\VV_{\bfs\beta}^{(\sigma_j)} = \ee^{i\sigma_j\varphi_{\bfs\beta}}C^{-\sigma_j^2},$ the above identity shows \eqref{eq: V*V}.

The computation is more transparent if we use formal OPE exponentials.
Let us formally define
$$\OO^{(\sigma)} (``\equiv" \ee^{*i\sigma\Phi_{\bfs\beta}^+}) = (w')^\lambda w^{\sigma b} \ee^{\odot i\sigma\Phi^+}$$
(which depends on the choice of conformal map) so that
$$\OO[\sigma\cdot z-\sigma\cdot z_0] = \OO^{(\sigma)}(z) \OO^{(-\sigma)}(z_0)$$
holds (formally) within formal correlations. 
In addition to the product of 1-point functions (and Wick's exponential), we have the interaction term
$$\ee^{\sigma^2\E\,[\Phi^+(z)\Phi^+(z_0)]}=(w(z)-w(z_0))^{-\sigma^2}$$
in $\OO^{(\sigma)}(z) \OO^{(-\sigma)}(z_0)$.
As in the OPE product of non-chiral vertex fields, we have (again formally)
\begin{equation}\label{eq: O*O}
\OO^{(\sigma_1)}*\OO^{(\sigma_2)}=\OO^{(\sigma_1+\sigma_2)},
\end{equation}
where the OPE product $*$ could be further specified as the $\star$ product (leading
coefficient) or as $*_{\sigma_1\sigma_2}$ (coefficient of $(\zeta-z)^{\sigma_1\sigma_2}$ in the operator product expansion of $\OO^{(\sigma_1)}(\zeta) \OO^{(\sigma_2)}(z)$).
Let us compute \eqref{eq: O*O}:
\begin{align*}
\E\,\OO^{(\sigma_1)}(\zeta)\OO^{(\sigma_2)}(z) &= (w(\zeta)-w(z))^{\sigma_1\sigma_2}
w'(\zeta)^{\lambda_b(\sigma_1)}w(\zeta)^{\sigma_1b}
w'(z)^{\lambda_b(\sigma_2)}w(z)^{\sigma_2b} \\
&\sim (\zeta-z)^{\sigma_1\sigma_2}w'(z)^{\sigma_1\sigma_2+\lambda_b(\sigma_1)+\lambda_b(\sigma_2)}
w(z)^{(\sigma_1+\sigma_2)b}.
\end{align*}
It is trivial to compute Wick's exponential part in the operator product expansion.
The identity $\lambda_b(\sigma_1+\sigma_2) = \lambda_b(\sigma_1)+\lambda_b(\sigma_2)+\sigma_1\sigma_2$ now shows \eqref{eq: O*O}. 

Note that the last OPE has the form
\begin{equation} \label{OPE alpha}
X(\zeta)Y(z)= (\zeta-z)^\alpha\sum {C_n(z)}{(\zeta-z)^n}, \qquad C_n:=X*_{\alpha + n}Y. 
\end{equation}

\subsection{Insertions and changes of background charges}

Denote by $\FF_{\bfs\beta}(D')$ the \emph{OPE family} of $\Phi_{\bfs\beta}$ on $D'\subseteq D,$ the algebra (over $\C$) spanned by $1$ and the derivatives of $\Phi_{\bfs\beta},$ $\OO_{\bfs\beta}$ with nodes in $D'$ under the OPE multiplication $*.$
Let $D^*:=D\setminus (\supp\,\bfs\beta_1^\pm\cup\,\supp\,\bfs\beta_2^\pm).$

\begin{thm} \label{change of beta}
Given background charges $\bfs\beta_1, \bfs\beta_2$ on $S$ with the neutrality conditions $(\NC_b),$ 
the image of $\FF_{\bfs\beta_1}(D^*)$
under the insertion of 
$e^{\odot i\Phi[\bfs\beta_2-\bfs\beta_1]}$ is 
$\FF_{\bfs\beta_2}(D^*)$ within correlations. 
\end{thm}

Theorem~\ref{change of beta} follows immediately from Theorem~\ref{change of beta's} and the Schottky double construction \eqref{eq: Schottky4Phibeta} of $\Phi_{\bfs\beta}$ from $\Psi_{\sqrt2\,\bfs\beta^+,\sqrt2\,\bfs\beta_*^-}.$

\smallskip
We now present a typical example of the insertion operators which create the chordal/radial SLE curves.
In the radial case with $\check{\bfs\beta\,} = b\cdot q + b\cdot q^*, (q\in D),$ the insertion of one-leg operator 
$$\leg_p\equiv\leg(p):=\OO_{\check{\bfs\beta\,}}[\bfs\tau],\quad \bfs\tau= a\cdot p-\frac a2\cdot q -\frac a2\cdot q^*, \quad (p\in \partial D)$$ 
produces an operator 
$$\XX\mapsto\wh\XX$$
acting on Fock space functionals/fields by the rules 
\begin{equation} \label{eq: BCrules}
\pa\XX\mapsto\pa\wh\XX,\quad \bp\XX\mapsto\bp\wh\XX,\quad \alpha\XX+\beta\YY\mapsto \alpha\wh\XX+\beta\wh\YY,\quad \XX\odot\YY\mapsto\wh\XX\odot\wh\YY
\end{equation}
for Fock space functionals $\XX$ and $\YY$ in $D^*$ and the formula
\begin{equation} \label{eq: BC}
\wh\Phi_{\check{\bfs\beta\,}}(z) = \Phi_{\bfs\beta}(z) = \Phi_{\check{\bfs\beta\,}}(z)+a\arg \frac{(1-w(z))^2}{w(z)},
\end{equation}
where $w:(D, p,q)\to (\D, 1,0)$ is a conformal map and 
\begin{equation} \label{eq: beta radial} 
\bfs\beta = \check{\bfs\beta\,} + \bfs\tau = a\cdot p+ \Big(b-\frac a2\Big)\cdot q + \Big(b-\frac a2\Big)\cdot q^*.
\end{equation}

In the chordal case with $\check{\bfs\beta\,} = 2b\cdot q, (q\in \pa D),$ recall that the insertion of one-leg operator 
$$\leg_p\equiv\leg(p):=\OO_{\check{\bfs\beta\,}}[\bfs\tau],\quad \bfs\tau = a\cdot p-a\cdot q, \quad (p\in \partial D)$$ 
produces an operator $\XX\mapsto\wh\XX$ by the rules~\eqref{eq: BCrules} and the formula
$$
\wh\Phi_{\check{\bfs\beta\,}}(z) = \Phi_{\bfs\beta}(z) = \Phi_{\check{\bfs\beta\,}}(z)+2a\arg w(z),
$$
where $w:(D, p,q)\to (\H, 0,\infty)$ is a conformal map and 
$$\bfs\beta = \check{\bfs\beta\,} + \bfs\tau = a\cdot p + (2b-a)\cdot q.$$ 
See \cite[Section~14.2]{KM13}.

Let 
$$\wh \E[\XX]:=\frac{\E [\leg_p\XX]}{\E [\leg_p]}=\E [\ee^{\odot i \Phi[\bfs\tau]}\XX], \quad \bfs\tau= a\cdot p-\frac a2\cdot q -\frac a2\cdot q^*.$$

Let $\wh\XX\in \FF_{\bfs\beta}(D^*)$ correspond to the string $\XX\in \FF_{\check{\bfs\beta}\,}(D^*)$ under the map given by \eqref{eq: BCrules} -- \eqref{eq: BC}.
Then by Theorem~\ref{change of beta} we have 
\begin{equation} \label{eq: hat E=E hat}
\wh \E[\XX]=\E[\wh\XX].
\end{equation}

\begin{egs*} Let $\varphi_{\bfs\beta} := \E\,\Phi_{\bfs\beta}$ in the radial case with $\bfs\beta$ in \eqref{eq: beta radial}.

\ss \no (a) The current $ J_{\bfs\beta}$ is a pre-Schwarzian form of order $ib,$
\begin{equation} \label{eq: J hat}
J_{\bfs\beta} =J + ia\frac{w'}{1-w} + \Big(\frac{ia}2-ib\Big)\frac{w'}{w} + ib\frac{w''}{w'}.
\end{equation}
In the $(\D,1,0)$-uniformization, 
$$j_{\bfs\beta}(z)=\E\,J_{\bfs\beta}(z) = ia\dfrac{1}{1-z} + \Big(\dfrac{ia}2-ib\Big)\dfrac{1}{z};$$

\ss \no (b) The Virasoro field $T_{\bfs\beta}$ is a Schwarzian form of order $\frac1{12}c,$
\begin{align} \label{eq: T hat}
T_{\bfs\beta} &= -\dfrac12 J_{\bfs\beta}* J _{\bfs\beta}+ ib\pa J _{\bfs\beta} \\
&= A - j_{\bfs\beta} J + ib\pa J + \frac{c}{12}S_w + h_{1,2}\frac{w'^2}{w(1-w)^2} + h_{0,1/2}\frac{w'^2}{w^2} \qquad (A = -\frac12 J\odot J), 
\nonumber
\end{align}
where $h_{1,2}=\frac12a^2-ab,$ $h_{0,1/2}=\frac18{a^2}-\frac12b^2,$ 
(see \eqref{eq: h12} and \eqref{eq: h01/2})
and $S_w = N_w' -\frac12{N_w^2}, \, (N_w = (\log w')')$ is the Schwarzian derivative of $w.$ 
In the $(\D,1,0)$-uniformization, 
$$\E\,T_{\bfs\beta}(z) = \dfrac{h_{1,2}}{z(1-z)^2} + \dfrac{h_{0,1/2}}{z^2};$$

\ss \no (c) The operator produced by the insertion of one-leg operator can be extended to the formal fields. 
For example, we have 
$$\Phi_{\bfs\beta}^+ = \Phi_{\check{\bfs\beta\,}}^+ + \dfrac{ia}2\log w -ia\log(1-w);$$
and 
$$\OO_{\bfs\beta}^{(\sigma)}=\Big(\frac{1-w}{\sqrt{w}}\Big)^{a\sigma} \OO_{\check{\bfs\beta\,}}^{(\sigma)}.$$
\end{egs*}

\section{Extended OPE family} \label{sec: F} 
In this section we extend the OPE family $\FF_{\bfs\beta}(D^*)$ ($D^*:=D\setminus (\supp\,\bfs\beta^+\cup\,\supp\,\bfs\beta^-)$, $\bfs\beta=\bfs\beta^++\bfs\beta^-_*$) of $\Phi_{\bfs\beta}$ by adding the generators obtained from the operator product expansion of fields in $\FF_{\bfs\beta}(D^*)$ at the punctures $q_k\in \supp\,\bfs\beta^+\cup\,\supp\,\bfs\beta^-$ or the rooting procedure. 
Examples include the OPE exponentials of $\Phi_{\bfs\beta}$ with nodes at $q_k\in \supp\,\bfs\beta^+\cup\,\supp\,\bfs\beta^-.$ 

\subsection{OPE at the puncture} 
Let us consider the radial case first. 
For simplicity we only consider the ``formal" holomorphic puncture differential in the case that $\check{\bfs\beta\,} = b\cdot q + b\cdot q^*:$
$$\OO_q^{(\tau)} 
=(w_q')^{\tau^2/2}\ee^{\odot i\tau\Phi^+(q)}.$$
This is a (formal) Fock space correlation functional with conformal dimension $h_q = \tau_q^2/2$ at $q.$
\begin{prop}
We have 
$$\OO^{(\sigma)}(z)\OO_q^{(\tau)} \sim (z-q)^{\sigma(\tau+b)}\OO_q^{(\sigma+\tau)}.$$ 
\end{prop}
\begin{proof} As $z\to q,$ we have 
\begin{align*}
\OO^{(\sigma)}(z)\OO_q^{(\tau)}&= w^{\sigma b} (w')^{\lambda_b(\sigma)} \ee^{\odot i\sigma\Phi^+(z)}(w_q')^{\tau^2/2}\ee^{\odot i\tau\Phi^+(q)}\\
&=w^{\sigma (\tau+b)} (w')^{\lambda_b(\sigma)} (w_q')^{\tau^2/2}\ee^{\odot i\sigma\Phi^+(z) + i\tau\Phi^+(q)} \\
&\sim (z-q)^{\sigma (\tau+b)} (w_q')^{\sigma (\tau+b)+\lambda_b(\sigma)+\tau^2/2} \ee^{\odot i(\sigma+\tau)\Phi^+(q)} \\
&= (z-q)^{\sigma (\tau+b)} \OO_q^{(\sigma+\tau)}.
\end{align*}
Here we use $\ee^{\odot i\sigma\Phi^+(z)}\ee^{\odot i\tau\Phi^+(q)} = w^{\sigma \tau}\ee^{\odot i(\sigma\Phi^+(z) + \tau\Phi^+(q))},$ $w(z)\sim(z-q)w_q'$ as $z\to q,$ and the identity $\lambda_b(\sigma)+\tau^2/2+\sigma (\tau+b)= (\sigma+\tau)^2/2.$
The notation $\sim$ means the first leading term of the operator product expansion.
\end{proof}

We can express the above proposition in any of the following formulas:
$$\OO^{(\sigma)}*\OO_q^{(\tau)}\equiv \OO^{(\sigma)}*_{\sigma(\tau+b)}\OO_q^{(\tau)}= \OO_q^{(\sigma+\tau)},\qquad \OO^{(\sigma)}\star\OO_q^{(\tau)}=\OO_q^{(\sigma+\tau)}.$$ 
The point is that the arithmetic of divisors has the OPE nature, both in $D^*:=D\sm\{q\}$ and at the puncture $q.$

Let us mention a special case ($\tau=0$) of the proposition.

\begin{cor}
We have 
$$\OO_q^{(\sigma)}=\OO^{(\sigma)}*1_q.$$
\end{cor}

This OPE $\OO^{(\sigma)}(z) \sim (z-q)^{\sigma b} \OO_q^{(\sigma)}\,(z\to q)$ is related to the ``rooting" procedure (see \cite[Section~12.3]{KM13} for this procedure in the chordal case):
$$\OO_q^{(\sigma)} = \lim_{\ve\to0} z_\ve^{-\sigma b}\,\OO^{(\sigma)}(z_\ve),$$
where the point $z_\ve$ is at distance $\ve$ from $q=0$ in the chart. 
Equivalently, this definition can be reached by applying the following rooting rule to the formal field $\OO^{(\sigma)}(z) = w(z)^{\sigma b}w'(z)^\lambda\,\ee^{\odot i\sigma\Phi(z)}:$

\renewcommand{\theenumi}{\alph{enumi}}
\begin{enumerate}
\item the point $z$ in the terms $w'(z)^\lambda$ and $\ee^{\odot i\sigma\Phi(z)}$ is replaced by the puncture $q;$
\item the term $w(z)^{\sigma b}$ is replaced by $(w_q')^{\sigma b}.$
\end{enumerate}

\ms

In the chordal case with $\check{\bfs\beta\,} = 2b\cdot q,$ we consider the ``formal" boundary puncture differential 
$$\OO_q^{(\tau)} =(w_q')^{\lambda_b(-\tau)}\ee^{\odot i\tau\Phi^+(q).}$$
In terms of a uniformizing map $w:(D,q)\to(\H,0)$
we have $\OO^{(\sigma)} = (w')^{\lambda_b(\sigma)}w^{2\sigma b}$ and 
$$\OO^{(\sigma)}(z)\OO_q^{(\tau)} \sim (z-q)^{\sigma(\tau+2b)}\OO_q^{(\sigma+\tau)}$$
as $z\to q.$
We use the identity $\lambda_b(\sigma)+\lambda_b(-\tau)+\sigma(\tau+2b) =\lambda_b(-\sigma-\tau).$ 
This OPE can be expressed as 
$$\OO^{(\sigma)}*\OO_q^{(\tau)}\equiv \OO^{(\sigma)}*_{\sigma(\tau+2b)}\OO_q^{(\tau)}= \OO_q^{(\sigma+\tau)},\qquad \OO^{(\sigma)}\star\OO_q^{(\tau)}=\OO_q^{(\sigma+\tau)}.$$ 
We refer to \eqref{OPE alpha} to remind the readers what this type of OPE is. 
In the special case $\tau = 0,$ we have 
$$\OO_q^{(\sigma)}=\OO^{(\sigma)}*_{2\sigma b}1_q.$$
This is related to the ``rooting" procedure:
$$\OO_q^{(\sigma)} = \lim_{\ve\to0} z_\ve^{-2\sigma b}\OO^{(\sigma)}(z_\ve),$$
where the point $z_\ve$ is at distance $\ve$ from $q=0$ in the chart. 
In terms of a uniformizing map $w:(D,q)\to(\H,\infty)$ satisfying $w(\zeta) \sim -1/(\zeta-q)$ as $\zeta\to q$ in a fixed boundary chart at $q,$ then the rooting procedure becomes 
$$\OO_q^{(\sigma)} = \lim_{\ve\to0} w'(z_\ve)^{\sigma b}\OO^{(\sigma)}(z_\ve),$$
where the point $z_\ve$ is at (spherical) distance $\ve$ from $q$ in the chart $\phi,$ see \cite[Section~12.3]{KM13}.

\subsection{Definition of the extended OPE family} \label{ss: extended OPE family}
For simplicity, we consider the standard radial case with $\bfs\beta= b\cdot q + b\cdot q^*.$
If $X\in\FF_{\bfs\beta}(D^*),$ $(D^*:=D\setminus\{q\}),$ then for all points $z\in D$ and all holomorphic (local) vector fields $v$ we have
\begin{equation} \label{eq: FD*}
\frac1{2\pi i}\oint_{(z)} vT_{\bfs\beta}\,X(z) = \LL_v^+X(z) , \qquad
 \frac1{2\pi i}\oint_{(z)} vT_{\bfs\beta}\,\bar X(z) = \LL_v^+\bar X(z).
\end{equation}
This is not a characterization of fields in $\FF_{\bfs\beta}(D^*),$ e.g., $X\Phi(q)$ would satisfy \eqref{eq: FD*}. 
Note that the fields in $\FF_{\bfs\beta}(D^*)$ are typically not defined at $q,$ e.g., $\Phi_{\bfs\beta}, J_{\bfs\beta},$ and $T_{\bfs\beta}.$

A Fock space functional at $q$ is the value of a Fock space field at $q,$
$$X_q = X(q).$$
Here Fock space means the Fock space of the Gaussian free field $\Phi$ and the functionals that depend on the chart at the puncture $q.$ 
A Fock space functional at $q$ correlates with Fock space strings $\XX$ in $D^*.$ 
In fact we can consider the (tensor) products $\XX X_q.$ 
Such products are well-defined Fock space functionals.
Examples of Fock space functional at $q$ include $1_q, \Phi(q),$ and even formal functionals $\OO^{(\sigma)}(q)$ (with our usual proviso of neutrality).

\begin{def*}
We say $X_q\in \FF_{\bfs\beta}(q)$ if $X_q = X_1*(\cdots*(X_n*1_q)\cdots)$ for some fields $X_1,\cdots X_n\in \FF_{\bfs\beta}(D^*).$
The (extended) OPE family $\FF_{\bfs\beta} = \FF_{\bfs\beta}(D,q)$ is the collection of Fock space strings
$$\{\YY = \XX X_q\,|\,\XX\in\FF_{\bfs\beta}(D^*), X_q\in \FF_{\bfs\beta}(q)\}.$$
\end{def*}

For example, if $X = Y = J_{\check{\bfs\beta\,}}$ ($\check{\bfs\beta\,} = b\cdot q + b\cdot q^*$), then 
both $Y*(X*1_q)$ and $(Y*X)*1_q$ are in $\FF_{\bfs\beta}(q),$ but they are different, see the last example in the next subsection.

\subsection{Examples of the OPE family} 
\label{ss: Eg4Fq}
It is clear that $1_q\in\FF(q),$ so every $X\in\FF(D^*)$ belongs to the extended family.
We now present some less obvious examples in the standard radial case, $\check{\bfs\beta\,} = b\cdot q + b\cdot q^*.$

\begin{eg*} As mentioned before,
$$\OO_q^{(\sigma)} = \OO^{(\sigma)}*1_q\in\FF_{\check{\bfs\beta\,}}(q),$$
and more generally
$\OO^{(\sigma)}*\OO_q^{(\tau)} = \OO^{(\sigma+\tau)}_q.$
We define the formal functional $\OO_q^{(\tau,\tau_*)}$ by 
$$\OO_q^{(\tau,\tau_*)} := \OO_q^{(\tau)}\overline{\OO_q^{(\overline{\tau_*})} }=(w_q')^{\tau^2/2} (\overline{w_q'})^{\tau_*^2/2}\ee^{\odot i\tau\Phi^+(q)-i\tau_*\Phi^-(q)}.$$ 
Let us discuss the vertex algebra at the puncture. 
We define the multiplication of OPE exponentials at the puncture by 
$$\OO_q^{(\sigma,\sigma_*)}\OO_q^{(\tau,\tau_*)}=\OO^{(\sigma,\sigma_*)}*\OO_q^{(\tau,\tau_*)}$$
so that 
$\OO_q^{(\sigma,\sigma_*)}\OO_q^{(\tau,\tau_*)} =\OO_q^{(\sigma+\tau,\sigma_*+\tau_*)}.$
\end{eg*}

\begin{egs*} (a) We have $(\Phi_{\check{\bfs\beta\,}})_q = \Phi(q).$ It is because $(z-q)w'(z)/w(z) = 1 + o(z-q).$

\ss \no (b) We have 
$$(J_{\check{\bfs\beta\,}})_q = J(q) + \frac{ib}2N_w(q).$$
This follows from $J_{\check{\bfs\beta\,}}= J+j_{\check{\bfs\beta\,}}$ with $j_{\check{\bfs\beta\,}}= ib({w''}/{w'}-{w'}/{w}),$ and 
$$\Big(\frac{w'}w\Big)*1_q = \lim_{z\to q} \Big(\frac{w'(z)}{w(z)}-\frac1{z-q}\Big)=\lim_{z\to q}\frac{w''(z)}{w'(z)+w(z)/(z-q)}=\frac12 \frac{w''(q)}{w'(q)}.$$

 \ss \no (c) We have 
\begin{align*}
(T_{\check{\bfs\beta\,}})_q &= A(q)+2ib\pa J(q) -\frac{ib}2 N_w(q) J(q)
+\Big(\frac{c}{12}-\frac13\,b^2\Big)S_w(q) - \frac38\,b^2N_w(q)^2.
\end{align*}
\end{egs*}

\begin{eg*} For $X_q = X*1_q,$ $Y*X_q \ne (Y*X)*1_q$ in general. 
For example, with the choice of $X = Y = J_{\check{\bfs\beta\,}},$ we have $Y*(X*1_q) \ne (Y*X)*1_q.$
It is because 
$$\Big(\frac{w''}{w'}-\frac{w'}w\Big)^2*1_q = -\frac43\,S_w(q) - \frac14 \,N_w(q)^2 \ne \bigg(\Big(\frac{w''}{w'}-\frac{w'}w\Big)*1_q\bigg)^2=\frac14\, N_w(q)^2.$$
\end{eg*}

\subsection{Insertions}
In this subsection we extend the insertion of the one-leg operator to the strings containing puncture functionals in particular to the fields in the extended OPE family. 
For simplicity, we consider the standard radial case with $\check{\bfs\beta\,} = b\cdot q + b\cdot q^*.$
Let $\leg_\zeta$ denote the one-leg operator,
$$\leg_\zeta = \OO_{\check{\bfs\beta\,}}[a\cdot\zeta-\frac a2\cdot q -\frac a2\cdot q^*]=\OO^{(a)}(\zeta)\OO^{(-a/2)}(q)\overline{\OO^{(-a/2)}(q)}.$$
If $q$ is a node of $\XX,$ e.g.,
$$\XX = \YY X_q, \qquad X_q = X*1_q\in\FF_{\bfs\beta}(q),$$
then the accurate definition of the product $\leg_\zeta\XX$ is 
$$\leg_\zeta\XX:= \OO^{(a)}(\zeta)\YY(X*\OO_q), \qquad \OO_q:= \OO^{(-a/2)}(q)\overline{\OO^{(-a/2)}(q)}.$$

\begin{prop} \label{prop: insertion}
Let $\XX = \YY X_q, $ $X_q = X*1_q$ and $\wh X_q = \wh X*1_q.$ Then we have 
$$\wh \E\,\YY X_q = \E\,\wh\YY \wh X_q.$$
\end{prop}

\begin{proof}
By definition,
$$\wh \E\,\YY X_q=\frac{\E\,\leg_p\YY X_q}{\E\,\leg_p}= \frac{\E\,\OO^{(a)}(p)\YY(X*\OO_q)}{\E\,\leg_p}.$$
In follows from \eqref{eq: hat E=E hat} that 
$$\frac{\E\,\YY X(z) \leg_p}{\E\,\leg_p} = \E\,\wh\YY\wh X(z) 1_q.$$
Comparing the OPE coefficients at the puncture, we have 
$$\frac{\E\,\OO^{(a)}(p)\YY(X*\OO_q)}{\E\,\leg_p}=\E\,\wh\YY \wh X_q.$$
\end{proof}

\begin{egs*} Let $\check{\bfs\beta\,} = b\cdot q + b\cdot q^*$ and $\bfs\beta = a\cdot p+ (b-a/2)\cdot q + (b-a/2)\cdot q^*.$

\ms\no (a) Recall that (see \eqref{eq: BC} and Subsection~\ref{ss: Eg4Fq}, Example~(a))
$$\wh\Phi_{\check{\bfs\beta\,}} = \Phi_{\check{\bfs\beta\,}} + a\arg\frac{(1-w)^2}{w} = \Phi_{\bfs\beta},$$ 
and $(\Phi_{\check{\bfs\beta\,}})_q = \Phi(q).$
Since $(\arg w)*1_q = \arg w_q',$ we have 
$$(\wh\Phi_{\check{\bfs\beta\,}})_q=\Phi(q)-a\arg w_q'.$$

\ms\no (b) It follows from \eqref{eq: J hat} and Example~(b) in Subsection~\ref{ss: Eg4Fq} that 
$$(\wh J_{\check{\bfs\beta\,}})_q = J(q) + i\Big(\frac{a}4+\frac{b}2\Big) N_w(q)+ iaw'(q).$$

\ms\no (c) Note $\wh\OO_q^{(\sigma)}:=\wh\OO^{(\sigma)}*1_q = (w_q')^{-a\sigma/2}\OO_q^{(\sigma)}, $
cf. $$ \wh\OO^{(\sigma)} = \Big(\frac{1-w}{\sqrt{w}}\Big)^{a\sigma}\OO^{(\sigma)}.$$
We have $\wh\E \,\OO_q^{(\sigma)}= \E \,\wh\OO_q^{(\sigma)}$ 
because the differentials on both sides have the same dimension and are equal to 1 in $(\D, 0,1).$
\end{egs*}

It is not difficult to see that Theorem~\ref{change of beta} extends to Theorem~\ref{main: change of beta} using the argument in the proof of Proposition~\ref{prop: insertion}. 
We remark that Proposition~\ref{prop: insertion} is a special case ($\check{\bfs\beta\,} = b\cdot q + b\cdot q^*$ and $\bfs\beta = a\cdot p + (b-a/2)\cdot q + (b-a/2)\cdot q^*$) of Theorem~\ref{main: change of beta}. 

For two background charges $\bfs\beta, \check{\bfs\beta\,}$ satisfying the neutrality condition $(\NC_b),$ let 
$$\wh\E\, \XX_{\check{\bfs\beta\,}} := \E\,\XX_{\check{\bfs\beta\,}} V^\odot[\bfs\beta-\check{\bfs\beta\,}] = \frac{\E\, \XX_{\check{\bfs\beta\,}}\OO_{\check{\bfs\beta\,}}[\bfs\beta-\check{\bfs\beta\,}]}{\E\, \OO_{\check{\bfs\beta\,}}[\bfs\beta-\check{\bfs\beta\,}]}.$$
If the nodes of $\XX_{\check{\bfs\beta\,}}$ intersect with the those of $\bfs\beta-\check{\bfs\beta\,}$, then the product of $\XX_{\check{\bfs\beta\,}}$ and $\OO_{\check{\bfs\beta\,}}[\bfs\beta-\check{\bfs\beta\,}]$ should be understood as the sense of OPE's. 

The insertion of $V^\odot[\bfs\beta-\check{\bfs\beta\,}]$ is an operator $\XX\mapsto \wh\XX$ on Fock space functionals/fields given by the rules \eqref{eq: BCrules} and the formula
$$\wh\Phi_{\check{\bfs\beta\,}} = \Phi_{\bfs\beta}.$$

\begin{eg*} 
Suppose that a divisor $\bfs\tau$ on $S$ satisfies the neutrality condition $(\NC_0).$ 
Then we have 
\begin{equation} \label{eq: O hat}
\wh\OO_{\check{\bfs\beta\,}}[\bfs\tau] = \OO_{\bfs\beta}[\bfs\tau].
\end{equation}
Since both sides contain the same Wick's exponential $V^\odot [\bfs\tau] = \ee^{\odot i \Phi[\bfs\tau]},$ all we need is to check that both sides have the same correlation functions. 
It follows from Theorem~\ref{main: change of beta} and the algebra of OPE exponentials that 
$$
\E\,\wh\OO_{\check{\bfs\beta\,}}[\bfs\tau] = \wh \E\, \OO_{\check{\bfs\beta\,}}[\bfs\tau]= \frac{\E\, \OO_{\check{\bfs\beta\,}}[\bfs\tau]\OO_{\check{\bfs\beta\,}}[\bfs\beta-\check{\bfs\beta\,}]}{\E\, \OO_{\check{\bfs\beta\,}}[\bfs\beta-\check{\bfs\beta\,}]} = \frac{\E\, \OO_{\check{\bfs\beta\,}}[\bfs\tau+\bfs\beta-\check{\bfs\beta\,}]}{\E\, \OO_{\check{\bfs\beta\,}}[\bfs\beta-\check{\bfs\beta\,}]} = \frac{C[\bfs\beta+\bfs\tau]}{C[\bfs\beta]} = \E\,\OO_{\bfs\beta}[\bfs\tau].
$$
\end{eg*}

\section{Ward identities and BPZ equations}
We represent the Ward functionals as the Lie derivative operators within correlations of fields in the extended OPE family.
Combining these representations (Ward identities) with the expression of the Virasoro fields in terms of the Ward functionals, we derive the Belavin-Polyakov-Zamolodchikov equations (BPZ equations) for correlations involving the Virasoro fields or the Virasoro generators. 

\subsection{Ward's identity}
Let $D^*:=D\setminus(\supp\,\bfs\beta^+\cup\,\supp\,\bfs\beta^-)$ and $\bfs\beta=\bfs\beta^++\bfs\beta^-_*$ as before.

\begin{lem} \label{*F}
For a non-random local vector field $v,$ all fields in $\FF_{\bfs\beta}(D^*)$ satisfy local Ward's identity~\eqref{eq: LvWv} in $D_\hol(v) \cap D^*.$
\end{lem}

\begin{proof}
For simplicity, let us consider the case of \emph{holomorphic} fields $X$ and $Y$ only. 
We need to show that if $X$ and $Y$ satisfy the residue form~\eqref{eq: LvWv} of Ward's identity at the nodes, then so does $X*_{\alpha+n}Y\in\FF_{\bfs\beta}(D^*).$ 
(The case $\alpha =0$ was covered in \cite[Proposition~5.8]{KM13}.)
In this case, we have 
$\LL_v^+X = (vT)*_{-1}X$ and $\LL_v^+Y = (vT)*_{-1}Y.$
We need to check that
$$\LL_v^+(X*_{\alpha+n}Y) = (vT)*_{-1}(X*_{\alpha+n}Y).$$
By Leibniz's rule, the left-hand side is
$(\LL_v^+X)*_{\alpha+n}Y+X*_{\alpha+n}(\LL_v^+Y).$
Since $X$ and $Y$ satisfy the residue form of Ward's identity, it is equal to 
$$[(vT)*_{-1}X]*_{\alpha+n}Y + X*_{\alpha+n}[(vT)*_{-1}Y].$$
Let us show $[(vT)*_{-1}X]*_{\alpha+n}Y + X*_{\alpha+n}[(vT)*_{-1}Y] = (vT)*_{-1}(X*_{\alpha+n}Y).$

Let $C_-,$ $C,$ and $C_+$ be three concentric circles centered at $z$ with increasing radii, $\ve_-<\ve<\ve_+(\ll1).$
In the correlations with the string $\XX$ whose nodes are outside of the discs, we have
\begin{align*}
(vT)*_{-1}(X*_{\alpha+n}Y)(z) =\frac1{2\pi i}\, \oint_{\zeta\in C_+}v(\zeta)T(\zeta)\, (X*_{\alpha+n}Y)(z)\,\dd &\zeta\\
=\frac1{(2\pi i)^2}\, \oint_{\zeta\in C_+}\oint_{\eta\in C} v(\zeta)T(\zeta)\, (\eta-z)^{-\alpha-n-1}X(\eta) \,Y(z)\,\dd \eta\,\dd &\zeta.
\end{align*}
For $\alpha\notin \mathbb{Z},$ the multivalued factor $(\eta-z)^{-\alpha}$ is canceled out by the $\OPE_\alpha$ of $X(\eta)$ and $Y(z).$
In a similar way, we compute $X*_{\alpha+n}[(vT)*_{-1}Y](z)$ as 
$$\qquad\frac1{(2\pi i)^2}\, \oint_{\zeta\in C_-}\oint_{\eta\in C} v(\zeta)T(\zeta)\, (\eta-z)^{-\alpha-n-1}X(\eta) \,Y(z)\,\dd \eta\,\dd \zeta.$$
Subtracting, we express $(vT)*_{-1}(X*_{\alpha+n}Y)(z)-X*_{\alpha+n}[(vT)*_{-1}Y](z)$ as 
\begin{gather*}
\frac1{(2\pi i)^2}\, \oint_{\eta\in C}\oint_{\zeta\in [C_+-C_-]} v(\zeta)T(\zeta)\, (\eta-z)^{-\alpha-n-1}X(\eta) \,Y(z)\,\dd \zeta\,\dd \eta\\
=\frac1{(2\pi i)^2}\, \oint_{\eta\in C}(\eta-z)^{-\alpha-n-1}\oint_{(\eta)} v(\zeta)T(\zeta)X(\eta)\, \dd \zeta\,Y(z)\,\dd \eta.
\end{gather*}
The last integral simplifies to $[(vT)*_{-1}X]*_{\alpha+n}Y(z).$
\end{proof}

The next lemma states that within correlations of fields in $\FF_{\bfs\beta}(D^*)$ the residue form of Ward's identity holds at the marked points $q_k$'s where the background charges are placed.

\begin{lem} \label{Ward@q}
We have
$$ \frac1{2\pi i}\oint_{(q_k)} vA_{\bfs\beta} = \LL_{v}^+(q_k) = v(q_k)\pa_{q_k}$$
within correlations of fields in $\FF_{\bfs\beta}(D^*).$ 
\end{lem}

\begin{proof}
From the relation
$$A_{\bfs\beta} = A + ib\pa J -j_{\bfs\beta}J, \qquad j_{\bfs\beta}(\zeta) = \E\,J_{\bfs\beta}(\zeta),$$
and the fact that the meromorphic function $j_{\bfs\beta}$ has a simple pole at $q_k$ with $\Res_{q_k}j_{\bfs\beta}=-i\beta_k,$
it follows that 
$$ \frac1{2\pi i}\oint_{(q_k)} vA_{\bfs\beta} = - \frac1{2\pi i}\oint_{(q_k)} v j_{\bfs\beta} J = i\beta_k v(q_k) J(q_k).$$
We now claim that 
\begin{equation}\label{eq: Ward@q claim}
i\beta_k\E\, J(q_k) \XX_{\bfs\beta} = \pa_{q_k}\E\,\XX_{\bfs\beta}
\end{equation}
for any string $\XX_{\bfs\beta}$ of fields in $\FF_{\bfs\beta}(D^*).$ 
Then lemma follows immediately. 
Note that the second identity 
$$\LL_{v}^+(q_k) \E\,\XX_{\bfs\beta}= v(q_k)\pa_{q_k}\E\,\XX_{\bfs\beta}$$
is obvious since $\E\,\XX_{\bfs\beta}$ is a scalar with respect to $q_k$.

To prove the claim, let us consider a reference background charge $\bfs\beta_0$ such that 
$$\supp\, \bfs\beta \cap \supp\, \bfs\beta_0 = \emptyset.$$ 
Differentiating the relation (Theorem~\ref{change of beta}) $\E\,\XX_{\bfs\beta} = \E\,\XX_{\bfs\beta_0} \ee^{\odot i\Phi[\bfs\beta-\bfs\beta_0]},$
we have 
\begin{equation}\label{eq: Ward@q0}
\pa_{q_k}\E\XX_{\bfs\beta} = \pa_{q_k}\E\,\XX_{\bfs\beta_0} \ee^{\odot i\Phi[\bfs\beta-\bfs\beta_0]} = i\beta_k \E\,\XX_{\bfs\beta_0} J(q_k)\odot \ee^{\odot i\Phi[\bfs\beta-\bfs\beta_0]}.
\end{equation}
It follows from Wick's calculus that
\begin{equation}\label{eq: Ward@q1}
J(q_k')\odot \ee^{\odot i\Phi[\bfs\beta-\bfs\beta_0]}
=J(q_k')\, \ee^{\odot i\Phi[\bfs\beta-\bfs\beta_0]} -
i \E[J(q_k')\Phi[\bfs\beta-\bfs\beta_0]]\,\ee^{\odot i\Phi[\bfs\beta-\bfs\beta_0]}.
\end{equation}
(Indeed, if one applies Wick's calculus to $J(q_k')\, \ee^{\odot i\Phi[\bfs\beta-\bfs\beta_0]},$ the term $J(q_k')\odot \ee^{\odot i\Phi[\bfs\beta-\bfs\beta_0]}$ comes from $0$ contraction and the term $i \E[J(q_k')\Phi[\bfs\beta-\bfs\beta_0]]\,\ee^{\odot i\Phi[\bfs\beta-\bfs\beta_0]}$ comes from $1$ contraction.)
Here we take $q_k'\ne q_k$ so that the term $J_{\bfs\beta_0}(q_k')\, \ee^{\odot i\Phi[\bfs\beta-\bfs\beta_0]}$ in the above makes sense but send $q_k'$ to $q_k$ at the end.
By Wick's calculus, we find the non-random factor in the last term as 
\begin{equation}\label{eq: Ward@q2}
i \E[J(q_k')\Phi[\bfs\beta-\bfs\beta_0]] = \E[J(q_k')\,\ee^{\odot i\Phi[\bfs\beta-\bfs\beta_0]}] = j_{\bfs\beta}(q_k')- j_{\bfs\beta_0}(q_k').
\end{equation}
Alternatively, the last identity in the above follows from Theorem~\ref{change of beta}: 
\begin{align*}
j_{\bfs\beta}(q_k') = \E\, J_{\bfs\beta}(q_k') &= \E[J_{\bfs\beta_0}(q_k') \,\ee^{\odot i\Phi[\bfs\beta-\bfs\beta_0]}] \\&= \E[(J(q_k')+j_{\bfs\beta_0}(q_k')) \,\ee^{\odot i\Phi[\bfs\beta-\bfs\beta_0]}] = \E[J(q_k')\,\ee^{\odot i\Phi[\bfs\beta-\bfs\beta_0]}] + j_{\bfs\beta_0}(q_k').
\end{align*}
Combining \eqref{eq: Ward@q1} and \eqref{eq: Ward@q2}, we have 
$$J(q_k')\odot \ee^{\odot i\Phi[\bfs\beta-\bfs\beta_0]}
=J(q_k')\, \ee^{\odot i\Phi[\bfs\beta-\bfs\beta_0]} +\big( j_{\bfs\beta_0}(q_k')-j_{\bfs\beta}(q_k')\big)\,\ee^{\odot i\Phi[\bfs\beta-\bfs\beta_0]}.$$
In application to $\XX_{\bfs\beta_0}$ within correlations, 
$$\E\,\XX_{\bfs\beta_0} J(q_k')\odot \ee^{\odot i\Phi[\bfs\beta-\bfs\beta_0]}
=\E\,\XX_{\bfs\beta_0} J_{\bfs\beta_0}(q_k')\, \ee^{\odot i\Phi[\bfs\beta-\bfs\beta_0]} - j_{\bfs\beta}(q_k') \E\,\XX_{\bfs\beta_0}\, \ee^{\odot i\Phi[\bfs\beta-\bfs\beta_0]}.$$
Applying Theorem~\ref{change of beta} again to each term on the right-hand side, we obtain 
$$\E\,\XX_{\bfs\beta_0} J(q_k')\odot \ee^{\odot i\Phi[\bfs\beta-\bfs\beta_0]}= \E\,\XX_{\bfs\beta} J_{\bfs\beta}(q_k') - j_{\bfs\beta}(q_k') \E\,\XX_{\bfs\beta} = \E\, J(q_k') \XX_{\bfs\beta}.$$ 
Sending $q_k'$ to $q_k,$ we have 
$$\E\,\XX_{\bfs\beta_0} J(q_k)\odot \ee^{\odot i\Phi[\bfs\beta-\bfs\beta_0]} = \E\, J(q_k) \XX_{\bfs\beta}.$$ 
The claim \eqref{eq: Ward@q claim} now follows from the above equation and \eqref{eq: Ward@q0}. 
We now finish the proof of the lemma. 
\end{proof}

Combining the above two lemmas, we claim that all fields in $\FF_{\bfs\beta}(D^*)$ satisfy Ward's identity~\eqref{eq: LvWv} in $D_\hol(v) \cap D$ within correlations.
Applying the same arguments as in the proof of Lemma~\ref{*F} to this claim, we extend this claim to $\FF_{\bfs\beta}(D).$

\begin{thm} \label{Ward's identity}
For a non-random local vector field $v,$ all fields in $\FF_{\bfs\beta}(D)$ satisfy local Ward's identity~\eqref{eq: LvWv} in $D_\hol(v) \cap D$ within correlations.
\end{thm}

\subsection{Ward equations on the Riemann sphere} 
In this subsection we express the Virasoro fields $T_{\bfs\beta}$ in terms of Lie derivative within correlations of fields in $\FF_{\bfs\beta}(S).$

Given a meromorphic vector field with poles $\xi_l,$ we define the Ward functional $W_{\bfs\beta}^+(v)$ by
$$\E\,W_{\bfs\beta}^+(v)\XX := - \frac1{2\pi i}\sum_l \oint_{(\xi_l)} v\,\E\,A_{\bfs\beta}\XX$$
for any tensor product $\XX$ of Fock space fields such that the set $S_\XX$ of all nodes of $\XX$ does not intersect the set of poles of $v.$ 

\begin{prop} \label{Ward identities on S}
In correlations with any string of fields in the extended OPE family $\FF_{\bfs\beta}(S),$ we have 
$$W_{\bfs\beta}^+(v) = \LL_v^+.$$
\end{prop}

\begin{proof}
Since the sum of all residues of the meromorphic $1$-differential $\zeta\mapsto v(\zeta)\E\,A_{\bfs\beta}(\zeta)\XX_{\bfs\beta}$ is zero, we have
$$\E\,W_{\bfs\beta}^+(v)\XX_{\bfs\beta} = \frac1{2\pi i} \sum_j \oint_{(z_j)} v\,\E\,A_{\bfs\beta}\XX_{\bfs\beta}+\frac1{2\pi i} \sum_k \oint_{(q_k)} v\,\E\,A_{\bfs\beta}\XX_{\bfs\beta}.$$ 
It follows from Ward's identity (Theorem~\ref{Ward's identity}) that 
$$\frac1{2\pi i} \sum_j \oint_{(z_j)} v\,\E\,A_{\bfs\beta}\XX_{\bfs\beta}+\frac1{2\pi i} \sum_k \oint_{(q_k)} v\,\E\,A_{\bfs\beta}\XX_{\bfs\beta} =\sum_j \E\, \LL_v^+(z_j)\XX_{\bfs\beta}+\sum_k \E\, \LL_v^+(q_k)\XX_{\bfs\beta} = \E\,\LL_v^+\XX_{\bfs\beta},$$
which completes the proof.
\end{proof}

Given $\xi\in\C,$ let us consider the vector fields $k_\xi, v_\xi$ given by 
$$k_\xi(z) = \frac1{\xi-z},\qquad v_\xi(z) = z\frac{\xi+z}{\xi-z} $$
in the identity chart of $\C.$
Ward's equations on the Riemann sphere now follow from the previous theorem.
\begin{cor} \label{Ward equation0}
In the $\wh\C$-uniformization, for any tensor product $\XX_{\bfs\beta}$ of fields in the extended OPE family $\FF_{\bfs\beta},$ we have 
\begin{equation} \label{eq: Ward equation0}
\E\,T_{\bfs\beta}(\xi) \XX_{\bfs\beta} = \E\,T_{\bfs\beta}(\xi)\,\E\,\XX_{\bfs\beta} + \E\, \LL_{k_\xi}^+\XX_{\bfs\beta}
\end{equation}
and 
\begin{equation} \label{eq: Ward equation0'}
2\xi^2\E\,T_{\bfs\beta}(\xi) \XX_{\bfs\beta} = 2\xi^2 \E\, T_{\bfs\beta}(\xi)\,\E\,\XX_{\bfs\beta} + \E\, \LL_{v_\xi}^+\XX_{\bfs\beta}.
\end{equation}
\end{cor}

\subsection{Ward's equations in a simply-connected domain}
To derive Ward's equations in $\H$ and $\D$, we modify Ward's functional to be defined in a simply-connected domain.
Ward's equations in $\H$ (resp. in $\D$) represent the insertion of the Virasoro fields within correlations of fields in $\FF_{\bfs\beta}$ in terms of Lie derivative operators with respect to the chordal (resp. radial) Loewner vector fields.

We first consider the special case that no background charge is placed on the boundary $\pa D.$
The general case can be treated through a suitable limit procedure. 
In this special case, $A_{\bfs\beta}$ is continuous on the boundary. 
(The continuity on the boundary should be understood in terms of standard boundary charts.)
As we mentioned before, we mostly concern ourselves with a symmetric background charge $\bfs\beta.$
In this case, $A_{\bfs\beta}$ is real on the boundary. 
(Again, we understand the real-valuedness on the boundary in terms of standard boundary charts.)

Given a meromorphic vector field $v$ in $D$ (continuous up to the boundary) with poles $\xi_l$'s ($\xi_l\in D$), we define the Ward functional $W_{\bfs\beta}^+(v)$ by
$$W_{\bfs\beta}^+(v) := \lim_{\ve\to 0 }\Big(\frac1{2\pi i} \int_{\pa D_\ve} vA_{\bfs\beta} - \frac1\pi \int_{D_\ve} (\bp v)A_{\bfs\beta}\Big),$$
where $D_\ve = D \setminus \bigcup B(\xi_l,\ve).$
Both integrals are coordinate independent since $vA_{\bfs\beta}$ is a $[1,0]$-differential and $(\bp v)A_{\bfs\beta}$ is a $[1,1]$-differential. Their correlations with Fock space functionals $\XX$ are well-defined provided that any node of $\XX$ is not a pole of $v.$

Somewhat symbolically, the Ward functional $W_{\bfs\beta}^+(v)$ can be represented as 
$$W_{\bfs\beta}^+(v)=\frac1{2\pi i}\int_{\pa D} vA_{\bfs\beta}-\sum_l\frac1{2\pi i}\oint_{(\xi_l)} vA_{\bfs\beta},$$
or 
$$
W_{\bfs\beta}^+(v)=\frac1{2\pi i}\int_{\pa D} vA_{\bfs\beta}-\frac1{\pi }\int_{D} (\bp v)A_{\bfs\beta},$$
where $\bp v$ in the last integral should be understood in the sense of distributions.
If a background charge is placed on the boundary, then the first integral should be taken in the sense of the Cauchy principal value. 

From now on, we only consider the case that $\bfs\beta$ is symmetric.
Then $W_{\bfs\beta} = (A_{\bfs\beta},\overline{A_{\bfs\beta}}).$ 
We now prove Theorem~\ref{main: Ward identities in D}. 

\begin{proof}[Proof of Theorem~\ref{main: Ward identities in D}]
We may consider the case that $\supp\,\bfs\beta^\pm \subseteq D.$
It is enough to show that
$$W_{\bfs\beta}^+(v) = \LL_v^+.$$
in correlations with any string of fields in the extended OPE family $\FF_{\bfs\beta}(D).$

Let $D_j \,(D'_k, D''_l)$ be a sufficiently small open disk centered at $z_j \,(q_k, \xi_l),$ respectively. 
Applying Stokes' theorem to 
$$\dd\big(v(z)\E\,A_{\bfs\beta}(z)\XX_{\bfs\beta}\,\dd z\big) = -\bp\big(v(z)\E\,A_{\bfs\beta}(z)\XX_{\bfs\beta}\big)\,\dd z\wedge \dd \bar z = 2i \bp\big(v(z)\E\,A_{\bfs\beta}(z)\XX_{\bfs\beta}\big) \,\dd x \wedge \dd y$$
over $D \sm (\bigcup_j D_j \cup \bigcup_k D'_k\cup \bigcup_l D''_l),$ we have 
$$\E\,W_{\bfs\beta}^+(v)\XX_{\bfs\beta} = \frac1{2\pi i} \sum_j \oint_{(z_j)} v\,\E\,A_{\bfs\beta}\XX_{\bfs\beta}+\frac1{2\pi i} \sum_j \oint_{(q_k)} v\,\E\,A_{\bfs\beta}\XX_{\bfs\beta}.$$ 
Theorem now follows from Ward's identity (Theorem~\ref{Ward's identity}) as in the proof of Proposition~\ref{Ward identities on S}. 
\end{proof}

We recall the representation of a stress tensor $A$ in terms of Ward's functionals $W_{\bfs\beta}^+(v)$ with the meromorphic vector field $v = k_\zeta (\zeta\in\C).$ 

\begin{prop}[Proposition~5.11 in \cite{KM13}] \label{represent A}
Let $A$ be a holomorphic quadratic differential in $\H,$ and $W=(A,\bar A).$
Suppose $A$ is continuous and real on the boundary (including $\infty$).
Then
\begin{equation} \label{eq: A}
(A\,\|\,\id)(\xi)=W^+(k_\xi)+\overline{W^+(k_{\bar \xi})}.
\end{equation}
\end{prop}

Applying the above proposition to $A_{\bfs\beta}$ in the case that $\supp\,\bfs\beta^\pm \subseteq D,$ we obtain \eqref{eq: A} for $A = A_{\bfs\beta}$ and $W^+ =W^+_{\bfs\beta}.$ 
Combining the above proposition with Theorem~\ref{main: Ward identities in D} and using a limit procedure in the general case that $\supp\,\bfs\beta^\pm \subseteq \overbar{D},$ we derive Ward's equation in $\H.$

\begin{cor}[Ward's equations in $\H$] \label{Ward equationH}
In the $\H$-uniformization, for any tensor product $\XX_{\bfs\beta}$ of fields in the extended OPE family $\FF_{\bfs\beta},$ we have 
\begin{equation} \label{eq: Ward equationH}
\E\,T_{\bfs\beta}(\xi) \XX_{\bfs\beta} = \E\,T_{\bfs\beta}(\xi)\,\E\,\XX_{\bfs\beta} + \E\, (\LL_{k_\xi}^++\LL_{k_{\bar \xi}}^-)\XX_{\bfs\beta}.
\end{equation}
\end{cor}

\begin{eg*} Let $\XX_{\bfs\beta} = X_1(z_1)\cdots X_n(z_n)$ be the tensor product of $[\lambda_j^+,\lambda_j^-]$-differentials $X_j$ in $\FF_{\bfs\beta}.$ 
If $\xi\in\R,$ then we have 
\begin{align*}
\E\,T_{\bfs\beta}(\xi) \XX_{\bfs\beta} &= \E\,T_{\bfs\beta}(\xi)\,\E\,\XX_{\bfs\beta} + \sum_k\Big(\frac{\pa_{q_k}}{\xi-q_k}+\frac{\bp_{q_k}}{\xi-\bar q_k}\Big)\,\E\,\XX_{\bfs\beta}\\
&+ \sum_j\Big(\frac{\pa_j}{\xi-z_j}+\frac{\lambda_j^+}{(\xi-z_j)^2}+\frac{\bp_j}{\xi-\bar z_j}+\frac{\lambda_j^-}{(\xi-\bar z_j)^2}\Big)\,\E\,\XX_{\bfs\beta}.
\end{align*}
\end{eg*}

Given a meromorphic vector field $v$ with a local flow $z_t$ on $\wh\C,$
we define its reflected vector field $v^*$ with respect to the unit circle $\partial \mathbb{D}$ by the vector field of the reflected flow $z_t^*$ of $z_t.$
Indeed, if $z_t$ is the flow of $v$, then $z_t^*= 1/\bar z_t$ is the reflected flow and
its vector field $v^*$ is given by the equation $\dot{z}^* = -\overline{\dot{z}/z^2} = -\overline{v(z)}(z^*)^2.$
Thus we have 
$v^*(z) = -z^2\overline{v({1}/{\bar z})}.$
Using a similar method as in the previous corollary, we obtain the following form of Ward's equations in the radial case. 

\begin{cor}[Ward's equations in $\D$] \label{Ward equationD}
In the $\D$-uniformization, for any tensor product $\XX_{\bfs\beta}$ of fields in the extended OPE family $\FF_{\bfs\beta},$
\begin{equation} \label{eq: Ward equationD}
2\zeta^2\E\,T_{\bfs\beta}(\zeta) \XX_{\bfs\beta} = 2\zeta^2 \E\,T_{\bfs\beta}(\zeta)\,\E\,\XX_{\bfs\beta} + \E\, (\LL_{v_\zeta}^++\LL_{v_{\zeta^*}}^-)\XX_{\bfs\beta}.
\end{equation}
\end{cor}

Let us emphasize that the above formulas apply to the strings without nodes at $\zeta.$
The case where $\zeta$ is a node is discussed in the following subsection. 
We use OPE calculus and local operators.
Theorem~\ref{main: Ward's equations} now follows from Corollaries~\ref{Ward equationH}~--~\ref{Ward equationD} and the next lemma.

\begin{lem} \label{ETZ}
We have
\begin{equation} \label{eq: ETZ in H}
\E\,T_{\bfs\beta}(z) = \PP_{\bfs\beta}^{-1}\LL^+_{k_z}\PP_{\bfs\beta} + \PP_{\bfs\beta}^{-1}\LL^-_{k_{\bar z}}\PP_{\bfs\beta}
\end{equation}
in the identity chart of $\H$
and 
\begin{equation} \label{eq: ETZ in D}
2z^2\E\,T_{\bfs\beta}(z) = \PP_{\bfs\beta}^{-1}\LL^+_{v_z}\PP_{\bfs\beta} + \PP_{\bfs\beta}^{-1}\LL^-_{v_{z^*}}\PP_{\bfs\beta}
\end{equation}
 in the identity chart of $\D.$
\end{lem}
\begin{proof}
Let $\lambda_k^\pm=\lambda(\beta_k^\pm), (\lambda(\sigma) = \sigma^2/2-\sigma b).$
Since $\PP_{\bfs\beta} = C_{(b)}[\bfs\beta]$ is a $[\lambda_j^+,\lambda_j^-]$-differential at $q_j,$
$$\PP_{\bfs\beta}^{-1}\LL^+_{k_z}\PP_{\bfs\beta} = \sum_j \frac{\lambda_j^+}{(z-q_j)^2} + \frac{\pa_j\log \PP_{\bfs\beta}}{z-q_j}, \qquad \PP_{\bfs\beta}^{-1}\LL^-_{k_{\bar z}}\PP_{\bfs\beta} = \sum_j \frac{\lambda_j^-}{(z-\bar q_j)^2} + \frac{\bp_j\log \PP_{\bfs\beta}}{z-\bar q_j}$$
in the identity chart of $\H.$ 
Differentiating $\PP_{\bfs\beta} = C_{(b)}[\bfs\beta],$ we have 
$$\pa_j\log \PP_{\bfs\beta} = \sum_{k\ne j} \frac{\beta_j^+\beta_k^+}{q_j-q_k} + \sum_k \frac{\beta_j^+\beta_k^-}{q_j-\bar q_k}, \qquad \bp_j\log \PP_{\bfs\beta} = \sum_{k\ne j} \frac{\beta_j^-\beta_k^-}{\bar q_j-\bar q_k} + \sum_k \frac{\beta_j^-\beta_k^+}{\bar q_j- q_k}.$$
On the other hand, it follows from \eqref{eq: j_beta} and $\E\,T_{\bfs\beta} = -\frac12 j_{\bfs\beta}^2 + ib j'_{\bfs\beta}$ that 
\begin{align*}
\E\,T_{\bfs\beta}(z) &= \sum _j \bigg(\frac{\lambda_j^+}{(z-q_j)^2} +\frac{\lambda_j^-}{(z-\bar q_j)^2}\bigg) \\
&+ \sum_{j<k} \bigg(\frac{\beta_j^+\beta_k^+}{(z-q_j)(z-q_k)} + \frac{\beta_j^-\beta_j^-}{(z-\bar q_j)(z-\bar q_k)} \bigg) + \sum_{j,k}\frac{\beta_j^+\beta_k^-}{(z-q_j)(z-\bar q_k)}.
\end{align*}
As a rational function in $z,$
\begin{align*}
\sum_j \frac1{z-q_j}&\bigg(\sum_{k\ne j} \frac{\beta_j^+\beta_k^+}{q_j-q_k} + \sum_k \frac{\beta_j^+\beta_k^-}{q_j-\bar q_k}\bigg) + \sum_j \frac1{z-\bar q_j}\bigg( \sum_{k\ne j} \frac{\beta_j^-\beta_k^-}{\bar q_j-\bar q_k} + \sum_k \frac{\beta_j^-\beta_k^+}{\bar q_j- q_k}\bigg) \\
&=\sum_{j<k} \bigg(\frac{\beta_j^+\beta_k^+}{(z-q_j)(z-q_k)} + \frac{\beta_j^-\beta_j^-}{(z-\bar q_j)(z-\bar q_k)} \bigg) + \sum_{j,k}\frac{\beta_j^+\beta_k^-}{(z-q_j)(z-\bar q_k)},
\end{align*}
which shows \eqref{eq: ETZ in H}.
The formula \eqref{eq: ETZ in D} can be shown similarly.
\end{proof}

\subsection{BPZ equations}

Recall the definition of Virasoro generators $L_n\equiv L_n^{\bfs\beta}$: they are operators $X\mapsto L_nX$ acting on fields, $L_n^{\bfs\beta}X=T_{\bfs\beta}*_{(-n-2)}X,$ i.e.,
\begin{equation} \label{eq: Ln}
T_{\bfs\beta}(\zeta)X(z)=
\cdots+\frac{(L_0X)(z)}{(\zeta-z)^2}+\frac{(L_{-1}X)(z)}{\zeta-z}+(L_{-2}X)(z)+\cdots \qquad (\zeta\to z).
\end{equation}
We write $L(z) = L_{-2}(z)$ in $(\H,\infty)$ and 
\begin{equation} \label{eq: L in D}
L(z) = 2z^2L_{-2}(z) + 3zL_{-1}(z) + L_0(z)
\end{equation}
in $(\D,0).$

Recall the Belavin-Polyakov-Zamolodchikov equations (BPZ equations \cite{BPZ84}) in the chordal case with $\bfs\beta = 2b\cdot q,$ see \cite[Proposition~5.13]{KM13}. 
\begin{prop}
Let $Y, X_1,\cdots, X_n\in \FF_{\bfs\beta}$ and let $X$ be the tensor product of $X_j$'s.
Then
$$\E\, (T_{\bfs\beta}*Y)(z)\,X=\E\,Y(z) \LL^+_{k_z}X+\E\,\LL^-_{k_{\bar z}}[Y(z)X], \qquad \bfs\beta = 2b\cdot q,$$
where all fields are evaluated in the identity chart of $(\H,\infty).$
\end{prop}
Somewhat symbolically, we have 
$$L(z) = \check{\LL}^+_{k_z}+\LL^-_{k_{\bar z}},$$
where $ \check{\LL}^+_{k_z}$ means $\LL^+_{k_z}(\H\sm\{z\}),$ e.g., $\check{\LL}^+_{k_z}Y(z) X:=Y(z) \LL^+_{k_z}X .$ 
We can neglect the check mark if we come to terms that we never differentiate at the poles of $v.$ 

The previous proposition can be extended to the general background charge, see the next two theorems.
We only present the proof of BPZ equations in the radial case. 
The chordal case can be proved in a similar way. 

\begin{thm}[BPZ equations in $\H$] \label{BPZinH}
Let $Y, X_1,\cdots, X_n\in \FF_{\bfs\beta}$ and let $X$ be the tensor product of $X_j$'s.
Then
$$\E\, (T_{\bfs\beta}*Y)(z)\,X=\E\,Y(z) \PP_{\bfs\beta}^{-1}\LL^+_{k_z}\PP_{\bfs\beta}X+\E\,\PP_{\bfs\beta}^{-1}\LL^-_{k_{\bar z}}[\PP_{\bfs\beta}Y(z)X],$$
where all fields are evaluated in the identity chart of $\H.$
\end{thm}

\begin{thm}[BPZ equations in $\D$] \label{BPZinD}
Let $Y, X_1,\cdots, X_n\in \FF_{\bfs\beta}$ and let $X$ be the tensor product of $X_j$'s.
Then
$$\E\, L(z)Y(z)\,X=\E\,Y(z)\PP_{\bfs\beta}^{-1} \LL^+_{v_z} \PP_{\bfs\beta} X+\E\,\PP_{\bfs\beta}^{-1} \LL^-_{v_{z^*}}[\PP_{\bfs\beta} Y(z)X],$$
where all fields are evaluated in the identity chart of $\D$ and $L(z)$ is given by \eqref{eq: L in D}. 
\end{thm}

\begin{proof}
Recall Ward's OPE, $\Sing_{\zeta\to z} T_{\bfs\beta}(\zeta) Y(z) = \LL^+_{k_\zeta}Y(z)$ (see \cite[Proposition~5.3]{KM13}). 
It follows immediately that 
\begin{equation} \label{eq: Ward's OPE4Y}
2\zeta^2\Sing_{\zeta\to z} T_{\bfs\beta}(\zeta) Y(z) = \LL^+_{2\zeta^2k_\zeta}Y(z).
\end{equation}
Denote $\delta_\zeta(\eta) = 2\zeta^2 k_\zeta(\eta) - v_\zeta(\eta) = 2\zeta+\eta.$ 
By \eqref{eq: Ln}, we observe that 
$$\LL_{\delta_\zeta}^+Y(z) = \frac1{2\pi i}\oint_{(z)} \delta_\zeta(\eta)T_{\bfs\beta}(\eta)Y(z)\,\dd \eta = L_0Y(z) + (2\zeta+z)L_{-1}Y(z).$$
It follows from \eqref{eq: Ward's OPE4Y} that 
\begin{equation} \label{eq: sing OPE4TY}
2\zeta^2\Sing_{\zeta\to z} T_{\bfs\beta}(\zeta) Y(z)- \LL^+_{v_\zeta}Y(z)=L_0Y(z) + (2\zeta+z)L_{-1}Y(z).
\end{equation}
On the other hand, by Theorem~\ref{main: Ward's equations}, 
$$2\zeta^2\,\E\,T_{\bfs\beta}(\zeta) Y(z)X=\E\,\PP_{\bfs\beta}^{-1}\LL^+_{v_\zeta}\PP_{\bfs\beta}Y(z)X+\E\,\PP_{\bfs\beta}^{-1}\LL^-_{v_{\zeta^*}}\PP_{\bfs\beta}Y(z)X,$$
where the first term on the right-hand side is 
$$\E\,\PP_{\bfs\beta}^{-1}\LL^+_{v_\zeta}\PP_{\bfs\beta}Y(z)X = \E\,\LL^+_{v_\zeta}Y(z)X+ \E\,Y(z)\PP_{\bfs\beta}^{-1}\LL^+_{v_\zeta}\PP_{\bfs\beta}X.$$
Combining the last three equations and subtracting all singular terms, we get
\begin{align*}
\E\, \big(2\zeta^2\Reg_{\zeta\to z}(T_{\bfs\beta}(\zeta) Y(z))X &+ L_0Y(z)X + (2\zeta+z)L_{-1}Y(z)X\big) \\
&= \E\,Y(z)\PP_{\bfs\beta}^{-1}\LL^+_{v_\zeta}\PP_{\bfs\beta}X + \E\,\PP_{\bfs\beta}^{-1}\LL^-_{v_{\zeta^*}}\PP_{\bfs\beta}Y(z)X,
\end{align*}
where $\Reg_{\zeta\to z}T_{\bfs\beta}(\zeta) Y(z)$ is the regular part of the operator product expansion. 
Theorem now follows by taking the limit $\zeta\to z$ in both sides. 
For instance, the left-hand side converges to $\E\, L(z)Y(z)\,X.$
\end{proof}

\subsection{Null vectors}

After we briefly review the level two degeneracy equations for current primary fields, we apply these equations and the BPZ equations to the (tensor/OPE) products of the one-leg operators and the fields in the extended OPE family.

We denote by $\{ J_n\}$ and $\{ L_n\}$ the current generators and the Virasoro generators, the modes of $J_{\bfs\beta}$ and $T_{\bfs\beta}$ in $\FF_{\bfs\beta}$ theory, respectively: 
$$J_n(z):=\frac1{2\pi i}\oint_{(z)}(\zeta-z)^{n} J_{\bfs\beta}(\zeta)\,\dd\zeta,\qquad L_n(z):=\frac1{2\pi i}\oint_{(z)}(\zeta-z)^{n+1} T_{\bfs\beta}(\zeta)\,\dd\zeta.$$
As operators acting on fields in $\FF_{\bfs\beta},$ we have the following equations:
$$[J_m,J_n] = n\delta_{m+n,0}, \quad [L_m,J_n] = -nJ_{m+n}+ibm(m+1)\delta_{m+n,0},$$
and
$$[L_m, L_n]=(m-n)L_{m+n}+\frac c{12}m(m^2-1)\delta_{m+n,0},$$
where $c$ is the central charge, $c=1-12b^2.$

By definition, $X\in\FF_{\bfs\beta}$ is (Virasoro) primary if $X$ is a $[\lambda,\lambda_*]$-differential; equivalently
\begin{equation} \label{eq: T primary} 
L_{\ge1}X=0,\qquad L_0X=\lambda X, \qquad L_{-1}X=\partial X,
\end{equation}
and similar equations hold for $\bar X.$
Here $L_{\ge k}X = 0$ means that $L_nX = 0$ for all $n\ge k.$ 
A (Virasoro) primary field $X$ is called \emph{current primary} with charges $\sigma$ and $\sigma_*$ if 
$J_{\ge1}X = J_{\ge1}\bar X = 0,$
and
$$J_0X = -i\sigma X, \quad J_0\bar X = i\bar \sigma_*\bar X.$$
Let us recall the characterization of level two degenerate current primary fields.

\begin{prop}[Proposition~11.2 in \cite{KM13}] \label{degeneracy2} 
Let $\OO$ be a current primary field in $\FF_{\bfs\beta},$ and let $\sigma,\sigma_*$ be charges of $\OO.$
If $2\sigma(b+\sigma) = 1,\, \eta = -1/(2\sigma^2),$ then
$$(L_{-2}+\eta L_{-1}^2)\OO =0.$$
\end{prop}

If the parameters $a$ and $b$ are related to the SLE parameter $\kappa$ as 
$$a = \pm\sqrt{2/\kappa},\qquad b = a(\kappa/4-1),$$
then $2a(a+b) = 1$ and the formal fields $\OO^{(a)}, \OO^{(-a-b)}$ satisfy
\begin{equation}\label{eq: level2}
L_{-2}\, \OO^{(a)} =\frac\kappa4\,\pa_z^2\OO^{(a)},\qquad L_{-2}\, \OO^{(-a-b)} =\frac4\kappa\,\pa_z^2\OO^{(-a-b)}
\end{equation}
(in all charts).
Since $\OO^{(a)}, \OO^{(-a-b)}$ are formal fields, we need to assume the neutrality condition $(\NC_0)$ on $\OO^{(a)}X, \OO^{(-a-b)}X$ for some product $X$ of formal fields. 
For example, the one-leg operator $\leg_p$ (both in the chordal case and in the radial one) satisfy the level two degeneracy equation.
As mentioned before, we will use this level two degeneracy equation for $\leg_p$ to establish the connection between the chordal/radial SLE theory and conformal field theory.

\subsubsec{Standard chordal case with $\bfs\beta = 2b\cdot q$}
Recall the BPZ equations in $(\H,\infty):$
$$L(z) = \check{\LL}^+_{k_z}+\LL^-_{k_{\bar z}},\qquad k_\zeta(\eta) = \frac1{\zeta-\eta},$$
which we apply to the field $Y(z)X$ in $\FF_{\bfs\beta},$ where $X$ has no node at $z.$
In all applications $Y$ is holomorphic, so we have
\begin{equation} \label{eq: BPZ in H}
\big(L(z)Y(z)\big)X = Y(z)(\check{\LL}^+_{k_z}+\LL^-_{k_{\bar z}})X = Y(z)(\LL^+_{k_z}+\LL^-_{k_{\bar z}})X.
\end{equation}
If $X$ is also holomorphic (in particular, in the ``boundary" situation), then \eqref{eq: BPZ in H} becomes
$$\big(L(z)Y(z)\big)X = Y(z)\LL^+_{k_z}X.$$
Let us discuss both sides in the equation~\eqref{eq: BPZ in H}.

We study the cases that $Y$ is the null vector 
$$Y(z) = \OO^{(a)}(z) , \qquad Y(z) = \OO^{(-a-b)}(z).$$
Of course, $Y$ is a formal field, so we need to assume the neutrality condition $(\NC_0)$ in $Y(z)X.$
Recall that the conformal dimensions $h,h'$ of $\OO^{(a)},\OO^{(-a-b)},$ respectively at $z$ are 
\begin{equation} \label{eq: h,h'}
h := \frac{3a^2}2-\frac12 = \frac{6-\kappa}{2\kappa}, \qquad h' := \frac{3}{8a^2}-\frac12 = \frac{3\kappa-8}{16}.
\end{equation}

Recall the level two degeneracy equation \eqref{eq: level2} for the formal field $\OO^{(a)}$.
In the boundary case $z= x\in\R$ we have 
\begin{equation}\label{eq: level2chordal}
L_{-2}\, \OO^{(a)}(x) =\frac\kappa4\,\pa_x^2\OO^{(a)}(x)
\end{equation}
in the $(\H,\infty)$-uniformization because $\pa_x = \pa + \bp$ and $\OO^{(a)}$ is holomorphic. 

In the case that $X = \OO_{\bfs\beta}[\bfs\tau]$ with $\bfs\tau = \sum\tau_j\cdot z_j,$ we have 
$$Y(z)\LL^+_{k_z}X = \big(\sum k_z(z_j)\pa_{z_j}+\lambda_j\pa_{z_j}k_z(z_j)\big)Y(z)X,$$
where $\lambda_j = \lambda_b(\tau_j) = \frac12\tau_j^2 - b\tau_j.$

\subsubsec{Standard radial case with $\bfs\beta = b\cdot q + b\cdot q^*$} \label{ss: BPZ in D}
We apply the BPZ equation to the tensor product $Y(z)X$ of fields in $\FF_{\bfs\beta}(D,q)$ with holomorphic $Y:$
\begin{equation} \label{eq: BPZ in D}
\big(L(z)Y(z)\big)X = Y(z)\PP_q^{-1}(\LL^+_{v_z}+\LL^-_{v_{z^*}})\PP_qX, \qquad v_\zeta(\eta) = \eta\frac{\zeta+\eta}{\zeta-\eta}.
\end{equation}
Let us discuss both sides in the equation~\eqref{eq: BPZ in D}.

We consider the $(\D,0)$-uniformization. 
For $Y=\OO^{(a)},$ we have 
\begin{equation}\label{eq: level2radial}
L(z)\OO^{(a)}=\Big(\frac\kappa2\,z^2\pa_z^2 + 3z\pa_z + h\Big)\OO^{(a)},
\end{equation}
where $h$ is given by \eqref{eq: h,h'}.
In particular, for $z = \ee^{i\theta} \in \pa \D,$ 
$$L(\ee^{i\theta}) \OO^{(a)} = \Big(\!-\kappa\big(\frac12\pa_\theta^2 + ih\pa_\theta\big)+h\Big)\OO^{(a)}.$$
In the case of $O^{(-a-b)}$ the corresponding differential operator is
$$L(\ee^{i\theta}) = -\kappa'\big(\frac12\pa_\theta^2 + ih'\pa_\theta\big)+h',$$
where $\kappa'=16/\kappa$ and $h'$ is given by \eqref{eq: h,h'}.

By Lemma~\ref{ETZ}, the right hand-side becomes
$$Y(z)(\LL^+_{v_z}+\LL^-_{v_{z^*}})X+2z^2Y(z)X\, \E\,T_{\bfs\beta}(z).$$

\section{Connection to the theory of forward chordal/radial SLE with forces and spins} 

In Subsections~\ref{ss: 1-leg spin chordal}~--~\ref{ss: 1-leg radial} we define the one-leg operators for chordal/radial SLEs with forces and spins as the OPE exponentials with specific background charge $\bfs\beta$ at the marked points. 
The insertion of a one-leg operator changes the boundary values of Fock space fields. 
As the PDEs for correlation functions of the fields in the extended OPE family $\FF_{\bfs\beta}(D),$ the BPZ-Cardy equations for such insertion operations are used to prove the connection (Theorem~\ref{main}) between the theory $\FF_{\bfs\beta}$ and the chordal/radial $\SLE[\bfs\beta]$ theory.
A dipolar theory studied in \cite{KT13} is the special of this theory since it is well known (\cite{SW05,Zhan08}) that dipolar $\SLE(\kappa)$ with two force points $q_\pm$ is equivalent to chordal $\SLE(\kappa,\bfs\rho)$ with $\bfs\rho = \frac12(\kappa-6)\cdot q_+ +\frac12(\kappa-6)\cdot q_-.$ 
In Subsection~\ref{ss: Restriction formulas} we present field theoretic approach to chordal $\SLE(\kappa,\bfs\rho)$ restriction martingale-observable introduced in \cite[Lemma~4]{Dubedat05}.

\subsection{One-leg operators for chordal SLEs with forces and spins} \label{ss: 1-leg spin chordal}
Let $(D,p,q)$ be a simply-connected domain with two distinct marked boundary points $p,q\in\pa D,$ 
and $\bfs\beta$ be a symmetric background charge on $S = D^\mathrm{double}$ with 
$$\bfs\beta = a\cdot p + \bfs\eta, \quad p\notin\supp\,\bfs\eta $$
satisfying the neutrality condition $(\NC_b):$ 
$$\int \bfs\beta = 2b.$$
Here the parameters $a$ and $b$ are related to the SLE parameter $\kappa$ as 
$$a = \pm\sqrt{2/\kappa},\qquad b = a(\kappa/4-1).$$
Let $\bfs\tau = a\cdot p - a \cdot q.$ 
For any divisor $\bfs\sigma,$ we denote $\check{\bfs\sigma\,} = \bfs\sigma -\bfs\tau$ and $\hat{\bfs\sigma\,} = \bfs\sigma +\bfs\tau.$ 

\subsubsec{Insertions}
The insertion of the one-leg operator 
$$\leg_p \equiv \leg(p) := \OO_{\check{\bfs\beta\,}}[\bfs\tau]$$
produces an operator $\XX\mapsto \wh\XX$ acting on Fock space functionals/fields by the rules~\eqref{eq: BCrules} and the formula 
$$\wh\Phi_{\check{\bfs\beta\,}}(z) = \Phi_{\check{\bfs\beta\,}}(z) + 2a \arg w(z) = \Phi_{\bfs\beta}(z),$$
where $w:(D,p,q)\to(\H,0,\infty)$ is a conformal map.

We denote 
$$\wh\E\,\XX := \frac{\E\,\leg_p\XX}{\E\,\leg_p} = \frac{\E\,\leg_p^\eff\XX}{\E\,\leg_p^\eff} = \E\, \XX \ee^{\odot i a \Phi^+(p,q)},$$
where $\leg_p^\eff$ is the effective one-leg operator 
$$\leg_p^\eff = \PP_{\check{\bfs\beta\,}}\leg_p = C_{(b)}[\check{\bfs\beta\,}] \OO_{\check{\bfs\beta\,}}[\bfs\tau] = C_{(b)}[\bfs\beta] V^{\odot}[\bfs\tau] = V_{\check{\bfs\beta\,}}[\bfs\beta].$$ 
Theorem~\ref{main: change of beta} says that 
$$\wh\E\,\XX = \E\, \wh\XX$$
for any string $\XX$ of fields in $\FF_{\check{\bfs\beta\,}}.$
We define the partition function associated with a symmetric background charge $\bfs\beta$ by 
$$Z_{\bfs\beta}:= \big|C_{(b)}[\bfs\beta]\big|.$$ 
In the $\H$-uniformization, $C_{(b)}[\bfs\beta]$ is non-negative (up to a phase), see Example (b) in Subsection~\ref{ss: C} and thus we have $Z_{\bfs\beta}= C_{(b)}[\bfs\beta].$ 

\subsubsec{BPZ equations}
Let $\xi\in\pa D$ and $\bfs\beta_\xi = \bfs\beta - a\cdot p + a\cdot \xi.$
Let $X$ be the tensor product $X= X_1(z_1)\cdots X_n(z_n)$ of fields $X_j$ in $\FF_{\check{\bfs\beta\,}}$ $(z_j\in D).$
It follows from Theorem~\ref{BPZinH} that 
\begin{equation} \label{eq: BPZinH} 
\frac1{2a^2} \pa_\xi^2 \E\leg_\xi^\eff X = \E\, \check{\LL}_{k_\xi} \leg_\xi^\eff X, \qquad k_\xi(z):=\frac1{\xi-z},
\end{equation}
in the $\H$-uniformization.
Here $\pa_\xi = \pa + \bar\pa$ is the operator of differentiation with respect to the real variable $\xi$ and $\check{\LL}_{k_\xi}$ is taken over the finite notes of $X$ and $\supp\,\bfs\beta_\xi\setminus\{\xi\}.$
In particular, the partition function $Z_\xi\equiv Z_{\bfs\beta_\xi}$ satisfies the null vector equation
\begin{equation} \label{eq: chordal NV}
\frac1{2a^2} \pa_\xi^2 Z_\xi = \check{\LL}_{k_\xi} Z_\xi,
\end{equation}
in the $\H$-uniformization since $Z_\xi = \E\leg_\xi^\eff$ up to a phase in the $\H$-uniformization.
Here $\check{\LL}_{k_\xi}$ is taken over $\supp\,\bfs\beta_\xi\setminus\{\xi\}.$

\subsubsec{BPZ-Cardy equations}
For $\xi\in\R$ and a tensor product $X = X_1(z_1)\cdots X_n(z_n)$ of fields $X_j$ in $\FF_{\check{\bfs\beta\,}}$ $(z_j\in\H)$, we denote
$$R_\xi(\bfs z,\bfs q) \equiv\wh\E_\xi X:= \frac{\E\,\leg_\xi X}{\E\,\leg_\xi}=\frac{\E\,\leg_\xi^\eff X}{\E\,\leg_\xi^\eff}= \E\, X \ee^{\odot i a \Phi^+(\xi,q)}.$$

The following is the first half of Theorem~\ref{main: BPZ-Cardy}.

\begin{prop}[BPZ-Cardy equations in $\H$] \label{prop: BPZ-Cardy4chordalSLE(kappa,rho)}
In the identity chart of $\H,$ we have
\begin{equation} \label{BPZ_Cardy}
\frac1{2a^2}\big(\pa_\xi^2\wh\E_{\xi}\,X + 2(\pa_\xi \log Z_\xi)\pa_\xi\wh\E_{\xi}\,X \big) = \check \LL_{k_\xi}\wh\E_{\xi}\,X, \qquad k_\xi(z):=\frac1{\xi-z},
\end{equation}
where $\pa_\xi = \pa + \bp$ and $\check{\LL}_{k_\xi}$ is taken over the finite notes of $X$ and $\supp\,\bfs\beta_\xi\setminus\{\xi\}.$ 
\end{prop}

\begin{proof}
By differentiation,
$$\frac{\pa_\xi^2\E\,\leg_\xi^\eff X}{\E\,\leg_\xi ^\eff} =\pa_\xi^2\wh\E_{\xi}\,X + 2\frac{\pa_\xi Z_\xi}{Z_\xi} \pa_\xi\wh\E_{\xi}\,X + \frac{\pa_\xi^2 Z_\xi}{Z_\xi}\wh\E_{\xi}\,X.$$
On the other hand, it follows from the BPZ equations \eqref{eq: BPZinH} that the left-hand side of the above becomes $2a^2{\E\, \check{\LL}_{k_\xi} \leg_\xi^\eff X}/{\E\,\leg_\xi^\eff }.$
By Leibniz's rule for Lie derivatives, we have 
$$\frac{\E\, \check{\LL}_{k_\xi} \leg_\xi^\eff X}{\E\,\leg_\xi^\eff }= \check{\LL}_{k_\xi} \wh\E_\xi\,X + \frac{\check{\LL}_{k_\xi}Z_\xi}{Z_\xi}\, \wh\E_\xi\,X.$$
Combining all of the above, we find
$$\frac1{2a^2}\big(\pa_\xi^2\wh\E_{\xi}\,X + 2\frac{\pa_\xi Z_\xi}{Z_\xi} \pa_\xi\wh\E_{\xi}\,X + \frac{\pa_\xi^2 Z_\xi}{Z_\xi}\wh\E_{\xi}\,X\big) = \check{\LL}_{k_\xi} \wh\E_\xi\,X + \frac{\check{\LL}_{k_\xi}Z_\xi}{Z_\xi}\, \wh\E_\xi\,X.$$
We remark that the coefficient of $\wh\E_{\xi}\,X$ vanishes since the partition function $Z_\xi$ satisfies the null vector equation~\eqref{eq: chordal NV}.
\end{proof}

\subsection{One-leg operators for radial SLEs with forces and spins} \label{ss: 1-leg radial}
Let $(D,p,q)$ be a simply-connected domain $D$ with a marked boundary point $p\in\pa D,$ and a marked interior point $q\in D.$ 
As in the previous subsection, we consider a symmetric background charge $\bfs\beta$ on $S = D^\mathrm{double}$ with 
$$\bfs\beta = a\cdot p + \bfs\eta, \quad p\notin\supp\,\bfs\eta $$
satisfying the neutrality condition $(\NC_b).$

Let $\eta\in\R.$ We introduce the one-leg operator $\leg_p^s$ with spin $s = i\eta {a^2}/2$ as 

$$\leg_p^s \equiv \OO_{\check{\bfs\beta\,}}[\bfs\tau], \qquad \bfs\tau = a\cdot p -\frac{a+i\delta}2\cdot q -\frac{a-i\delta}2\cdot q^*,\qquad \delta = \eta a,$$
where $\check{\bfs\beta\,} = \bfs\beta - \bfs\tau$ is a symmetric background charge.

\subsubsec{Insertions}
As in the chordal theory, the insertion of the one-leg operator $\leg_p^s$
produces an operator $\XX\mapsto \wh\XX$ acting on Fock space functionals/fields by the rules~\eqref{eq: BCrules} and the formula 
$$\wh\Phi_{\check{\bfs\beta\,}}(z) = \Phi_{\bfs\beta}(z).$$
We denote 
$$\wh \E[\XX] :=\frac{\E [\leg_p^s\XX]}{\E [\leg_p^s]}=\frac{\E [(\leg_p^s)^\eff\XX]}{\E [(\leg_p^s)^\eff]}=\E\,\XX \ee^{\odot i\Phi[\bfs\tau]},$$
where $(\leg_p^s)^\eff$ is the effective one-leg operator with spin $s$:
$(\leg_p^s)^\eff = \PP_{\check{\bfs\beta\,}}\leg_p^s.$
Then we have 
$\wh \E[\XX]=\E[\wh\XX].$

\begin{egs*} We consider the case 
$$\bfs\beta^s =a \cdot p + \Big(b-\frac {a+i\delta}2\Big)\cdot q+ \Big(b-\frac {a-i\delta}2\Big)\cdot q^*,\qquad \delta = \eta a.$$
Let $\bfs\beta = \bfs\beta^0.$ 

\ss \no (a) The bosonic $\Phi_{\bfs\beta^s}$ is a pre-pre-Schwarzian form of order $(ib,-ib),$
$$\Phi_{\bfs\beta^s} = \Phi_{\bfs\beta} - \delta\log|w|,$$
where $w:(D,p,q)\to(\D,1,0)$ is a conformal map;

\ss \no (b) The current $J_{\bfs\beta^s}$ is a pre-Schwarzian form of order $ib,$
$$J_{\bfs\beta^s} = J_{\bfs\beta} - \frac{\delta}2\frac{w'}{w};$$

\ss \no (c) The Virasoro field $ T_{\bfs\beta^s}$ is a Schwarzian form of order $\frac1{12}c,$
\begin{align*} 
T_{\bfs\beta^s} &= -\dfrac12 J_{\bfs\beta^s}* J_{\bfs\beta^s} + ib\pa J_{\bfs\beta^s}\\
&= A - j_{\bfs\beta^s} J + ib\pa J + \frac{c}{12}S_w + h_{1,2}\frac{w'^2}{w(1-w)^2} + h_{0,1/2}^s\frac{w'^2}{w^2} + s \frac{w'^2}{w(1-w)},
\end{align*}
where $j_{\bfs\beta^s}=\E\,J_{\bfs\beta^s},$ $h_{1,2}=\frac12a^2-ab$ and 
$h_{0,1/2}^s=\frac18 (a+i\delta)^2-\frac12b^2.$
\end{egs*}

Let $X$ be a tensor product of fields in the extended family $\FF_{\check{\bfs\beta\,}}(D).$
We assume that $\zeta$ is not a node of $X.$ 
Denote 
$$R_\zeta \equiv \wh\E_\zeta X= \frac{\E\,(\leg^s_\zeta)^\eff X}{\E\,(\leg^s_\zeta)^\eff}, \quad C_\zeta = \E\,(\leg^s_\zeta)^\eff = C_{(b)}[\bfs\beta_\zeta], \quad Z_\zeta = |C_\zeta|,$$
where $\bfs\beta_\zeta = \bfs\beta - a\cdot p + a\cdot \zeta.$
The following is the second half of Theorem~\ref{main: BPZ-Cardy}.

\begin{prop}[BPZ-Cardy equations in $\D$] \label{Cardy4spin}
In the $(\D,0)$-uniformization we have
$$-\frac2{a^2}\Big(\frac12\pa_\theta^2 + \big(\pa_\theta\log Z_\zeta \big) \pa_\theta\Big)\wh\E_\zeta X = \check{\LL}_{v_\zeta}\wh\E_\zeta X, \qquad(\zeta = \ee^{i\theta}, \theta\in\R),$$
where the Lie derivative $\check{\LL}_{v_\zeta}$ is taken over the finite notes of $X$ and $\supp\,\bfs\beta_\zeta\setminus\{\zeta\}.$
\end{prop}

\begin{proof}
Let $\zeta = \ee^{i\theta} \in \pa D.$
Applying the BPZ equation in $(\D,0)$ to $(\leg^s_\zeta)^\eff X,$ we have 
$$\E\,L(\zeta)(\leg^s_\zeta)^\eff X = \check{\LL}_{v_\zeta}\E\,(\leg^s_\zeta)^\eff X,$$
where $L$ is the differential operator in Subsection~\ref{ss: BPZ in D}:
$$L(z) = \frac\kappa2 z^2\pa_z^2 + 3z\pa_z + h, \textrm{ or } L(\ee^{i\theta}) = -\kappa\big(\frac12\pa_\theta^2 + ih\pa_\theta\big)+h.$$
In particular, $C_\zeta = \E\,(\leg^s_\zeta)^\eff$ satisfies the null vector equation
\begin{equation}\label{eq: radial NV}
L(\zeta)C_\zeta = \check{\LL}_{v_\zeta}C_\zeta .
\end{equation}
By differentiation, 
$$\frac{\E\,L(\zeta)(\leg^s_\zeta)^\eff X }{\E\,(\leg^s_\zeta)^\eff } = (L(\zeta)-h) \wh\E_\zeta X + \kappa\,\zeta^2\frac{\pa_\zeta C_\zeta}{C_\zeta} \pa_\zeta \wh\E_\zeta X + \frac{L(\zeta) C_\zeta}{C_\zeta} \wh\E_\zeta X.$$
It follows from Leibniz rule for Lie derivative that
$$\frac{\check{\LL}_{v_\zeta}\E\,(\leg^s_\zeta)^\eff X}{\E\,(\leg^s_\zeta)^\eff }=\check{\LL}_{v_\zeta}\wh\E_\zeta X +\frac{\check{\LL}_{v_\zeta}C_\zeta}{C_\zeta}\,\wh\E_\zeta X.$$
Combining all of the above, we have 
\begin{equation} \label{eq: Cardy4spin0}
-\kappa\Big(\frac12\pa_\theta^2 + i\big(h + \zeta\frac{\pa_\zeta C_\zeta}{C_\zeta}\big)\pa_\theta\Big)\wh\E_\zeta X = \check{\LL}_{v_\zeta}\wh\E_\zeta X
\end{equation}
in the $(\D,0)$-uniformization. 
Note that the coefficient of $\wh\E_\zeta X$ vanishes since $C_\zeta$ satisfies the null vector equation~\eqref{eq: radial NV}.

Let 
$$\bfs\beta_\zeta^0 =a \cdot \zeta + \Big(b-\frac {a}2\Big)\cdot q+ \Big(b-\frac {a}2\Big)\cdot q^*,\qquad C_\zeta^0 = C_{(b)}[\bfs\beta_\zeta^0].$$
Then we have 
$$\zeta\frac{\pa_\zeta C_\zeta^0}{C_\zeta^0} = -h, \quad h = h_{1,2} := \frac12a^2-ab.$$
and rewrite \eqref{eq: Cardy4spin0} as
$$-\kappa\Big(\frac12\pa_\theta^2 +\big(\pa_\theta \log\frac {C_\zeta}{C^0_\zeta}\big)\pa_\theta\Big)\wh\E_\zeta X = \check{\LL}_{v_\zeta}\wh\E_\zeta X.$$
It is easy to see that 
$$\pa_\theta\log Z_\zeta = \pa_\theta \log\frac {C_\zeta}{C^0_\zeta}$$
in the $\D$-uniformization. 
\end{proof}

\subsection{Martingale-observables for chordal/radial SLE with forces and spins} \label{ss: proof of main theorem}
We now prove Theorem~\ref{main}.

\begin{proof}[Proof of Theorem~\ref{main}]
We first consider the chordal case.
Let $g_t$ be the chordal $\SLE[\bfs\beta]$ map driven by the real process $\xi_t:$ 
$$\dd\xi_t = \sqrt\kappa\, \dd B_t + \lambda(t)\,\dd t, \quad \lambda(t) = (\lambda\,\|\,g_t^{-1}), \quad \lambda = \kappa\, \pa_\xi \log Z_\xi,$$
Then we have 
$$M_t = m(\xi_t,t), \quad m(\xi,t) = (R_\xi \,\|\, g_t^{-1}).$$
It follows from It\^o's formula that 
\begin{align*}
\dd M_t & = \sqrt\kappa\, \pa_\xi\big|_{\xi =\xi_t} m(\xi,t)\,\dd B_t + \kappa\, \pa_\xi\big|_{\xi =\xi_t} m(\xi,t) \,\pa_\xi\big|_{\xi =\xi_t} (\log Z_\xi \,\|\, g_t^{-1})\,\dd t \\
&+\frac\kappa2 \, \pa_\xi^2\big|_{\xi =\xi_t} m(\xi,t)\,\dd t + L_t\,\dd t,
\end{align*}
where 
$$L_t := \frac {\dd}{\dd s}\Big|_{s=0} \big(R_{\xi_t}\,\|\,g_{t+s}^{-1}\big).$$
At each time $t,$ $f_{s,t} := g_{t+s}\circ g_t^{-1}$ satisfy the differential equations 
$$\frac{\dd}{\dd s}f_{s,t}(\zeta) = \frac{2}{f_{s,t}(\zeta)-\xi_{t+s}}.$$
As $s\to 0,$
$$ f_{s,t} = \id -2sk_{\xi_t} + o(s), \qquad k_{\xi}(z)=\frac1{\xi-z}.$$
In terms of the time-dependent flows $f_{s,t},$ $L_t$ is rewritten as 
$$L_t=\frac {\dd}{\dd s}\Big|_{s=0}\big(R_{\xi_t}\,\|\,g_t^{-1}\circ f_{s,t}^{-1}\big)=-2\big(\check\LL_{k_{\xi_t}}R_{\xi_t}\,\|\,g_{t}^{-1}\big).$$
The last equality follows from the fact any fields in $\FF_{\bfs\beta}$ depend smoothly on local charts.
Thus the drift term of $\dd M_t$ simplifies to 
$$\Big(\frac\kappa2 \, \pa_\xi^2\big|_{\xi =\xi_t} m(\xi,t) +
\kappa\, \pa_\xi\big|_{\xi =\xi_t} m(\xi,t) \frac{\pa_\xi|_{\xi=\xi_t} (Z_\xi\,\|\,g_t^{-1})}{(Z_\xi\,\|\,g_t^{-1})} -2\big(\check{\LL}_{k_{\xi_t}}R_{\xi_t}\,\|\,g_{t}^{-1}\big)\Big)\,\dd t.
$$
It vanishes by the BPZ-Cardy equations (Proposition~\ref{prop: BPZ-Cardy4chordalSLE(kappa,rho)}). 

\medskip
Next, we consider the radial case. 
For $\XX = X_1(z_1)\cdots X_n(z_n), X_j\in\FF_{\check{\bfs\beta\,}}(D),$ denote 
$$R_\zeta (z_1,\cdots, z_n)\equiv\wh\E_\zeta [X_1(z_1)\cdots X_n(z_n)].$$
Then the process $M_t(z_1,\cdots, z_n)$
is represented by 
$$M_t = m(\zeta_t,t), \qquad m(\zeta,t) =\big(R_\zeta \,\|\,g_t^{-1}\big),$$
where $g_t:(D_t,\gamma_t,q)\to(\D,\zeta_t,0)$ is the radial $\SLE[\bfs\beta]$ map driven by the real process $\theta_t:$
$$\dd\theta_t = \sqrt\kappa\, \dd B_t + \lambda(t)\,\dd t, \quad \lambda(t) = (\lambda\,\|\,g_t^{-1}), \quad \lambda = \kappa \,\pa_\theta \log Z_{\bfs\beta_\zeta}, \quad \zeta_t = \ee^{i\theta_t}.$$
Using a similar argument in the chordal case, we find the drift term of $\dd M_t$ as 
$$\Big(\frac\kappa2\pa_\theta^2+\kappa\frac{\pa_\theta (Z_\zeta\,\|\,g_{t}^{-1})}{(Z_\zeta\,\|\,g_{t}^{-1})}\big)\pa_\theta\Big)\Big|_{\zeta=\zeta_t}~m(\zeta,t)\,\dd t\\
+\ \big(\check\LL_{v_{\zeta_t}}R_{\zeta_t}\,\|\,g_{t}^{-1}\big) \dd t.$$
By the BPZ-Cardy equations (Proposition~\ref{Cardy4spin}) in the radial case, $M_t$ is driftless. 
\end{proof}

\begin{eg*}
Non-chiral vertex fields
$$\VV_{\bfs\beta}^{(\sigma)}:=\ee^{*i\sigma\Phi_{\bfs\beta}}=\OO_{\bfs\beta}[\sigma\cdot z, - \sigma\cdot z]$$
have real conformal dimensions $[\lambda^+,\lambda^-]$ at $z:$
$$\lambda^+ = \frac{\sigma^2}2-\sigma b, \qquad \lambda^- = \frac{\sigma^2}2+\sigma b.$$
If the conformal spin $s:=\lambda^+-\lambda^-$ is 1, then $\VV_{\bfs\beta}^{(\sigma)}$ transforms as a (formal) vector field and its flow lines can be identified with $\SLE[\bfs\beta],$ see \cite{MS16} for a chordal case and \cite{MS17} for a radial case with spin at $q.$ 
\end{eg*}

\subsection{Examples of radial SLE martingale-observables} \label{ss: EgsMO}
In this subsection we present examples of radial SLE martingale-observables including Schramm-Sheffield's observables, Friedrich-Werner's formula, and the restriction formula in the standard radial case with the following background charge $\bfs\beta$: 
\begin{equation}\label{eq: radial beta}
\bfs\beta =a \cdot p + \Big(b-\frac a2\Big)\cdot q+ \Big(b-\frac a2\Big)\cdot q^*,
\end{equation}

\begin{eg*}[Schramm-Sheffield's observables] In the standard chordal case with 
$$\bfs\beta =a \cdot p + (2b-a)\cdot q,$$
the 1-point functions of the bosonic fields
$$\varphi_{\bfs\beta}(z) = \E[\Phi_{\bfs\beta}(z)] = 2a\arg w(z) -2b\arg w'(z), \quad w:(D,p,q)\to(\H,0,\infty)$$
were introduced as SLE martingale-observables by Schramm and Sheffield, see \cite{SS13}.
Similarly, the 1-point functions $\varphi_{\bfs\beta}= \E\,\Phi_{\bfs\beta}$ (with standard radial background charge $\bfs\beta$ in \eqref{eq: radial beta}) of the bosonic fields are martingale-observables of radial SLEs:
$$\varphi_{\bfs\beta}(z) = \E[\Phi_{\bfs\beta}(z)] = 2a\arg (1-w(z)) -a\arg w(z) -2b\arg \frac{w'(z)}{w(z)},$$
where $w:(D,p,q)\to(\D,1,0)$ is a conformal map.
By It\^o's calculus, we have 
\begin{align*}
\varphi_t(z) &= \E[\Phi_{D_t,\gamma_t,q}(z)]= 2a\arg (1-w_t(z)) -a\arg w_t(z) -2b\arg \frac{w_t'(z)}{w_t(z)} \\
&= 2a\arg (1-w(z)) -a\arg w(z) -2b\arg \frac{w'(z)}{w(z)} +\sqrt2\int_0^t\Re\, \frac{1+w_s(z)}{1-w_s(z)}\,\dd B_s.
\end{align*}
One can use the $2$-point martingale-observables
$$\E[\Phi_{\bfs\beta}(z_1)\Phi_{\bfs\beta}(z_2)]=2G(z_1,z_2)+\varphi_{\bfs\beta}(z_1)\varphi_{\bfs\beta}(z_2)$$
or Hadamard's variation formula
\begin{equation} \label{eq: Hadamard}
\dd G_{D_t}(z_1,z_2)=-\,\Re\,\frac{1+w_t(z_1)}{1-w_t(z_1)}\,\Re\,\frac{1+w_t(z_2)}{1-w_t(z_2)}\,\dd t=-\frac12\,\dd \langle \varphi_{\bfs\beta}(z_1),\varphi_{\bfs\beta}(z_2)\rangle_t
\end{equation}
to construct a coupling of radial SLE and the Gaussian free field such that
$$\E[\, \Phi_{D,p,q}\,|\,\gamma[0,t]\,]= \Phi_{D_t,\gamma_t,q},$$
see \cite{Dubedat09}.
\end{eg*}

Let us recall the \emph{restriction property} of radial $\SLE(8/3)$ (see \cite[Section~6.5]{Lawler05}):
{\setlength{\leftmargini}{1.7em}
\begin{itemize}
 \item the law of $\SLE(8/3)$ in $\D$ conditioned to avoid a fixed hull $K$ is identical to the law of $\SLE(8/3)$ in $\D\sm K;$
 \item equivalently, there exist $\lambda$ and $\mu$ such that for all $K,$
$$\P(\SLE(8/3) \textrm{ avoids }K) = |\Psi_K'(1)|^\lambda (\Psi_K'(0))^\mu,$$
where $\Psi_K$ is the conformal map $(\D\sm K,0)\to (\D,0)$ satisfying $\Psi_K'(0)>0.$
(The restriction exponents $\lambda$ and $\mu$ of radial $\SLE(8/3)$ are equal to $5/8$ and $5/48,$ respectively.)
\end{itemize}}

Let $\kappa\le4.$ 
On the event $\gamma[0,\infty)\cap K = \emptyset,$ a conformal map $h_t:\Omega_t = g_t(D_t\sm K)\to\D$ is defined by
$$h_t = \wt g_t \circ \Psi_K \circ g_t^{-1},$$
where $\wt g_t$ is a Loewner map from $\wt D_t = D\sm\wt\gamma[0,t]$ onto $\D,$ and 
$\wt\gamma(t) = \Psi_K \circ \gamma(t).$
Let $\leg^\eff$ be the effective one-leg operator, i.e., $\leg^\eff=\PP_q\leg$ and 
$$M_t := \frac{(Z_{\bfs\beta_t} \|\,\id_{\Omega_t})}{(Z_{\bfs\beta_t} \|\,\id_{\H})}, \quad
\bfs\beta_t = a\cdot\zeta_t + (b-\frac{a+\beta}2)\cdot q+ (b-\frac{a+\beta}2)\cdot q^*,\qquad(\zeta_t = \ee^{i\sqrt\kappa B_t}).$$
Then 
$$M_t=|h_t'(\zeta_t)|^\lambda h_t'(0)^\mu=\Big(\zeta_t\frac{h_t'(\zeta_t)}{h_t(\zeta_t)}\Big)^\lambda h_t'(0)^\mu,$$
where exponents are given by 
$$\lambda = \lambda_b(\leg^\eff) \equiv \frac{a^2}2-ab = \frac{6-\kappa}{2\kappa}, \qquad \mu = H_q(\leg^\eff)\equiv\frac{a^2}4-b^2 = \frac{(\kappa-2)(6-\kappa)}{8\kappa}.$$ 
Restriction property of radial $\SLE(8/3)$ follows from the local martingale property of $M_t$ (by optional stopping theorem). 
This is a special case of the following formula:
\begin{equation} \label{eq: restriction}
\textrm{the drift term of } \dd M_t= -\frac c6 \zeta_t^2S_{h_t}(\zeta_t) M_t\,\dd t.
\end{equation}
In Subsection~\ref{ss: Restriction formulas} we use the CFT argument to prove \eqref{eq: restriction} for radial $\SLE[\bfs\beta].$

We now prove Friedrich-Werner's formula in the radial case.
\begin{thm}\label{radial FW}
Let $\theta_j\in (0,2\pi)$ all distinct. Then we have 
$$ \lim_{t\to0}\frac{1}{t^n} 
\P(\SLE(8/3)\textrm{ hits all }[r_t\ee^{i\theta_j},\ee^{i\theta_j}])=(-2)^n\ee^{2i\sum_{j=1}^n\theta_j}\E\,[\,T_{\bfs\beta}(\ee^{i\theta_1} )\cdots T_{\bfs\beta}(\ee^{i\theta_n} )\,\|\,\id_\D\,],$$
where $\bfs\beta=a\cdot p + (b-a/2)\cdot q + (b-a/2)\cdot q^*,$ $a = \frac12\sqrt3,$ $b = -\frac16\sqrt3,$ $p=1,q=0,$ and $r_t = 1-2\sqrt t.$
\end{thm}
\begin{proof}
Let $z_j = \ee^{i\theta_j}.$ 
We apply Ward's equations to the function
$$\E\,[\,T_{\bfs\beta}(z) T_{\bfs\beta}(z_1)\cdots T_{\bfs\beta}(z_n)\,\|\,\id_\D\,] = \E\,[\,T_{\check{\bfs\beta\,}}(z)\leg^\eff(1) T_{\check{\bfs\beta\,}}(z_1)\cdots T_{\check{\bfs\beta\,}}(z_n)\,\|\,\id_\D\,],$$
by replacing $T_{\check{\bfs\beta\,}}(z)$ on the right-hand side with the corresponding Ward's functional.
Denote $\bfs{z} = (z_1,\cdots,z_n),$ $\bfs{z}_j = (z_1,\cdots,\wh z_j\cdots,z_n),$ and
$$R(\zeta;\bfs{z}) = \E\,[\,\leg^\eff(\zeta)\,T_{\check{\bfs\beta\,}}(z_1)\cdots T_{\check{\bfs\beta\,}}(z_n)\,].$$
The non-random field $R(\zeta;\bfs{z})$ is a boundary differential of conformal dimension $\lambda = \frac12{a^2}-ab=5/8$ with respect to $\zeta$ and of conformal dimension $2$ with respect to $z_j.$
It is also a differential of conformal dimension $\mu=\frac14a^2-b^2 = 5/48$ with respect to $q=0.$ 
It follows from Ward's equation for $\leg^\eff$ that
\begin{equation} \label{eq: recursion4T}
R(\zeta;z,\bfs{z})=\frac{1}{2z^2}\,\LL({v_z},\D)\,R(\zeta;\bfs{z}),\quad (\textrm{in }\id_\D),
\end{equation}
where $v_z(\eta) = \eta(z+\eta)/(z-\eta).$ 

Let 
$$U(\theta_1,\cdots,\theta_n) = \lim_{t\to0}\frac{(-1)^n}{(2t)^n} \ee^{-2i\sum_{j=1}^n\theta_j}
\P(\SLE(8/3)\textrm{ hits all }[r_t\ee^{i\theta_j},\ee^{i\theta_j}])$$
(if the limit exists).
Define the non-random field $T(\zeta;z_1,\cdots,z_n)$ as follows:
{\setlength{\leftmargini}{1.7em}
\begin{itemize}
 \item $T$ is a boundary differential of conformal dimension $\lambda=5/8$ with respect to $\zeta,$ and of conformal dimension $2$ with respect to $z_j;$ 
 \item $T$ is a differential of conformal dimension $\mu=\frac14a^2-b^2 = 5/48$ with respect to $q=0;$ 
 \item $(T(\ee^{i\varphi};\ee^{i\theta_1},\cdots,\ee^{i\theta_n})\,\|\,\id_\D) = U(\theta_1-\varphi,\cdots,\theta_n-\varphi).$
\end{itemize}}

We now claim that if the limit $U(\theta_1,\cdots,\theta_n)$ exists then the limit $U(\theta,\theta_1,\cdots,\theta_n)$ exists and 
\begin{equation} \label{eq: FW0}
T(1;z,\bfs{z}) = \frac1{2z^2} \,\LL({v_z},\D)\, T(1;\bfs{z})\qquad z,z_j\in\pa\D.
\end{equation}
By \eqref{eq: recursion4T} and \eqref{eq: FW0}, $T(1;\cdot)$ and $R(1;\cdot)$ satisfy the same recursive equation (see the remark at the end of this subsection) and are therefore equal since $T(1;\cdot) = R(1;\cdot) = 1$ for $n=0.$
Thus $$U(\theta_1,\cdots,\theta_n) = R(1;\ee^{i\theta_1},\cdots,\ee^{i\theta_n}).$$

To verify the induction argument for existence of the limit $U(\theta_1,\cdots,\theta_n)$ and show \eqref{eq: FW0}, denote $\bfs{\theta} = (\theta_1,\cdots,\theta_n).$
We write $\P(\bfs{\theta})$ for the probability that radial $\SLE(8/3)$ path hits all segments $[r_t\ee^{i\theta_j},\ee^{i\theta_j}]\,(1\le j\le n)$ and $\P(\bfs{\theta}\,|\,\neg \theta)$ for the same probability conditioned on the event that the path avoids $[r_t\ee^{i\theta},\ee^{i\theta}].$
By the induction hypothesis,
\begin{equation} \label{eq: FW1}
\P(\bfs \theta)\approx (-2t)^{n} z_1^2\cdots z_n^2 \, T(1;\bfs z),
\end{equation}
where $z_j = \ee^{i\theta_j}.$
On the other hand, by the restriction property of radial $\SLE(8/3),$ we have
\begin{equation} \label{eq: FW2}
1-\P(\theta) =|\psi_t'(1)|^\lambda \, \psi_t'(0)^{\frac16\lambda}
\end{equation}
and
\begin{equation} \label{eq: FW3}
\P(\bfs{\theta}\,|\,\neg \theta)\approx (-2t)^{n} \, T(\psi_t(1);\psi_t(z_1),\cdots, \psi_t(z_n))\prod_{j=1}^nz_j^2\psi_t'(z_j)^2 , \end{equation}
where $\psi_t$ is a slit map from $(\D\sm [r_t \ee^{i\theta},\ee^{i\theta}],0)$ onto $(\D,0)$ with $\psi_t'(0)>0$.
It follows from \eqref{eq: FW1}~--~\eqref{eq: FW3} and $\P(\theta,\bfs \theta) = \P(\bfs \theta)- \P(\bfs \theta\,|\,\neg \theta)(1-\P(\theta))$\,$(z = \ee^{i\theta})$ that 
$$\frac{\P(\theta,\bfs\theta)}{z^2z_1^2\cdots z_n^2(-2t)^{n+1}}=\frac{\psi_t'(0)^{\lambda/6} |\psi_t'(1)|^\lambda \prod_{j=1}^n\psi_t'(z_j)^2T(\psi_t(1);\psi_t(z_1),\cdots, \psi_t(z_n))-T(1;\bfs z)}{2tz^2} $$ 
up to $o(t)$ terms.
Thus the limit $U(\theta,\theta_1,\cdots,\theta_n)$ exists. 
Since $\psi_t'(0) = \ee^t,$ we have \eqref{eq: FW0}:
$$T(1;z,\bfs z) = \frac1{2z^2}\Big(\LL(v_z,\D\sm\{0\}) +\frac{\lambda}6\Big)T(1;\bfs z) = \frac1{2z^2} \,\LL({v_z},\D)\, T(1;\bfs z).$$
\end{proof}

\begin{rmk*} 
The formula \eqref{eq: recursion4T} holds for all $\kappa.$
Setting $R(z,\bfs z)\equiv R(1;z,\bfs z),$ 
the formula \eqref{eq: recursion4T} at $\zeta = 1$ gives the following recursive formula for $R:$
\begin{align*}
R(z,\bfs{z}) &= \frac1{2z^2}\Big(2n\frac{1+z}{1-z} + 2\lambda \frac{z}{(1-z)^2} + (\frac {a^2}4-b^2) + \frac{1+z}{1-z}\sum_{j=1}^n z_j\pa_{z_j}\Big) R(\bfs{z}) \\
&+\frac1{2z^2}\sum_{j=1}^n\Big(z_j\frac{z+z_j}{z-z_j} \pa_{z_j} + 2\frac{z^2+2z z_j - z_j^2}{(z-z_j)^2}\Big) R(\bfs{z}) + \frac c2 \sum_{j=1}^n \frac1{(z-z_j)^4} R(\bfs{z}_j).
\end{align*}
\end{rmk*}

\subsection{1-point vertex observables} \label{ss: 1pt O}
In this subsection we discuss some basic examples of 1-point vertex observables, including Lawler-Schramm-Werner's derivative exponents of radial SLEs on the boundary. 
Let
$$M^{(\tau^+,\tau^-;\tau_q^+,\tau_q^-)}(z)=\wh\E\,\OO_{\check{\bfs\beta\,}}[\bfs\tau] = \E\,\OO_{\bfs\beta}[\bfs\tau] , \qquad \bfs\tau = \tau^+\cdot z + \tau^-\cdot z^* + \tau_q^+\cdot q +\tau_q^-\cdot q^*.$$
Then by \eqref{eq: O hat} we have 
$$
M^{(\tau^+,\tau^-;\tau_q^+,\tau_q^-)} = (w_q')^{h_q^+}(\overline{w_q'})^{h_q^-} (w')^{h^+}(\overline{w'})^{h^-}w^{\nu^+}(\bar w)^{\nu^-}
(1-w)^{a\tau^+}(1-\bar w)^{a\tau^-}(1-|w|^2)^{\tau^+\tau^-},
$$
where the exponents are $\nu^\pm = \tau^\pm(\tau_q^\pm + b-a/2)$ and the dimensions are 
$$h^\pm = \lambda_b(\tau^\pm), \qquad h_q^\pm = \frac{\tau_q^\pm(\tau_q^\pm-a)}2.$$
The last formulas come from $h_q^\pm=\lambda_b(\tau_q^\pm+ b-a/2) - \lambda_b(\tau_q^\pm).$

\subsubsec{Constant fields}
The simplest examples of 1-point vertex fields are constant fields, i.e., vertex fields with $\tau^\pm= 0.$ 
By the neutrality condition, $\tau_q^- = -\tau_q^+.$ 
Since $w_t'(0) = \ee^{t-i\sqrt\kappa B_t},$
$$M_t^{(0,0;\tau_q,-\tau_q)} = \ee^{\tau_q^2t+ia\sqrt\kappa\tau_q B_t} = \ee^{\tau_q^2t+i\sqrt2\tau_q B_t}.$$ 
This is a martingale.

\subsubsec{Real fields} The 1-point vertex fields are real if and only if $\tau^+ = \tau^-$ and $\tau_q^+ = \tau_q^-.$
By the neutrality condition, $\tau_q^+ = -\tau^+.$
Thus the only real fields are
$$M^{(\tau,\tau;-\tau,-\tau)} = |w_q'|^{\tau^2+a\tau}|w'|^{\tau^2-2b\tau}|w|^{\tau(2b-a-2\tau)}|1-w|^{2a\tau}(1-|w|^2)^{\tau^2}.$$
When $\tau = -a,$ there is no covariance at $q.$ 
In this special case,
$$M^{(-a,-a;a,a)} = \bigg|\frac{w'}{w}\bigg|^{1-2/\kappa}\bigg(\frac{1-|w|^2}{|1-w|^2}\bigg)^{2/\kappa}.$$

\begin{eg*} If $\kappa =2,$ then $M^{(-a,-a;a,a)}$ coincides with the Lawler-Schramm-Werner observable 
$$M =\frac{1-|w|^2}{|1-w|^2} = \frac{P_\D(1,w)}{P_\D(1,0)} = \frac{P_D(p,z)}{P_D(p,q)},$$
where $P_D$ is the Poisson kernel of a domain $D.$
As mentioned in Subsection~\ref{ss: intro MO}, this 1-point field is an important observable in the theory of LERW.
\end{eg*}

\begin{eg*} If $\kappa=4,$ then 
$$M = \bigg|\frac{w'}{w}\bigg|^{1/2}\bigg(\frac{1-|w|^2}{|1-w|^2}\bigg)^{1/2}$$
is Beffara's type observable for radial $\SLE(4),$ see \cite{AKL12}.
In the chordal case, Beffara's observables are real martingale-observables of conformal dimensions
$$h^+=h^-=\frac12-\frac\kappa{16}, \qquad h_q=0$$
with the estimate 
$$\P(z,\ve)\asymp \ve^{1-\frac18\kappa}~M(z),$$
where
$\P(z,\ve)$ is the probability that the $\SLE(\kappa)$ curve ($\kappa<8$) hits the disc at $z$ of size $\ve(\ll1)$ measured in a local chart $\phi.$ 
See \cite{Beffara08}.
In \cite{AKL12}, Feynman-Kac formula is used to construct radial SLE martingale-observables with the desired dimensions. 
\end{eg*}

\subsubsec{1-point vertex fields without covariance at $q$} 
A 1-point vertex field $M$ has no covariance at $q$ if and only if 
$$(\tau_q^+,\tau_q^-) = (0,0),\quad (a,0), \quad (0,a), \quad \textrm{or} \quad (a,a).$$
The first case is just the non-chiral vertex field
$M = M^{(\tau,-\tau;0,0)}.$

When $\tau^-=0,$ the 1-point vertex observable 
$$M^{(\tau,0;\tau_q^+,\tau_q^-)} = (w'_q)^{h_q^+}(\overline{w'_q})^{ h_q^-} (w')^h w^{\nu} (1-w)^{a\tau}$$
is holomorphic.
Recall the expression for the exponents and the dimensions
$$h = \frac12\tau^2-b\tau, \quad
\nu = \tau(b-\frac12a+\tau_q^+), \quad
h_q^\pm =\frac{\tau_q^\pm(\tau_q^\pm-a)}2.$$

\subsubsec{Holomorphic 1-point fields without spin at $q$} 
A holomorphic 1-point field $M$ has no spin at $q$ if and only if 
$h_q^+ = h_q^-.$
Equivalently,
$$\tau_q^+ =\tau_q^- \quad\textrm{or}\quad \tau_q^+ +\tau_q^-=a.$$

\noindent \textit{Case 1.} $\tau_q^+ =\tau_q^-.$
By the neutrality condition, we have 
$$\bfs\tau=\tau\cdot z -\frac\tau2\cdot q-\frac\tau2\cdot q^*.$$
The holomorphic 1-point field $\OO_{\bfs\beta}[\bfs\tau]$ is a generalization of the one-leg operator $\leg_z.$
In this case, its conformal dimensions are 
$$h = \frac{\tau^2}2-b\tau, \quad h_q^\pm = h_q =\frac{\tau^2}{8} + \frac{a\tau}4.$$
Thus
$$M_t= \ee^{2h_qt}\Big(\frac{w_t'}{w_t}\Big)^{\frac12\tau^2-b\tau} \Big(\frac{w_t}{(1-w_t)^2}\Big)^{-\frac12{a\tau}}.$$

\begin{eg*}[Derivative exponents on the boundary \cite{LSW01c}]
On the unit circle, (up to constant)
$$M_t(\ee^{i\theta}) = \ee^{2 h_qt} |w_t'(\ee^{i\theta})|^h\Big(\sin^2\frac{\theta_t}{2}\Big)^{\frac12a\tau},$$
where $w_t(\ee^{i\theta}) = \ee^{i\theta_t}.$
Given $h,$ the equation $h = \frac12\tau^2-b\tau $ is solved by 
$$\tau_\pm = \frac a 4\big(\kappa-4 \pm \sqrt{(\kappa-4)^2+16\kappa h } \big).$$
With the choice of $\tau = \tau_+,$ Lawler, Schramm, and Werner proved that 
$M_t(\ee^{i\theta})$ is a martingale. 
They applied the optional stopping theorem to $M_t(\ee^{i\theta})$ and used the estimate 
$$\E[M_t(\ee^{i\theta})] \asymp \ee^{2 h_qt}\, \E[|w'_t(\ee^{i\theta})|^h \mathbf{1}_{\{\tau_{\ee^{i\theta}} > t\}}]$$
to derive the derivative exponents:
$$\E[|w'_t(\ee^{i\theta})|^h \mathbf{1}_{\{\tau_{\ee^{i\theta}} > t\}}] \asymp \ee^{-2 h_q t} \Big(\sin^2 \frac{\theta}2\Big)^{\frac12a\tau}.$$
(Recall that $\tau_z$ is the first time when a point $z$ is swallowed by the hull of SLE, see Subsection~\ref{ss: intro MO}.) 
From the derivative exponent for $\kappa=6,$ they obtained the annulus crossing exponent for $\SLE_6$ and combined it with other exponents to prove Mandelbrot's conjecture that the Hausdorff dimension of the planar Brownian frontier is $4/3$. 
See \cite{LSW01a} and references therein.
\end{eg*}

\begin{eg*}
The field $M$ is a scalar if $\tau = 2b.$
In this case,
$$M_t= \ee^{t(\kappa-4)/{8}}\Big(w_t+\frac1{w_t}-2\Big)^{(\kappa-4)/(2\kappa)}.$$
Its derivative $\widetilde M=\pa M$ has the conformal dimensions $[1,0;h_q,h_q].$ 
It is not a vertex observable. 
If $\widetilde M = M^{(\tilde \tau^+,\tilde \tau^-; \tilde\tau_q^+,\tilde\tau_q^-)}$ with $\tilde h^+ = \lambda_b(\tilde\tau^+)= 1,$ then
$\tilde\tau^+ = -2a$ or $\tilde\tau^+ = 2a+2b.$ 
As we will see below, the holomorphic 1-differentials without spin at $q$ are not forms of $\pa M^{(2b,0,-b,-b)}.$
Unlike the chordal case (see \cite[Proposition~15.2]{KM13}), the holomorphic differential observables are not necessarily vertex observables. 
\end{eg*}

\begin{eg*}
If we take $\tau = 2b-a$ so that $h = h_{1,2}:=\frac12a^2-ab $ (there are only two values of $\tau$ such that $h = h_{1,2},$ namely $\tau=a,2b-a$), then 
$$\bfs\tau = (2b-a)\cdot z + (-b+\frac a2)\cdot q + (-b+\frac a2)\cdot q^*,$$ 
and 
$$h = \frac{6-\kappa}{2\kappa} = -\frac{a\tau}2,\quad h_q = \frac{(2-\kappa)(6-\kappa)}{16\kappa}=-h_{0,1/2}, \quad \nu =0.$$ 
In this case, we have 
$$M_t = \ee^{2 h_qt} \Big(\frac{w_t'}{(1-w_t)^2}\Big)^{(6-\kappa)/(2\kappa)}.$$
If $\kappa=2,$ then
$$M = \frac{w'}{(1-w)^2}$$
and its anti-derivative is $1/(1-w).$ 
See the next example.
\end{eg*}

\ss\noindent \textit{Case 2.} $\tau_q^+ +\tau_q^-=a.$ Let $\tau_q = \tau_q^+.$ 
It follows from the neutrality condition that 
$$\bfs\tau = -a\cdot z + \tau_q\cdot q + (a-\tau_q)\cdot q^*.$$
Thus we have
$$M = (w')^{(\kappa-2)/(2\kappa)}w^{a(-\tau_q+\frac12a-b)}(1-w)^{-2/\kappa}|w_q'|^{\tau_q^2-a\tau_q}.$$

\ss\begin{eg*} The fields $M$ have no covariance at $q$ if and only if $\tau_q = 0$ or $a.$ 
In these cases, we have
$$M^{(-a,0;a,0)} = \Big(\frac{w'}{w}\Big)^{(\kappa-2)/(2\kappa)}(1-w)^{-2/\kappa}$$
and
$$M^{(-a,0;0,a)} = (w')^{(\kappa-2)/(2\kappa)}w^{(6-\kappa)/(2\kappa)}(1-w)^{-2/\kappa}.$$
For example, if $\kappa =6,$ then 
$$M^{(-a,0;0,a)} = \Big(\dfrac{w'}{1-w}\Big)^{1/3},$$
and if $\kappa = 2,$ then both $M^{(-a,0;a,0)}$ and $M^{(-a,0;0,a)}$ produce the same observable 
$$M^{(-a,0;a,0)}=1+M^{(-a,0;0,a)}=\frac1{1-w}.$$
\end{eg*}

\subsubsec{Holomorphic differentials without spin at $q$} 
Let $\tau = \tau^+,$ $\tau^- = 0.$
A holomorphic 1-point field $M$ is a 1-differential with respect to $z$ if and only if 
$\tau = -2a$ or $\tau = 2(a+b).$ 
Furthermore, if $M$ has no spin at $q,$ then $\tau_q^+ = \tau_q^-.$
(If $\tau = -2a$ or $\tau = 2(a+b),$ then the other possibility $\tau_q^++\tau_q^-=a$ never happens because of the neutrality condition.)

\ms\noindent \textit{Case 1.} $\tau = -2a.$ 
By the neutrality condition, we have 
$$\bfs\tau= -2a\cdot z + a\cdot q + a\cdot q^*$$
and
$$M = w'w^{2/\kappa-1}(1-w)^{-4/\kappa}.$$

\begin{eg*}
If $\kappa=2,$ then
$$M^{(-2a,0;a,a)} = \frac{w'}{(1-w)^2}$$ 
and its anti-derivative is 
$1/(1-w).$ 
See the previous example.
\end{eg*}

\begin{eg*}
If $\kappa=4,$ then
$M^{(-2a,0;a,a)} = w'(1-w)^{-1}w^{-1/2}$ and its anti-derivative is 
$$\log(1+\sqrt w)-\log(1-\sqrt w).$$ 
Its imaginary part is a bosonic observable for a twisted conformal field theory. 
\end{eg*}

\noindent \textit{Case 2.} $\tau = 2(a+b).$ 
By the neutrality condition, we have 
$$\bfs\tau= (2a+2b)\cdot z -(a+b)\cdot q -(a+b)\cdot q^*$$ 
and
$$M_t = \ee^{t(\kappa+4)/8}w_t' w_t^{-3/2}(1-w_t).$$

\subsection{Restriction formulas} \label{ss: Restriction formulas}
Let $\kappa \le 4.$
We first consider the chordal $\SLE[\bfs\beta]$ with 
$$\bfs\beta = a\cdot p + \sum \beta_k \cdot q_k + (2b-a-\sum \beta_k) \cdot q, \qquad p,q, q_k \in \pa D.$$
Fix a hull $K.$ 
Let $g_t$ be the chordal $\SLE[\bfs\beta]$ map with the hull $K_t.$ 
On the event $K_\infty \cap K = \emptyset,$ a conformal map $h_t:\Omega_t = g_t(D_t\sm K)\to\H$ is defined by
$$h_t = \widetilde g_t \circ \Psi_K \circ g_t^{-1},$$
where $\widetilde g_t$ is a Loewner map from $\widetilde D_t = D\sm\widetilde K_t$ onto $\H,$
$\widetilde K_t = \Psi_K(K_t),$ and
$\Psi_K$ is the conformal map $(\H\sm K,0,\infty)\to (\H,0,\infty)$ satisfying $\Psi_K'(\infty)=1.$
Let 
$$M_t := \frac{(Z_{\bfs\beta_t} \|\,\id_{\Omega_t})}{(Z_{\bfs\beta_t} \|\,\id_{\H})} = \frac{\E\,(\leg^\eff_{\Omega_t}\|\,\id)(\xi_t,\bfs q(t))}{\E\,(\leg^\eff_\H\|\,\id)(\xi_t,\bfs q(t))},$$
where 
$$\bfs\beta_t = a\cdot \xi_t + \sum \beta_k \cdot q_k(t), \qquad \dd q_k(t) = \frac2{q_k(t)-\xi_t}\,\dd t ,\qquad q_k(0) = q_k$$
and
$\leg^\eff = \PP_{\bfs\beta}\,\leg$ is the effective one-leg operator. 
Then the process $M_t$ is expressed in terms of $h_t,$ $\xi_t,$ and $q_k(t)$ as 
\begin{equation} \label{eq: chordal M_t}
M_t = h_t'(\xi_t)^\lambda \prod_j h_t'(q_j(t))^{\lambda_j} \Big(\frac{h_t(\xi_t)-h_t(q_j(t))}{\xi_t-q_j(t)}\Big)^{a\beta_j} \prod_{j<k} \Big(\frac{h_t(q_j(t))-h_t(q_k(t))}{q_j(t)-q_k(t)}\Big)^{\beta_j\beta_k}.
\end{equation}
We now present the conformal field theoretic proof for the restriction formula (\cite{Dubedat05}) of chordal $\SLE[\bfs\beta]:$ 
\begin{equation} \label{eq: restriction formula for chordal SLE(kappa,rho)}
\dd M_t = \frac c6 S_{h_t}(\xi_t)M_t\,\dd t + \textrm{martingale terms.}
\end{equation}

Let 
$$F(z,\bfs q,t) := \frac{G(z,\bfs q,t)}{H(z,\bfs q)}, \quad 
\begin{cases}
G(z,\bfs q,t) := &\E\,(\leg^\eff_{\Omega_t}\|\id)(z,\bfs q)(=\E\,(\leg^\eff_\H\|h_t^{-1})(z,\bfs q)), \\
H(z,\bfs q) := &\E\,(\leg^\eff_\H\|\id)(z,\bfs q).
\end{cases}
$$
Recall that the driving process $\xi_t$ satisfies 
$$\dd \xi_t = \sqrt\kappa B_t + \kappa \frac{H_\xi(\xi_t,\bfs q(t))}{H(\xi_t,\bfs q(t))}\,\dd t.$$
It follows from It\^o's formula that the drift term of $\dd M_t/M_t$ equals 
$$
\Big(\frac{\dot F}{F} + \frac{\kappa}2 \frac{F_{\xi\xi}}{F} + \sum_j \frac{F_{q_j}}{F} \dot q_j(t) + \kappa \frac{F_\xi}F \frac{H_\xi}H\Big)\,\dd t$$
evaluated at $(\xi_t,\bfs q(t),t).$
We rewrite the drift term of $\dd M_t/M_t$ in terms of $G$ and $H:$
$$\textrm{the drift term of } \frac{\dd M_t}{M_t} =\Big(\frac{\dot G}{G} + \frac{\kappa}2 \Big(\frac{G_{\xi\xi}}{G}-\frac{H_{\xi\xi}}{H}\Big)+ \sum_j \Big(\frac{G_{q_j}}{G}-\frac{H_{q_j}}{H}\Big)\frac2{q_j(t)-\xi_t}\Big)\,\dd t.$$
Here we use 
$$\frac{\dot F}{F} = \frac{\dot G}{G} , \qquad \frac{F_{q_j}}{F} = \frac{G_{q_j}}{G}-\frac{H_{q_j}}{H}, \qquad \frac{F_{\xi\xi}}{F} = \frac{G_{\xi\xi}}{G} -\frac{H_{\xi\xi}}{H}-2 \frac{F_\xi}{F}\frac{H_\xi}{H}.$$

By Ward's equations we have 
\begin{align*}
\dot G(z,\bfs q(t),t) = &-2 h_t'(\xi_t)^2h_t'(z)^\lambda\prod_j h_t'(q_j(t))^{\lambda_j}\E \big(A_\H(h_t(\xi_t)) \leg^\eff_\H(h_t(z), h_t(q_1(t)),\cdots)\|\id\big)\\
&+2\E\,(\LL_{v_{\xi_t}}\leg^\eff_{\Omega_t}\|\id)(z,\bfs q(t)).
\end{align*}
It follows from conformal invariance that
$$\dot G(z,\bfs q(t),t) = -2 (\E\,A_{\Omega_t}(\xi_t)\leg^\eff_{\Omega_t}(z,\bfs q(t))\|\id) + 2\E\,(\LL_{k_{\xi_t}}\leg^\eff_{\Omega_t}\|\id)(z,\bfs q(t)).$$
Sending $z$ to $\xi_t,$
$$\dot G(\xi_t,\bfs q(t),t) = -2 (\E\,A_{\Omega_t}*_\xi\leg^\eff_{\Omega_t}(\xi_t,\bfs q(t))\|\id) + 2\E\,(\check{\LL}_{k_{\xi_t}}\leg^\eff_{\Omega_t}\|\id)(\xi_t,\bfs q(t)),$$
where $\check{\LL}_{k_{\xi_t}}$ does not apply to $\xi_t.$
By the level two degeneracy equation for $\leg^\eff$ with respect to $\xi$, 
$$\frac{\kappa}2 \Big(\frac{G_{\xi\xi}}{G}-\frac{H_{\xi\xi}}{H}\Big) = 
2\frac{(\E\,T_{\Omega_t}*_\xi\leg^\eff_{\Omega_t}(\xi_t,\bfs q(t))\|\id)}{(\E\,\leg^\eff_{\Omega_t}(\xi_t,\bfs q(t))\|\id)} 
-2\frac{(\E\,T_{\H}*_\xi\leg^\eff_{\H}(\xi_t,\bfs q(t))\|\id)}{(\E\,\leg^\eff_{\H}(\xi_t,\bfs q(t))\|\id)}.$$
It follows from Ward's equation that 
$$(\E\,T_{\H}*_\xi\leg^\eff_{\H}(\xi_t,\bfs q(t))\|\id) = \E\,(\check{\LL}_{k_{\xi_t}}\leg^\eff_{\H}\|\id)(\xi_t,\bfs q(t)).$$ 
The formula for Lie derivatives of differentials gives
$$
 \frac{\E\,(\check{\LL}_{k_{\xi_t}}\leg^\eff_{\Omega_t}\|\id)(\xi_t,\bfs q(t))}{(\E\,\leg^\eff_{\Omega_t}(\xi_t,\bfs q(t))\|\id)}
- \frac{\E\,(\check{\LL}_{k_{\xi_t}}\leg^\eff_{\H}\|\id)(\xi_t,\bfs q(t))}{(\E\,\leg^\eff_{\H}(\xi_t,\bfs q(t))\|\id)} = -\sum_j \Big(\frac{G_{q_j}}{G} - \frac{H_{q_j}}{H} \Big)\frac1{q_j(t)-\xi_t}. 
$$
Combining all of the above, we have 
\begin{align*}
\textrm{the drift term of } \frac{\dd M_t}{M_t} &= -2 \frac{(\E\,A_{\Omega_t}*_\xi\leg^\eff_{\Omega_t}(\xi_t,\bfs q(t))\|\id)}{(\E\,\leg^\eff_{\Omega_t}(\xi_t,\bfs q(t))\|\id)}\,\dd t +2 \frac{(\E\,T_{\Omega_t}*_\xi\leg^\eff_{\Omega_t}(\xi_t,\bfs q(t))\|\id)}{(\E\,\leg^\eff_{\Omega_t}(\xi_t,\bfs q(t))\|\id)}\,\dd t \\&= \frac c6 S_{h_t}(\xi_t)\,\dd t.
\end{align*}

We now present a conformal field theoretic proof for the restriction formula of radial $\SLE[\bfs\beta]$ $(\kappa=8/3).$
For this purpose, we let $\zeta, q_k\in \pa D, q \in D,$ 
$$\bfs\beta = \bfs\tau + \check{\bfs\beta\,}, \quad \bfs\tau = a\cdot\zeta - \frac a2\cdot q - \frac a2\cdot q^*\quad \check{\bfs\beta\,} = \sum \beta_k\cdot q_k + (b-\frac\beta2)\cdot q+ (b-\frac\beta2)\cdot q^*, \quad\beta=\sum\beta_k,$$ 
and consider the effective one-leg operator $\leg^\eff,$ 
$$\leg^\eff(\zeta):=\PP_{\check{\bfs\beta\,}}\,\leg(\zeta) = C_{(b)}[\bfs\beta]\,\ee^{\odot i \Phi[\bfs\tau]}.$$

Fix a hull $K$ such that $K$ does not intersect $\zeta = 1, q_1, \cdots, q_k, \cdots.$
Let $g_t$ be the radial $\SLE[\bfs\beta]$ map with the hull $K_t.$ 
On the event $K_\infty\cap K = \emptyset,$ a conformal map $h_t:\Omega_t = g_t(D_t\sm K)\to\D$ is defined by
$$h_t = \widetilde g_t \circ \Psi_K \circ g_t^{-1},$$
where $\widetilde g_t$ is a radial Loewner map from $\widetilde D_t = D\sm\widetilde K_t$ onto $\D,$
$\widetilde K_t = \Psi_K(K_t),$ and $\Psi_K$ is the conformal map $(\D\sm K,0)\to (\D,0)$ satisfying $\Psi_K'(0)>0.$
Let 
\begin{equation} \label{eq: radial M_t}
M_t := \frac{(Z_{\bfs\beta_t} \|\,\id_{\Omega_t})}{(Z_{\bfs\beta_t} \|\,\id_{\H})}, \quad
\bfs\beta_t = a\cdot\zeta_t + (b-\frac{a+\beta}2)\cdot q+ (b-\frac{a+\beta}2)\cdot q^*+ \sum \beta_k\cdot q_k(t),
\end{equation}
where $q_k(t)$ satisfies $q_k(0) = q_k$ and 
$$
\dd q_k(t) = q_k(t)\frac{\zeta_t + q_k(t)}{\zeta_t - q_k(t)}\,\dd t.
$$

We now represent $M_t$ in terms of the effective one-leg operator:
$$M_t = F(\zeta_t, \bfs q(t),t), \quad F(z,\bfs q,t)= \frac{G(z,\bfs q,t)}{H(z,\bfs q)},$$
where $G(z,\bfs q,t) = \E(\leg_\D^\eff\|h_t^{-1})(z),$ and $H(z,\bfs q) = \E(\leg_\D^\eff\|\id)(z).$

\begin{thm}
We have
\begin{equation} \label{eq: radial restriction formula}
\textrm{the drift term of } \frac{\dd M_t}{M_t} = -\frac c6\, \zeta_t^2\,S_{h_t}(\zeta_t)\,\dd t.
\end{equation}
\end{thm}
\begin{proof}
Recall that
$$\dd\zeta_t = i\sqrt\kappa\,\zeta_t\,\dd B_t - \frac\kappa2\zeta_t\,\dd t - \kappa\zeta_t^2 \Big(\frac{H_\zeta}H + \frac{h}{\zeta_t}\Big)\,\dd t = i\sqrt\kappa\,\zeta_t\,\dd B_t - \kappa\zeta_t^2 \frac{H_\zeta}H\,\dd t-3\zeta_t\,\dd t .$$
It follows from It\^o's formula that 
$$\dd M_t = \dot F\,\dd t - \frac\kappa2\zeta_t^2\,F_{\zeta\zeta}\,\dd t + F_\zeta\,\dd\zeta_t + \sum F_{q_k}\,\dd q_k(t)
$$
evaluated at $(\zeta_t, \bfs q(t),t).$
The drift term of $\dd M_t/M_t$ is 
$$\Big(\frac{\dot F}{F} - \frac\kappa2\zeta_t^2\,\frac{F_{\zeta\zeta}}F -\Big(\kappa\zeta_t^2 \frac{H_\zeta}H+3\zeta_t\Big) \frac{F_{\zeta}}F+ \sum_k\Big(q_k(t)\frac{\zeta_t + q_k(t)}{\zeta_t-q_k(t)}\Big) \frac{F_{q_k}}F\Big)\Big|_{(\zeta_t, \bfs q(t),t)}.$$
We rewrite the drift term of $\dd M_t/M_t$ in terms of $G$ and $H:$ 
\begin{equation} \label{eq: dM/Mradial}
\Big(\frac{\dot G}{G} - \frac\kappa2\zeta_t^2\,\Big(\frac{G_{\zeta\zeta}}G-\frac{H_{\zeta\zeta}}H \Big) -3\zeta_t \Big(\frac{G_{\zeta}}G-\frac{H_{\zeta}}H \Big)
+ \sum_k\Big(q_k(t)\frac{\zeta_t + q_k(t)}{\zeta_t-q_k(t)}\Big)\Big( \frac{G_{q_k}}G - \frac{H_{q_k}}H \Big)\Big)
\end{equation}
evaluated at $(\xi_t,\bfs q(t),t).$

Using the similar method in \cite[Section~14.5]{KM13}, we represent $\dot G$ in terms of the Lie derivatives:
\begin{align} \label{eq: dotF1}
\dot G(z,\bfs q, t) &= \frac{\dd}{\dd s}\Big|_{s=0} (\E\, \leg^\eff_{\D}\,\|\, h_{t+s}^{-1}) (z)= \frac{\dd}{\dd s}\Big|_{s=0} (\E\, \leg^\eff_{\D}\,\|\, h_{t}^{-1}\circ f_{s,t}^{-1}) (z) \\
&= (\E\,\LL(v,\D) \,\leg^\eff_{\D}\,\|\,h_t^{-1})(z), \nonumber
\end{align}
where $f_{s,t} = h_{t+s}\circ h_{t}^{-1}$ and
$$(v\,\|\,\id_\D) = \frac{\dd}{\dd s}\Big|_{s=0} f_{s,t} = \dot h_t \circ h_t^{-1}.$$
We only need to compute the vector field $v.$
We represent $v$ as the difference of two Loewner vector fields associated with the flows in the domains $\D$ and $\Omega_t.$
Applying the chain rule to $h_t = \wt g_t \circ \Psi_K \circ g_t^{-1}$ and computing the capacity changes, we have 
$$\dot h_t(z) = |h_t'(\zeta_t)|^2v_{\wt \zeta_t}(h_t(z)) - h_t'(z)v_{\zeta_t}(z), \quad \Big(v_{\zeta}(z)=z\frac{\zeta+z}{\zeta-z}\Big),$$
where $\wt \zeta_t = h_t(\zeta_t).$
By the above equation and $(v\,\|\,\id_\D) = \dot h_t \circ h_t^{-1},$ 
\begin{equation} \label{eq: v}
(v\,\|\,\id_\D)(z)=|h_t'(\zeta_t)|^2v_{\wt\zeta_t}(z) - h_t'(h_t^{-1}(z))v_{\zeta_t}(h_t^{-1}(z)).
\end{equation}

It follows from \eqref{eq: dotF1} and \eqref{eq: v} that 
\begin{align*}
\dot G(z,\bfs q(t),t) &= |h_t'(\zeta_t)|^2 h_t'(z)^\lambda\prod \big(h_t'(q_k(t))\big)^{\lambda_k} \,\E\, \big(\LL(v_{\wt\zeta_t},\D) \,\leg^\eff_{\D} (h_t(z)\big) \,\|\, \id_\D) \\
&- \big(\E\,\LL(v_{\zeta_t}) \,\leg^\eff_{\Omega_t}\,\|\,\id_{\Omega_t}\big)(z).
\end{align*}
By Ward's equation, 
\begin{align*} 
\dot G(z,\bfs q(t),t) &= 2|h_t'(\zeta_t)|^2 \prod \big(h_t'(q_k(t))\big)^{\lambda_k}h_t'(z)^\lambda\wt\zeta_t^2(\E\, T_\D(\wt \zeta_t)\,\leg^\eff_\D (h_t(z))\,\|\,\id_\D)\\
& - (\E\,\LL(v_{\zeta_t})\,\leg^\eff_{\Omega_t}\,\|\,\id_{\Omega_t}\big)(z).
\end{align*}
It follows from conformal invariance that
\begin{align*}
\dot G(z,\bfs q(t),t) &=2\zeta_t^2(\E\, T_{\Omega_t}(\zeta_t)\,\leg^\eff_{\Omega_t}(z)\,\|\,\id_{\Omega_t}) -\frac c{6} \zeta_t^2 S_{h_t}(\zeta_t)\, (\E\, \leg^\eff_{\Omega_t}\,\|\,\id_{\Omega_t}\big)(z)\\
&-(\E\,\LL(v_{\zeta_t})\,\leg^\eff_{\Omega_t}\,\|\,\id_{\Omega_t}\big)(z).
\end{align*}
Sending $z$ to $\zeta_t$ and applying \eqref{eq: sing OPE4TY} (and $T*_{-1}\leg^\eff = \pa \leg^\eff, T*_{-2}\leg^\eff = h \leg^\eff$, see \eqref{eq: T primary}), 
\begin{equation*}
\frac{\dot G(\zeta_t,\bfs q(t),t)}{G(\zeta_t,\bfs q(t),t)} = -\frac c6 \zeta_t^2 S_{h_t}(\zeta_t)+\frac{L(\zeta_t)G(\zeta_t,\bfs q(t),t)}{G(\zeta_t,\bfs q(t),t)} - \frac{\check{\LL}_{v_{\zeta_t}}G(\zeta_t,\bfs q(t),t)}{G(\zeta_t,\bfs q(t),t)},
\end{equation*}
where $L(z) = 2z^2L_{-2}(z) + 3zL_{-1}(z) + L_0(z).$
By \eqref{eq: dM/Mradial} and the level two degeneracy equation, the drift of $\dd M_t/M_t$ simplifies to
\begin{align*}
-\frac c6 \zeta_t^2 S_{h_t}(\zeta_t) & - \frac{\check{\LL}_{v_{\zeta_t}}G(\zeta_t,\bfs q(t),t)}{G(\zeta_t,\bfs q(t),t)} +\Big(\frac\kappa2\zeta_t^2\,\frac{H_{\zeta\zeta}}H + 3\zeta_t\frac{H'}H + h\Big)\Big|_{(\zeta_t, \bfs q(t),t)}\\
&+ \sum_k\Big(q_k(t)\frac{\zeta_t + q_k(t)}{\zeta_t-q_k(t)}\Big)\Big( \frac{G_{q_k}}G - \frac{H_{q_k}}H \Big)\Big|_{(\zeta_t, \bfs q(t),t)}.
\end{align*}
It follows from the null vector equation that 
$$\frac\kappa2\zeta_t^2\,\frac{H_{\zeta\zeta}}H + 3\zeta_t\frac{H'}H + h = \frac{\check{\LL}_{v_{\zeta_t}}H}H.$$
Thus we find the drift term of $\dd M_t/M_t$ as 
$$-\zeta_t^2 \frac c6 S_{h_t}(\zeta_t)- \Big(\frac{\check{\LL}_{v_{\zeta_t}}G}G - \frac{\check{\LL}_{v_{\zeta_t}}H}H\Big)+ \sum_k\Big(q_k(t)\frac{\zeta_t + q_k(t)}{\zeta_t-q_k(t)}\Big)\Big( \frac{G_{q_k}}G - \frac{H_{q_k}}H \Big)$$
evaluated at $(\xi_t,\bfs q(t),t).$
Since 
$$\frac{\check{\LL}_{v_{\zeta_t}}G}G - \frac{\check{\LL}_{v_{\zeta_t}}H}H= \sum_k\Big(q_k(t)\frac{\zeta_t + q_k(t)}{\zeta_t-q_k(t)}\Big)\Big( \frac{G_{q_k}}G - \frac{H_{q_k}}H \Big),$$
we have \eqref{eq: radial restriction formula}.
\end{proof}

\section{Conformal field theory with Neumann boundary condition} \label{sec: Neumann CFT}

In this section we briefly implement a version of conformal field theory from background charge modifications of the Gaussian field with Neumann boundary condition. 
In the last subsection we present the connection of this theory to the theory of the backward SLEs. 

\subsection{Gaussian free field with Neumann boundary condition} \label{ss: GFF N}
The Gaussian free field $N$ in a planar domain $D$ with Neumann boundary condition is an isometry
$N:\EE_N(D) \to L^2(\Omega,\P)$
from the Neumann energy space $\EE_N(D)$ such that the image consists of centered Gaussian random variables.
Here $(\Omega,\P)$ is a probability space and $\EE_N(D)$ is the completion of smooth functions up to the boundary with mean zero (or the neutrality condition $(\NC_0)$), compact supports in $D\,\cup\,\pa D,$ and Neumann boundary condition $\nabla f \cdot n = 0$ (where $n$ is normal to $\pa D$)
with respect to the norm
$$\|f\|^2_{\EE_N}=\iint 2G_N(\zeta,z)\,f(\zeta)\,\overline{f(z)}~\dd A(\zeta)\,\dd A(z),$$
where $A$ is the normalized area measure and $G_N$ is the Neumann Green's function for $D.$
In the upper half-plane, we have 
$$G_N(\zeta,z) = \log\frac1{|(\zeta-z)(\zeta-\bar z)|}.$$
As in the Dirichlet case, $N$ can be constructed from the Gaussian free field $\Psi$ on its Schottky double $S = D^{\mathrm{double}}.$ 
The Neumann energy space $\EE_N(D)$ can be embedded isometrically into $\EE(S)$ in a natural way. 
For example, for $\mu$ in the energy space $\EE_N(\H),$ 
\begin{align*}
\|\mu\|_{\EE_N(\H)}^2 &= \int_{\H\times\H} \log\frac1{|z-w|^2} \, \mu(z) \overline{\mu(w)} + \int_{\H\times\H} \log\frac1{|z-\bar w|^2} \, \mu(z) \overline{\mu(w)}\\
&= \frac12 \int_{\C\times\C} \log\frac1{|z-w|^2} \, \nu(z) \overline{\nu(w)} = \frac12 \|\nu\|_{\EE(\wh\C)}^2,
\end{align*}
where $\nu = \mu$ in the upper half-plane $\H$ and $\nu = \mu^*$ in the lower half-plane $\H^*$ 
($\mu^*(E) = \mu(E^*),$ $E^* = \{\bar z, z\in E\}$). 
For example, a test function $f$ in $\H$ with mean zero and Neumann boundary condition extends smoothly to $\C$ such that $f(\bar z) = f(z).$

As a Fock space space, the Gaussian free field $N \equiv N_D$ in $D$ can be constructed from the Gaussian free field $\Psi \equiv \Psi_{S}:$
$$N(z) = \frac1{\sqrt 2} (\Psi(z) + \Psi(z^*)).$$
This field $N(z)$ is formal and thus we need to require the neutrality condition $(\NC_0)$ on $N.$ For example, $N(z,z_0):=N(z)-N(z_0)$ is a 2-variant well-defined Fock space field.
The formal field $N$ has the formal correlation:
$$\E\, N(z)N(z_0) = 2 G_N(z,z_0).$$

From the Schottky double construction, the current field $J_D^D = \pa \Phi_D$ in $D$ with Dirichlet boundary condition and the current field $J_S = \pa\Psi_S$ are related as follows. For example, 
$$J_\H^D(z) = \frac1{\sqrt2} \big(J_{\wh\C}(z) - \overline{J_{\wh\C}(\bar z)}\big).$$ 
On the other hand, the current field $J_\H^N$ with Neumann boundary condition is related to $J_\C$ as 
$$J_\H^N(z) = \frac1{\sqrt2} \big(J_{\wh\C}(z) + \overline{J_{\wh\C}(\bar z)}\big).$$ 
In a similar way, stress tensors are related as 
$$A_\H^D(z) = -\frac12 J_\H^D \odot J_\H^D(z) = \frac12\big(A_{\wh\C}(z)+ \overline{A_{\wh\C}(\bar z)} + J_{\wh\C}(z) \odot \overline{J_{\wh\C}(\bar z)} \big)$$ 
and
$$A_\H^N(z) = -\frac12 J_\H^N \odot J_\H^N(z)= \frac12\big(A_{\wh\C}(z)+ \overline{A_{\wh\C}(\bar z)} - J_{\wh\C}(z) \odot \overline{J_{\wh\C}(\bar z)} \big).$$ 
It is easy to check that $A^N$ is a stress tensor for $N.$ 
Indeed, as $\zeta\to z,$ we have the following Ward's OPE for $N:$
$$A^N(\zeta)N(z) = -\frac12 J^N(\zeta)\odot J^N(\zeta)N(z) = -\E[J^N(\zeta)N(z)]J^N(\zeta) \sim \frac{J^N(\zeta)}{\zeta-z}\sim \frac{\pa N(z)}{\zeta-z}$$
in the identity chart of $\H.$ 
Here we use
$$\E\,J^N(\zeta)N(z) = -\frac1{\zeta-z} -\frac1{\zeta-\bar z}$$
in $\H.$ 

Formal fields $N^\pm(z)$ with Neumann boundary condition are defined by 
$$N^\pm(z) = \frac1{\sqrt 2} (\Psi^\pm(z) + \Psi^\mp(z^*)).$$
For example, we have 
\begin{equation} \label{eq: NN+ NN-}
\E\, N(z)N^+(z_0) = \log\frac1{z-z_0} + \log\frac1{\bar z-z_0}, \qquad 
\E\, N(z)N^-(z_0) = \log\frac1{z-\bar z_0} + \log\frac1{\bar z-\bar z_0},
\end{equation}
and
$$\E\, N^+(z)N^+(z_0) = \log\frac1{z-z_0}, \qquad \E\, N^+(z)N^-(z_0) = \log\frac1{z-\bar z_0}$$
in $\H.$
As in the Dirichlet case, 2-variant fields $N^\pm(z,z_0):=N^\pm(z)-N^\pm(z_0)$ are well-defined multivalued Fock space fields. 

\subsection{Modifications and OPE exponentials} 
We first define background charge modifications of the chiral bosonic fields with Neumann boundary condition.
Given double background charges 
$\bfs\beta^\pm = \sum_k \beta_k^\pm\cdot q_k$ with the neutrality conditions $(\NC_{b^\pm})$
in $D,$ we define the formal fields $N_{\bfs\beta}^\pm \equiv N_{\bfs\beta^+,\bfs\beta^-}^\pm $ ($\bfs\beta = \bfs\beta^++\bfs\beta_*^-$) by 
$$N_{\bfs\beta}^\pm := N^\pm + n_{\bfs\beta}^\pm,$$
where $n_{\bfs\beta}^+\equiv n_{\bfs\beta^+,\bfs\beta^-}^+$ is a $\PPS(b,0)$ form ($b = b^++b^-$) and $n_{\bfs\beta}^-\equiv n_{\bfs\beta^+,\bfs\beta^-}^-$ is a $\PPS(0,b)$ form such that 
\begin{align} \label{eq: n_beta}
n_{\bfs\beta}^+(z) &= \sum \beta_k^+ \log\frac1{z-q_k} +\sum \beta_k^- \log\frac1{ z-\bar q_k}\\
n_{\bfs\beta}^-(z) &= \sum \beta_k^+ \log\frac1{\bar z-q_k} +\sum \beta_k^- \log\frac1{\bar z-\bar q_k} \nonumber
\end{align}
in the $\H$-uniformization. 
Furthermore, $n_{\bfs\beta}:=n_{\bfs\beta}^++n_{\bfs\beta}^-$ is a $\PPS(b,b)$ form. 
For two divisors $\bfs\tau^\pm$ satisfying the neutrality condition $(\NC_0),$ we define 
$$N_{\bfs\beta}[\bfs\tau]\equiv N_{\bfs\beta}[\bfs\tau^+,\bfs\tau^-]:= \sum_j \big(\tau_j^+N_{\bfs\beta}^+(z_j)+ \tau_j^-N_{\bfs\beta}^-(z_j) \big),$$
where $\bfs\tau = \bfs\tau^+ + \bfs\tau_*^-.$ 
Compare this definition to the definition of $\Phi_{\bfs\beta}[\bfs\tau]$ in the Dirichlet case: 
$$\Phi_{\bfs\beta}[\bfs\tau]:= \sum_j \big(\tau_j^+\Phi_{\bfs\beta}^+(z_j)- \tau_j^-\Phi_{\bfs\beta}^-(z_j) \big),$$
where $\bfs\tau = \bfs\tau^+ + \bfs\tau_*^-$ satisfies the neutrality condition $(\NC_0).$ 
We now explain why the signs of $\tau_j^-$ appear differently. 
The choices of signs are consistent with the representation of $\bfs\tau = \bfs\tau^+ + \bfs\tau_*^-$ when the nodes are on the boundary. 
For $z\in\pa D,$ we have $\Phi(z)=0$ and thus the formal field $\Phi[1\cdot z]$ can be represented either by $\Phi^+(z)$ or by $-\Phi^-(z).$ 
On the other hand, the formal field $N[1\cdot z]$ is represented by $N^+(z) = N^-(z).$ 
The relation $N^+(z) = N^-(z)$ for $z\in\pa D$ is obvious: 
$$\E\,N^+(z)N(z_0) = \log \frac1{(z-z_0)(z-\bar z_0)} = \log \frac1{(\bar z-z_0)(\bar z-\bar z_0)} = \E\,N^-(z)N(z_0)$$
in the $\H$-uniformization. 

A stress tensor $A_{\bfs\beta}\equiv A_{\bfs\beta}^N$ for $N_{\bfs\beta}$ is given by 
$$A_{\bfs\beta} = A + (b \pa - j_{\bfs\beta})J,\qquad j_{\bfs\beta}:= \E\,J_{\bfs\beta},\qquad J_{\bfs\beta} \equiv J_{\bfs\beta}^N = \pa N_{\bfs\beta}.$$
As $\zeta\to z,$ in $\H,$ Ward's OPE for $N_{\bfs\beta}$ holds: 
$$A_{\bfs\beta}(\zeta)N_{\bfs\beta}(z,z_0) = \Big(A(\zeta) + \big(b \pa_\zeta - j_{\bfs\beta}(\zeta)\big)J(\zeta)\Big)N(z,z_0)\\
\sim \frac{J_{\bfs\beta}(z)}{\zeta-z} + \frac{b}{(\zeta-z)^2}.$$
In the symmetric case $\bfs\beta = \overbar{\bfs\beta_*},$ the Virasoro field 
$$T_{\bfs\beta} = -\frac12 J_{\bfs\beta}*J_{\bfs\beta} + b \pa J_{\bfs\beta}$$
gives rise to the Virasoro pair $(T_{\bfs\beta}, \overbar{T_{\bfs\beta}}).$ 
The central charge $c$ in this theory is given by 
$$c = 1 + 12b^2.$$
If the parameter $b$ is related to the SLE parameter $\kappa$ as 
\begin{equation} \label{eq: ab backward SLE kappa} 
b = -a (\kappa/4+1), \qquad a =\pm\sqrt{2/\kappa},
\end{equation}
then 
$$c = 13+12(\kappa/8+2/\kappa) \ge 25.$$

Modified multi-vertex fields or OPE exponentials $\OO_{\bfs\beta}[\bfs\tau]$ are defined by 
$$\OO_{\bfs\beta}[\bfs\tau] = \frac{C_{(-ib)}[-i(\bfs\tau+\bfs\beta)]} {C_{(-ib)}[-i\bfs\beta]} \, \ee^{\odot N[\bfs\tau]}.$$
We denote by $\FF^N_{\bfs\beta^+,\bfs\beta^-}$ the OPE family of $N_{\bfs\beta^+,\bfs\beta^-}.$ 

\begin{thm} \label{Neumann change of beta}
Given two double background charges $\bfs\beta^\pm, \bfs\beta_0^\pm$ with the neutrality conditions $(\NC_{b^\pm}),$ $\FF^N_{\bfs\beta^+,\bfs\beta^-}$ is the image of $\FF^N_{\bfs\beta_0^+,\bfs\beta_0^-}$ under the insertion of $\ee^{\odot N[\bfs\beta^+-\bfs\beta_0^+,\bfs\beta^--\bfs\beta_0^-]}.$
\end{thm}
\begin{proof} By Wick's calculus, we have 
$$\E\,N_{\bfs\beta_0^+,\bfs\beta_0^-}(z,z') \, \ee^{\odot N[\bfs\beta^+-\bfs\beta_0^+,\bfs\beta^--\bfs\beta_0^-]} = n_{\bfs\beta_0^+,\bfs\beta_0^-}(z,z') + \E\,N(z,z')N[\bfs\beta^+-\bfs\beta_0^+,\bfs\beta^--\bfs\beta_0^-].$$
In the $\H$-uniformization, it follows from \eqref{eq: NN+ NN-} and \eqref{eq: n_beta} that 
$$\E\,N(z,z')N[\bfs\beta^+-\bfs\beta_0^+,\bfs\beta^--\bfs\beta_0^-] = n_{\bfs\beta^+,\bfs\beta^-}(z,z')- n_{\bfs\beta_0^+,\bfs\beta_0^-}(z,z').$$
Since both sides are scalars, the above identity holds for any chart. 
Thus we have 
$$\E\,N_{\bfs\beta_0^+,\bfs\beta_0^-}(z,z') \, \ee^{\odot N[\bfs\beta^+-\bfs\beta_0^+,\bfs\beta^--\bfs\beta_0^-]} = \E\,N_{\bfs\beta^+,\bfs\beta^-}(z,z').$$
As in the Dirichlet case, we use Wick's calculus again to derive 
 $$\E\,N_{\bfs\beta_0^+,\bfs\beta_0^-}(z_1,z_1') \odot \cdots \odot N_{\bfs\beta_0^+,\bfs\beta_0^-}(z_n,z_n')\, \ee^{\odot N[\bfs\beta^+-\bfs\beta_0^+,\bfs\beta^--\bfs\beta_0^-]} = \E\,N_{\bfs\beta^+,\bfs\beta^-}(z_1,z_1')\odot \cdots \odot N_{\bfs\beta^+,\bfs\beta^-}(z_n,z_n').$$ 
Differentiating with respect to $z_j,z_j'$, we complete the proof. 
\end{proof}

\begin{eg*}
Given a marked boundary point $q\in\pa D,$ we consider a conformal map $w: (D,q) \to (\H,\infty).$
For $\bfs\beta^+ = 2b\cdot q$ and $\bfs\beta^- = \bfs0,$
$$N_{\bfs\beta}(z,z_0) = N(z,z_0) + 2b (\log |w'(z)| - \log |w'(z_0)| ).$$
The function $(z,z_0)\mapsto \log |w'(z)| - \log |w'(z_0)|$ does not depend on the choice of $w.$
\end{eg*}

\subsection{One-leg operators}
We consider a simply-connected domain $D$ with two marked points $p\in \pa D,$ $q\in \pa D$ in the chordal case, and $q\in D$ in the radial case. 
For a symmetric background charge $\bfs\beta$ on $S = D^{\mathrm{double}}$ with $\bfs\beta(p) = a$ and the neutrality condition $(\NC_b),$ the one-leg operator $\leg_p$ in the chordal case (the radial case, respectively) is defined by 
$$\leg_p \equiv \leg(p) = \OO_{\check{\bfs\beta\,}}[\bfs\tau], \qquad \check{\bfs\beta\,} = \bfs\beta - \bfs\tau$$
where $\bfs\tau= a\cdot p - a\cdot q$ in the chordal case (and $\bfs\tau= a\cdot p - \frac12a\cdot q- \frac12a\cdot q^*$ in the radial case, respectively). 
The insertion of $\leg_p /\E\,\leg_p$ is an operator 
$$\XX \mapsto \wh\XX$$
on Fock space functionals/fields by the rules \eqref{eq: BCrules}
and the formula in the chordal case
$$\wh N_{\check{\bfs\beta\,}}(z,z_0)  = N_{\check{\bfs\beta\,}}(z,z_0) -2a \big(\log|w(z)|-\log|w(z_0)| \big)\big)$$
where $w:(D,p,q)\to(\H,0,\infty)$ is a conformal map.
In the radial case, the corresponding formula is given by
 $$\wh N_{\check{\bfs\beta\,}}(z,z_0) = N_{\check{\bfs\beta\,}}(z,z_0) -2a \big(\log|1-w(z)|-\log|1-w(z_0)| \big)\big)+a \big(\log|w(z)|-\log|w(z_0)| \big)\big)$$
where $w:(D,p,q)\to(\D,1,0)$ is a conformal map.

We denote 
$$\wh\E\, \XX := \frac{\E\,\leg_p \XX}{\E\,\leg_p}.$$
By Theorem~\ref{Neumann change of beta}, we have 
$$\wh\E\,\XX = \E\,\wh\XX $$
for any string $\XX$ of fields in $\FF^N_{\check{\bfs\beta\,}}.$

We now show the level two degeneracy equations for $\leg$ if the parameters $a$ and $b$ satisfy $2a(a+b)= -1.$ 
We remark that $a$ and $b$ in \eqref{eq: ab backward SLE kappa} have such a relation.

\begin{prop} \label{level2Neumann}
Provided that $2a(a+b)=-1,$ we have 
\begin{equation} \label{eq: level2Neumann}
T_{\check{\bfs\beta\,}}*\leg = -\frac1{2a^2}\pa^2 \leg.
\end{equation}
\end{prop}
\begin{proof}
Let $\tilde a = -ia, \tilde b= -ib.$ 
In terms of the action of Virasoro generators and current generators, the one-leg operators are Virasoro primary holomorphic fields of conformal dimension $\lambda = \frac12 \tilde a^2 - \tilde a \tilde b$ and current primary with charge $q = \tilde a.$ 
This implies the level two degeneracy equation \eqref{eq: level2Neumann} for the one-leg operator provided that $2\tilde a(\tilde a+\tilde b)= 1.$ 

We here present an alternate but direct proof using Wick's calculus. 
For simplicity, we consider the standard chordal case $\check{\bfs\beta\,} =2b\cdot q$ only. 
It is left to the reader as an exercise to prove \eqref{eq: level2Neumann} for general background charges. 
Since the difference is a differential, it suffices to show \eqref{eq: level2Neumann} in $\H.$
In the $\H$-uniformization we have 
$$N_{\check{\bfs\beta\,}} = N, \quad \pa N_{\check{\bfs\beta\,}} = J_{\check{\bfs\beta\,}} = J, \quad T_{\check{\bfs\beta\,}} = T + b\pa J, \quad T = -\frac12 J\odot J, \quad \leg(z) = \ee^{\odot aN^+(z,q)}.$$
Let us first compute $T*\leg:$ 
$$T*\leg = T\odot \leg +a \pa J\odot \leg.$$
Indeed, as $\zeta\to z,$ we have 
$$T(\zeta)\leg(z) = -\frac12 \sum_{n=0}^\infty \frac {a^n}{n!} J^{\odot2}(\zeta) (N^+)^{\odot n}(z,q) = T\odot \leg(z) + \mathrm{I} + \mathrm{II} + o(1),$$
where terms $\mathrm{I}$ and $\mathrm{II}$ come from 1 and 2 contractions, respectively. 
It follows from 
$\E\, J(\zeta)N^+(z,q) = -1/(\zeta-z)$ that 
$$\mathrm{I} = \frac{a}{\zeta-z}J(\zeta)\odot N^+(z,q).$$
Thus its contribution to $T*\leg$ is $a \pa J\odot \leg.$ 
On the other hand, we have
$$\mathrm{II} = -\frac12 a^2 \big(\E\,J(\zeta)N^+(z,q)\big)^2\leg(z) = -\frac12 \frac{a^2}{(\zeta-z)^2}\leg(z).$$
It has no contribution to $T*\leg.$
To compute $T_{\check{\bfs\beta\,}}*\leg,$ we need to compute $\pa J*\leg.$ 
We have 
$$\pa J(\zeta)\leg(z) = \pa J\odot \leg(z) + o(1).$$
This implies that $\pa J*\leg = \pa J\odot \leg$ and 
$$T_{\check{\bfs\beta\,}}*\leg = T\odot \leg +(a+b) \pa J\odot \leg.$$
Differentiating $\leg(z) = \ee^{\odot aN^+(z,q)},$ we have
$$-\frac1{2a^2} \pa^2 \leg = -\frac1{2a^2}\pa (a J\odot \leg) = -\frac1{2a^2} (a \pa J\odot \leg + a^2 J\odot J\odot \leg) = -\frac1{2a} \pa J \odot \leg + T\odot\leg.$$
Now \eqref{eq: level2Neumann} follows provided that $2a(a+b)=-1.$
\end{proof}

\subsection{Ward's equations and BPZ-Cardy equations}
We define the puncture operators by $\PP_{\bfs\beta} := C_{(-ib)}[-i\bfs\beta].$ 
As in the Dirichlet case, Ward's equations hold for the extended OPE family $\FF^N_{\bfs\beta}$ of $N_{\bfs\beta}.$
Since its proof is similar to the Dirichlet case, we leave it to the reader as an exercise. 

\begin{thm}[Ward's equations] \label{Ward for Neumann}
Let $Y, X_1,\cdots, X_n\in \FF^N_{\bfs\beta}$ and let $X$ be the tensor product of $X_j$'s.
Then
$$\E\, T_{\bfs\beta}(z)\,X=\E\,\PP_{\bfs\beta}^{-1} \LL^+_{k_z} \PP_{\bfs\beta} X+\E\,\PP_{\bfs\beta}^{-1} \LL^-_{k_{\bar z}}\PP_{\bfs\beta} X,$$
where all fields are evaluated in the identity chart of $\H$ and 
$$2z^2\E\, T_{\bfs\beta}(z)\,X=\E\,\PP_{\bfs\beta}^{-1} \LL^+_{v_z} \PP_{\bfs\beta} X+\E\,\PP_{\bfs\beta}^{-1} \LL^-_{v_{z^*}}\PP_{\bfs\beta} X,$$
where all fields are evaluated in the identity chart of $\D.$ 
\end{thm}

For a symmetric background charge $\bfs\beta$ with the neutrality condition $(\NC_b)$ and a specific charge $a$ at a marked boundary point $p\in\pa D,$ we define the backward SLE partition function $Z_{\bfs\beta} $ by 
$$Z_{\bfs\beta} := |\PP_{\bfs\beta}|=|C_{(-ib)}[-i\bfs\beta]|$$ 
and the effective one-leg operators by $\leg_p^\eff:= \PP_{\bfs\beta}\,\leg_p.$ 
Denote 
\begin{equation} \label{eq: R_xi Neumann}
R_\xi \equiv \wh\E_\xi\,\XX : = \frac{\E\,\leg_\xi \XX}{\E\,\leg_\xi}= \frac{\E\,\leg_\xi^\eff \XX}{\E\,\leg_\xi^\eff}.
\end{equation}

Let $\xi\in\pa D$ and $\bfs\beta_\xi = \bfs\beta - a\cdot p + a\cdot \xi.$
Let $X$ be the tensor product $X= X_1(z_1)\cdots X_n(z_n)$ of fields $X_j$ in $\FF_{\check{\bfs\beta\,}}$ $(z_j\in D).$
Combining Proposition~\ref{level2Neumann} with Theorem~\ref{Ward for Neumann}, we obtain the BPZ equations: 
\begin{equation} \label{eq: BPZ Neumann} 
-\frac1{2a^2} \pa_\xi^2 \E\leg_\xi^\eff X = \E\, \check{\LL}_{k_\xi} \leg_\xi^\eff X, \qquad k_\xi(z):=\frac1{\xi-z}
\end{equation}
in the $\H$-uniformization.
In particular, the partition function $Z_\xi\equiv Z_{\bfs\beta_\xi} := \big|\E\,\leg^\eff_\xi \big|= \big|C_{(-ib)}[-i\bfs\beta_\xi]\big|$ satisfies the null vector equation
\begin{equation} \label{eq: Neumann NV}
-\frac1{2a^2} \pa_\xi^2 Z_\xi = \check{\LL}_{k_\xi} Z_\xi
\end{equation}
in the $\H$-uniformization.
Here $\pa_\xi = \pa + \bar\pa$ is the operator of differentiation with respect to the real variable $\xi$ and and $\check{\LL}_{k_\xi}$ is taken over the finite notes of $X$ and $\supp\,\bfs\beta_\xi\setminus\{\xi\}.$
We obtain the following form of BPZ-Cardy equations in the chordal case. 
The radial case can be derived in a similar way. 

\begin{thm}[BPZ-Cardy equations] \label{BPZ Neumann}
Suppose that the parameters $a$ and $b$ satisfy 
$$2a(a+b) = -1.$$
If $X = X_1(z_1)\cdots X_n(z_n)$ with $X_j\in\FF^N_{\check{\bfs\beta\,}},$ then we have 
\begin{equation} \label{eq: BPZ-Cardy Neumann chordal}
-\frac1{2a^2}\big(\pa_\xi^2\wh\E_{\xi}\,X + 2(\pa_\xi \log Z_\xi) \pa_\xi\wh\E_{\xi}\,X \big) = \check \LL_{k_\xi}\wh\E_{\xi}\,X
\end{equation}
in the identity chart of $\H.$ 
In the $(\D,0)$-uniformization we have
\begin{equation} \label{eq: BPZ-Cardy Neumann radial}
\kappa\Big(\frac12\pa_\theta^2 + (\pa_\theta \log Z_\zeta\big)\pa_\theta\Big)\wh\E_\zeta X = \check{\LL}_{v_\zeta}\wh\E_\zeta X, \qquad \zeta = \ee^{i\theta}.
\end{equation}
\end{thm}

\subsection{Connection to backward SLE theory} \label{ss: Neumann and backward} 
We now prove Theorem~\ref{main: Neumann SLE} and present Sheffield's observables.
By definition, the backward radial (see \eqref{eq: f chordal} for the chordal case) $\SLE[\bfs\beta]$ map $f_t$ from $\D$ satisfies the equation 
\begin{equation} \label{eq: f radial}
\partial_t f_t(z) = -f_t(z)\frac{\zeta_t+f_t(z)}{\zeta_t-f_t(z)}, \qquad \zeta_t = \ee^{i\theta_t}
\end{equation}
driven by the real process $\theta_t:$
$$\dd\theta_t = \sqrt\kappa\, \dd B_t + \lambda(t)\,\dd t, \quad \lambda(t) = (\lambda\,\|\,f_t^{-1}), \quad \lambda = \kappa \, \pa_\theta \log Z_{\bfs\beta_\zeta},$$
where the partition function $Z_{\bfs\beta_\zeta}$ is given by $Z_{\bfs\beta_\zeta} = \big|C_{(-ib)}[-i\bfs\beta_\zeta]\big|,$ $\bfs\beta_\zeta= \bfs\beta + a\cdot\zeta-a\cdot p.$ 

\begin{proof}[Proof of Theorem~\ref{main: Neumann SLE}]
We first consider the chordal case.
For $\XX = X_1(z_1)\cdots X_n(z_n), X_j\in\FF_{\check{\bfs\beta\,}}(D),$ denote 
$$R_\xi (z_1,\cdots, z_n)\equiv\wh\E_\xi [X_1(z_1)\cdots X_n(z_n)].$$
We want to show that 
$$M_t = m(\xi_t,t), \quad m(\xi,t) = (R_\xi \,\|\, f_t^{-1})$$
is a local martingale on backward chordal SLE probability space. 
By It\^o's formula, we compute $\dd M_t$ as follows: 
\begin{align*}
\dd M_t & = \sqrt\kappa\, \pa_\xi\big|_{\xi =\xi_t} m(\xi,t)\,\dd B_t + \kappa\, \pa_\xi\big|_{\xi =\xi_t} m(\xi,t) \,\pa_\xi\big|_{\xi =\xi_t} (\log Z_\xi \,\|\, f_t^{-1}) \\
&+\frac\kappa2 \, \pa_\xi^2\big|_{\xi =\xi_t} m(\xi,t)\,\dd t + \frac {\dd}{\dd s}\Big|_{s=0} \big(R_{\xi_t}\,\|\,f_{t+s}^{-1}\big)\,\dd t.
\end{align*}
Using the same argument in the proof of Theorem~\ref{main}, the last term can be rewritten as 
$$\frac {\dd}{\dd s}\Big|_{s=0} \big(R_{\xi_t}\,\|\,f_{t+s}^{-1}\big)\,\dd t= 2\big(\check\LL_{k_{\xi_t}}R_{\xi_t}\,\|\,f_{t}^{-1}\big).$$
Thus we find the drift term of $\dd M_t$ as 
$$\Big(\frac\kappa2 \, \pa_\xi^2\big|_{\xi =\xi_t} m(\xi,t) +
\kappa\, \pa_\xi\big|_{\xi =\xi_t} m(\xi,t) \frac{\pa_\xi|_{\xi=\xi_t} (Z_\xi\,\|\,f_t^{-1})}{(Z_\xi\,\|\,f_t^{-1})} +2\big(\check{\LL}_{k_{\xi_t}}R_{\xi_t}\,\|\,f_{t}^{-1}\big)\Big)\,\dd t = 0
$$
employing the BPZ-Cardy equations~\eqref{eq: BPZ-Cardy Neumann chordal}. 

Next, we consider the radial case. 
For $\XX = X_1(z_1)\cdots X_n(z_n), X_j\in\FF_{\check{\bfs\beta\,}}(D),$ denote 
$$R_\zeta (z_1,\cdots, z_n)\equiv\wh\E_\zeta [X_1(z_1)\cdots X_n(z_n)].$$
As in the chordal case, the process 
$$M_t = m(\zeta_t,t), \qquad m(\zeta,t) =\big(R_\zeta \,\|\,f_t^{-1}\big)$$
is a local martingale by It\^o's formula and the BPZ-Cardy equations~\eqref{eq: BPZ-Cardy Neumann radial}. 
Indeed, we find the drift term of $\dd M_t$ as 
$$\Big(\frac\kappa2\pa_\theta^2+\kappa\frac{\pa_\theta (Z_\zeta\,\|\,f_t^{-1})}{(Z_\zeta\,\|\,f_t^{-1})}\pa_\theta\Big)\Big|_{\zeta=\zeta_t}~m(\zeta,t)\,\dd t\\
-\ \big(\check\LL_{v_{\zeta_t}}R_{\zeta_t}\,\|\,f_{t}^{-1}\big) \dd t =0.$$
\end{proof}

\begin{eg*}[Sheffield's observables] In the standard chordal case with 
$\bfs\beta =a \cdot p + (2b-a)\cdot q,$
the 1-point functions of the bosonic fields
\begin{align*}
n_t(z) &= \E[N_{\bfs\beta}(z)\,\|\,f_t^{-1}]= -2a \log|f_t(z)-\xi_t| +2b\log|f_t'(z)|\\
&= -2a \log|z| + 2 \sqrt2\int_0^t\Re\, \frac1{f_s(z)-\xi_s}\,\dd B_s
\end{align*}
were introduced as backward SLE martingale-observables in \cite{Sheffield16}.
Due to the following special case of Hadamard's variation formula
\begin{equation} \label{eq: Hadamard Neumann}
\dd G^N(f_t(z_1),f_t(z_2))=-4\,\Re\,\frac1{f_t(z_1)-\xi_t}\,\Re\,\frac1{f_t(z_2)-\xi_t}\,\dd t=-\frac12\,\dd \langle n_{\bfs\beta}(z_1),n_{\bfs\beta}(z_2)\rangle_t,
\end{equation}
the formal $2$-point functions 
$$\E[N_{\bfs\beta}(z_1)N_{\bfs\beta}(z_2)]=2G^N(z_1,z_2)+n_{\bfs\beta}(z_1)n_{\bfs\beta}(z_2)$$
are martingale-observables for backward SLE. 
Sheffield used \eqref{eq: Hadamard Neumann} to construct a coupling of backward SLE and the Gaussian free field with Neumann boundary condition, see \cite{Sheffield16}.
\end{eg*}

\subsection*{Acknowledgements}
We would like to thank Tom Alberts, and Sung-Soo Byun for careful reading and much-appreciated help improving this manuscript. 

\bibliographystyle{plain}

\begin{thebibliography}{10}

\bibitem{AKM}
Tom Alberts, Nam-Gyu Kang, and Nikolai~G. Makarov.
\newblock Conformal field theory for multiple {SLE}s.
\newblock in preparation.

\bibitem{AKM20}
Tom Alberts, Nam-Gyu Kang, and Nikolai~G. Makarov.
\newblock Pole dynamics and an integral of motion for multiple {SLE}(0).
\newblock 2020.
\newblock arXiv:2011.05714.

\bibitem{AKL12}
Tom Alberts, Michael~J. Kozdron, and Gregory~F. Lawler.
\newblock The {G}reen function for the radial {S}chramm-{L}oewner evolution.
\newblock {\em J. Phys. A}, 45(49):494015, 17, 2012.

\bibitem{BB03}
Michel Bauer and Denis Bernard.
\newblock Conformal field theories of stochastic {L}oewner evolutions.
\newblock {\em Comm. Math. Phys.}, 239(3):493--521, 2003.
\newblock \arXiv{hep-th/0210015}.

\bibitem{BB04}
Michel Bauer and Denis Bernard.
\newblock C{FT}s of {SLE}s: the radial case.
\newblock {\em Phys. Lett. B}, 583(3-4):324--330, 2004.

\bibitem{Beffara08}
Vincent Beffara.
\newblock The dimension of the {SLE} curves.
\newblock {\em Ann. Probab.}, 36(4):1421--1452, 2008.

\bibitem{BPZ84}
A.~A. Belavin, A.~M. Polyakov, and A.~B. Zamolodchikov.
\newblock Infinite conformal symmetry in two-dimensional quantum field theory.
\newblock {\em Nuclear Phys. B}, 241(2):333--380, 1984.

\bibitem{BKT}
Sung-Soo Byun, Nam-Gyu Kang, and Hee-Joon Tak.
\newblock Conformal field theory for annulus {SLE}: partition functions and
  martingale-observables.
\newblock 2018.
\newblock arXiv:1806.03638v2.

\bibitem{Cardy04}
John Cardy.
\newblock Calogero-{S}utherland model and bulk-boundary correlations in
  conformal field theory.
\newblock {\em Phys. Lett. B}, 582(1-2):121--126, 2004.

\bibitem{Cardy04b}
John Cardy.
\newblock {SLE}(kappa,rho) and {C}onformal {F}ield {T}heory.
\newblock 2004.
\newblock arXiv:math-ph/0412033.

\bibitem{DFMS97}
Philippe Di~Francesco, Pierre Mathieu, and David S{\'e}n{\'e}chal.
\newblock {\em Conformal field theory}.
\newblock Graduate Texts in Contemporary Physics. Springer-Verlag, New York,
  1997.

\bibitem{CD07}
B.~Doyon and J.~Cardy.
\newblock Calogero-{S}utherland eigenfunctions with mixed boundary conditions
  and conformal field theory correlators.
\newblock {\em J. Phys. A}, 40(10):2509--2540, 2007.

\bibitem{Dubedat05}
Julien Dub\'edat.
\newblock {${\rm SLE}(\kappa,\rho)$} martingales and duality.
\newblock {\em Ann. Probab.}, 33(1):223--243, 2005.

\bibitem{Dubedat09}
Julien Dub{\'e}dat.
\newblock S{LE} and the free field: partition functions and couplings.
\newblock {\em J. Amer. Math. Soc.}, 22(4):995--1054, 2009.

\bibitem{FW03}
Roland Friedrich and Wendelin Werner.
\newblock Conformal restriction, highest-weight representations and {SLE}.
\newblock {\em Comm. Math. Phys.}, 243(1):105--122, 2003.

\bibitem{KR87}
V.~G. Kac and A.~K. Raina.
\newblock {\em Bombay lectures on highest weight representations of
  infinite-dimensional {L}ie algebras}, volume~2 of {\em Advanced Series in
  Mathematical Physics}.
\newblock World Scientific Publishing Co. Inc., Teaneck, NJ, 1987.

\bibitem{KM17}
Nam-Gyu Kang and Nikolai~G. Makarov.
\newblock Calculus of conformal fields on a compact {R}iemann surface.
\newblock arXiv:1708.07361.

\bibitem{KM13}
Nam-Gyu Kang and Nikolai~G. Makarov.
\newblock Gaussian free field and conformal field theory.
\newblock {\em Ast\'erisque}, (353):viii+136, 2013.

\bibitem{KT13}
Nam-Gyu Kang and Hee-Joon Tak.
\newblock Conformal field theory of dipolar {SLE} with the {D}irichlet boundary
  condition.
\newblock {\em Anal. Math. Phys.}, 3(4):333--373, 2013.

\bibitem{Kytola06}
Kalle Kyt\"ol\"a.
\newblock On conformal field theory of {${\rm SLE}(\kappa,\rho)$}.
\newblock {\em J. Stat. Phys.}, 123(6):1169--1181, 2006.

\bibitem{Lawler05}
Gregory~F. Lawler.
\newblock {\em Conformally invariant processes in the plane}, volume 114 of
  {\em Mathematical Surveys and Monographs}.
\newblock American Mathematical Society, Providence, RI, 2005.

\bibitem{Lawler09}
Gregory~F. Lawler.
\newblock Partition functions, loop measure, and versions of {SLE}.
\newblock {\em J. Stat. Phys.}, 134(5-6):813--837, 2009.

\bibitem{LSW01a}
Gregory~F. Lawler, Oded Schramm, and Wendelin Werner.
\newblock The dimension of the planar {B}rownian frontier is {$4/3$}.
\newblock {\em Math. Res. Lett.}, 8(4):401--411, 2001.

\bibitem{LSW01c}
Gregory~F. Lawler, Oded Schramm, and Wendelin Werner.
\newblock Values of {B}rownian intersection exponents. {II}. {P}lane exponents.
\newblock {\em Acta Math.}, 187(2):275--308, 2001.

\bibitem{LSW04}
Gregory~F. Lawler, Oded Schramm, and Wendelin Werner.
\newblock Conformal invariance of planar loop-erased random walks and uniform
  spanning trees.
\newblock {\em Ann. Probab.}, 32(1B):939--995, 2004.

\bibitem{MS16}
Jason Miller and Scott Sheffield.
\newblock Imaginary geometry {I}: interacting {SLE}s.
\newblock {\em Probab. Theory Related Fields}, 164(3-4):553--705, 2016.

\bibitem{MS17}
Jason Miller and Scott Sheffield.
\newblock Imaginary geometry {IV}: interior rays, whole-plane reversibility,
  and space-filling trees.
\newblock {\em Probab. Theory Related Fields}, 169(3-4):729--869, 2017.

\bibitem{PW20}
Eveliina Peltola and Yilin Wang.
\newblock Large deviations of multichordal {SLE}$_{0+}$, real rational
  functions, and zeta-regularized determinants of {L}aplacians.
\newblock arXiv:2006.08574, to appear in J. Eur. Math. Soc.

\bibitem{RBGW07}
I.~Rushkin, E.~Bettelheim, I.~A. Gruzberg, and P.~Wiegmann.
\newblock Critical curves in conformally invariant statistical systems.
\newblock {\em J. Phys. A}, 40(9):2165--2195, 2007.

\bibitem{Schramm00}
Oded Schramm.
\newblock Scaling limits of loop-erased random walks and uniform spanning
  trees.
\newblock {\em Israel J. Math.}, 118:221--288, 2000.

\bibitem{SS09}
Oded Schramm and Scott Sheffield.
\newblock Contour lines of the two-dimensional discrete {G}aussian free field.
\newblock {\em Acta Math.}, 202(1):21--137, 2009.

\bibitem{SS13}
Oded Schramm and Scott Sheffield.
\newblock A contour line of the continuum {G}aussian free field.
\newblock {\em Probab. Theory Related Fields}, 157(1-2):47--80, 2013.

\bibitem{SW05}
Oded Schramm and David~B. Wilson.
\newblock S{LE} coordinate changes.
\newblock {\em New York J. Math.}, 11:659--669 (electronic), 2005.

\bibitem{Sheffield16}
Scott Sheffield.
\newblock Conformal weldings of random surfaces: {SLE} and the quantum gravity
  zipper.
\newblock {\em Ann. Probab.}, 44(5):3474--3545, 2016.

\bibitem{Smirnov01}
Stanislav Smirnov.
\newblock Critical percolation in the plane: conformal invariance, {C}ardy's
  formula, scaling limits.
\newblock {\em C. R. Acad. Sci. Paris S\'er. I Math.}, 333(3):239--244, 2001.

\bibitem{Smirnov10}
Stanislav Smirnov.
\newblock Conformal invariance in random cluster models. {I}. {H}olomorphic
  fermions in the {I}sing model.
\newblock {\em Ann. of Math. (2)}, 172(2):1435--1467, 2010.

\bibitem{Zhan08}
Dapeng Zhan.
\newblock Duality of chordal {SLE}.
\newblock {\em Invent. Math.}, 174(2):309--353, 2008.

\end{thebibliography}

\end{document}